\documentclass[11pt,a4paper]{article}

\usepackage[T1]{fontenc}
\usepackage[utf8]{inputenc}
\usepackage{amsmath,amssymb,amsthm,amsfonts}
\usepackage[margin=0.83in]{geometry} 
\usepackage{thmtools}
\usepackage{comment}
\usepackage{todonotes}

\usepackage{graphicx} %for figures
\usepackage{float} %for figures

\usepackage[math]{blindtext}

\usepackage{graphicx}
\usepackage{dsfont}
\usepackage{units}
\usepackage{enumitem}
\usepackage{booktabs}
\usepackage{color}

\newcommand{\MH}[1]{{\color{red}{#1}}}

\usepackage{tikz}			% to make graphics
\usetikzlibrary{arrows,automata,shapes, trees}

\theoremstyle{plain}
\newtheorem{theorem}{Theorem}[section]
\newtheorem{proposition}[theorem]{Proposition}

\newtheorem{corollary}[theorem]{Corollary}
\theoremstyle{definition}
\newtheorem{definition}[theorem]{Definition}
\newtheorem{remark}[theorem]{Remark}
\newtheorem{example}[theorem]{Example}

\theoremstyle{remark}
\numberwithin{equation}{section}
\newtheorem*{empty*}{}

\renewcommand\thmcontinues[1]{\textbf{Continued}}

\DeclareMathOperator*{\argmin}{argmin}

\DeclareMathOperator*{\Var}{Var}

\DeclareMathOperator*{\esssup}{ess\,sup}
\DeclareMathOperator*{\essinf}{ess\,inf}

\newcommand{\e}{\mathrm{e}}
\renewcommand{\theta}{\vartheta}
\renewcommand{\epsilon}{\varepsilon}
\renewcommand{\P}{\mathbb{P}}

\newcommand{\E}{\mathbb{E}}

\newcommand{\NN}{\mathbb{N}}

\newcommand{\RR}{\mathbb{R}}

\newcommand{\cA}{\mathcal{A}}

\newcommand{\cD}{\mathcal{D}}
\newcommand{\cE}{\mathcal{E}}
\newcommand{\cF}{\mathcal{F}}
\newcommand{\cG}{\mathcal{G}}

\newcommand{\cM}{\mathcal{M}}
\newcommand{\cN}{\mathcal{N}}
\newcommand{\cO}{\mathcal{O}}
\newcommand{\cP}{\mathcal{P}}
\newcommand{\cQ}{\mathcal{Q}}
\newcommand{\cS}{\mathcal{S}}
\newcommand{\cX}{\mathcal{X}}

\newcommand{\dd}{\,\mathrm{d}}

\newcommand{\1}{\mathbf{1}}
\newcommand{\0}{\mathbf{0}}

\newcommand*{\as}[1]{#1\text{-a.s.}}

\newcommand*{\ol}[1]{\bar{#1}}

\newcommand{\ES}{\mathrm{ES}}

\newcommand{\VaR}{\mathrm{VaR}}
\newcommand{\SR}{\mathrm{SR}}
\newcommand{\OCE}{\mathrm{OCE}}
\newcommand{\WC}{\mathrm{WC}}

\newcommand{\fE}{\mathfrak{E}}

\begin{document}
	
	\title{$\rho$-arbitrage and $\rho$-consistent pricing for star-shaped \\ risk measures\thanks{We would like to thank John Armstrong, Matteo Burzoni, Cosimo Munari and Ruodu Wang as well as the participants of the Seminar on Risk Management and Actuarial Science at the University of Waterloo for very helpful comments and discussions. We also thank two anonymous referees for their pertinent comments.}}
	
\author{Martin Herdegen\thanks{University of Warwick, Department of Statistics, Coventry, CV4 7AL, UK, email \texttt{m.herdegen@warwick.ac.uk}.}
	\and 
Nazem Khan\thanks{Dublin City University, School of Mathematics, Dublin, D09 FW22, Ireland, email \texttt{nazem.khan@dcu.ie}.}
	}

	\date{\today}
	
	\maketitle

\begin{abstract}
This paper revisits mean-risk portfolio selection in a one-period financial market, where risk is quantified by a star-shaped risk measure $\rho$. We make three contributions. First, we introduce the new axiom of sensitivity to large expected losses and show that it is key to ensure the existence of optimal portfolios.  Second, we give primal and dual characterisations of (strong) $\rho$-arbitrage.  Finally, we use our conditions for the absence of (strong) $\rho$-arbitrage to explicitly derive the (strong) $\rho$-consistent price interval for an external financial contract.  %This comes down to finding the correct subset of  equivalent martingale measures in which to take discounted expectations.
\end{abstract}

\bigskip
\noindent\textbf{Mathematics Subject Classification (2020):} 91G10, 90C46, 46N10 
%91G10 Porftolio Theory
%90C46 Optimality conditions, duality 
%Applications of functional analysis in optimization, convex analysis, mathematical programming, economics.

%\bigskip
%\noindent\textbf{JEL Classification:}  G11, D81, C61
%G11 Portfolio Choice
%D81 Criteria for Decision-Making under Risk and Uncertainty
%C61 Optimization Techniques 

\bigskip
\noindent \textbf{Area of review:} Stochastic Models

\bigskip
\noindent\textbf{Keywords:} portfolio selection, $\rho$-arbitrage, convex risk measures, star-shaped risk measures, dual characterisation, good-deals, $\rho$-consistent pricing

\section{Introduction}
The aim of a risk measure is to quantify the risk of a financial position by a single number.  This number can be interpreted in different ways in different contexts, as summarised by Wang \cite[p.~337]{wang:16}: ``In banking, it represents the capital requirement to regulate a risk; in insurance, it calculates the premium for an insurance contract; and in economics it ranks the preference of a risk for a market participant.''  

In this paper, we interpret the risk measure $\rho$ as a \emph{regulatory constraint} imposed by the regulator on a financial agent seeking to optimise a portfolio. Because our focus is on the \emph{effectiveness} of the risk constraint, we ignore any idiosyncratic risk aversion of the agent.  Denoting by $X_\pi$ the excess return of a portfolio $\pi \in \RR^d$, we consider the following two problems:
\begin{enumerate}[label=({\arabic*})] 
\item Given a minimal desired expected excess return $\nu^* \in [0,\infty)$, minimise the risk $\rho(X_\pi)$ among all portfolios $\pi \in \RR^d$ that satisfy $\mathbb{E}[X_{\pi}] \geq \nu^*$;
\item Given a maximal risk threshold $\rho^* \in [0,\infty)$, maximise the return $\mathbb{E}[X_\pi]$ among all portfolios $\pi \in \RR^d$ that satisfy $\rho(X_\pi) \leq \rho^*$.
\end{enumerate}
We refer to either problem as mean-$\rho$ portfolio selection in the sequel.  The way to tackle these two problems is to first study problem (1) with an equality constraint:
\begin{enumerate} 
 \item [(1')] Given $\nu \geq 0$, minimise the risk $\rho(X_\pi)$ among all portfolios $\pi \in \RR^d$ that satisfy $\mathbb{E}[X_{\pi}] = \nu$.
    %\item [(1')] Given $\nu \geq 0$, minimise $\rho(X_\pi)$ on $\Pi_\nu$ -- where $\Pi_\nu$ denotes all portfolios $\pi \in \RR^d$ with $\mathbb{E}[X_\pi] = \nu$ -- and find the corresponding minimal risk $\rho_\nu:=\inf\{ \rho(X_\pi) : \pi \in \Pi_\nu \}$.
\end{enumerate}
Portfolios minimising (1') are referred to as \emph{$\rho$-optimal} and the corresponding minimal risk is denoted by $\rho_\nu:=\inf\{ \rho(X_\pi) : \pi \in \Pi_\nu \}$, where $\Pi_\nu$ denotes all portfolios $\pi \in \RR^d$ with $\mathbb{E}[X_\pi] = \nu$.  One can then study the $\rho$-optimal boundary $\mathcal{O}_\rho:=\{(\rho_\nu,\nu) : \nu \geq 0\}$ and use this to find the solutions to (1) and (2) -- in case that they are well posed.

\medskip
There is a large literature on mean-$\rho$ portfolio selection.  \emph{Deviation risk measures}, which are a generalisation of the standard deviation, have been considered by de Giorgi \cite{de2005reward} and Rockafellar et al.~\cite{MasterFundsRockafellar}. As one would expect, the results in this case share a lot of similarities with the classical mean-variance framework of Markowitz \cite{Markowitz1952}. However, deviation risk measures quantify the degree of \emph{uncertainty} in a random variable, while regulators are more concerned with the overall seriousness of possible \emph{losses}. In particular, deviation risk measures are not monotone (or cash-invariant).

Presently, the most popular risk measures are normalised, monotone, cash-invariant and \emph{positively homogeneous}, i.e., $\rho(\lambda X) = \lambda \rho(X)$ for $\lambda \in (0,\infty)$, with Value at Risk (VaR) and Expected Shortfall (ES) being the most famous examples.\footnote{Among many other papers, mean-VaR portfolio selection has been studied by Alexander and Baptista \cite{alexander2002economic} and Gaivoronski and Pflug \cite{gaivoronski2005value} and mean-ES portfolio selection has been studied by Rockafellar and Uryasev \cite{rockafellar2002conditional} and Embrechts et al.~\cite{embrechts2021robustness}.  For an axiomatic justification of mean-ES portfolio selection, we refer to Han et al.~\cite{han:al:21}.}  However, in this setting, $\rho$-optimal portfolios need not exist.  And even if they do exist, mean-$\rho$ portfolio selection may be \emph{ill-posed} in the sense: there is a sequence of portfolios $(\pi_n)_{n \in \mathbb{N}}$ such that $\mathbb{E}[X_{\pi_n}] \uparrow \infty$ and $\rho(X_{\pi_n}) \leq 0$ for all $n$; or even worse (from the perspective of a regulator), there is a sequence of portfolios $(\pi_n)_{n \in \mathbb{N}}$ such that $\mathbb{E}[X_{\pi_n}] \uparrow \infty$ and $\rho(X_{\pi_n}) \downarrow -\infty$. These phenomena are referred to as \emph{$\rho$-arbitrage} and \emph{strong $\rho$-arbitrage}, respectively. %\footnote{A more precise name would be \emph{(strong) (mean-$\rho$)-arbitrage}, but we drop the qualifier ``mean''; cf.~Remark~\ref{rmk:(strong) rho arb definition}(b).}  
They are undesirable from a regulatory point of view since positions that are highly rewarding should come with a risk -- and the risk is being nullified here. Both concepts are generalisations of the classical notions of arbitrage of the first and second kind. Indeed if $\rho$ is the worst-case risk measure, then $\rho$-arbitrage is classical arbitrage of the first kind and strong $\rho$-arbitrage is classical arbitrage of the second kind; see \cite[Proposition 3.22]{herdegen2020dual} for details. For coherent risk measures, the concepts of (strong) $\rho$-arbitrage have been studied from a theoretical perspective in Herdegen and Khan \cite{herdegen2020dual}; the practical relevance of $\rho$-arbitrage has been discussed by Armstrong and Brigo in \cite{armstrong2019risk} and \cite{armstrong2019statistical}. 

Positive homogeneity is a strong property, which has been questioned on economic grounds early on.  It triggered the introduction of \emph{convex} risk measures by F{\"o}llmer and Schied \cite{follmer2002convex} and Fritelli and Gianin \cite{frittelli2002putting}. It is easy to check that convexity together with normalisation implies that the risk measure is \emph{star-shaped}, i.e., $\rho(\lambda X) \geq \lambda \rho(X)$ for $\lambda \in [1,\infty)$.\footnote{Star-shaped risk measures have recently been studied by Castagnoli et al.~\cite{castagnoli2021star} in a setting where there is no underlying probability measure.} This encapsulates the idea that financial positions become more risky when there is more at stake, for which there is empirical evidence, cf.~\cite{bosch2006reflections,bosch2010averting}.

\medskip
The first objective of this paper is to study mean-$\rho$ portfolio selection when $\rho$ is a \emph{risk functional}, i.e., star-shaped, monotone and normalised.  For some of our results, in particular, for our dual characterisations, we assume in addition that $\rho$ is cash-invariant (cash-invariant risk functionals are referred to as \emph{risk measures}), convex or satisfies the Fatou property. Assuming that $\rho$ lives on some Riesz space $L^\infty \subset L \subset L^1$ and is $(-\infty, \infty]$-valued, we first seek to answer the following three questions for a given market:
\begin{enumerate}
	\item[(Q1)] \textbf{Existence of optimal portfolios.}~What conditions guarantee that \emph{$\rho$-optimal portfolios} for a desired excess return $\nu \geq 0$ exist, i.e., $\argmin_{\pi \in \Pi_\nu} \rho(X_\pi) \neq \emptyset$, where $\Pi_\nu$ denotes all portfolios with $\mathbb{E}[X_\pi] = \nu$?
	\item[(Q2)] \textbf{Absence of (strong) $\rho$-arbitrage.}~What are necessary and sufficient conditions to ensure that the market does not admit (strong) $\rho$-arbitrage?
    \item[(Q3)] \textbf{Well-posedness of mean-$\rho$ portfolio selection.} When do the mean-$\rho$ problems (1) and (2) admit solutions for all $\nu^*,\rho^* \geq 0$?
\end{enumerate}

(Q1) and (Q3) are important from a practical perspective, whilst (Q2) is crucial for the regulator. To the best of our knowledge, all three questions are open at the level of generality we consider here.  The extant literature on mean-$\rho$ portfolio selection when $\rho$ fails to be positively homogeneous is sparse, which is maybe due to the popularity/practical relevance of VaR/ES/coherent risk measures.\footnote{The literature on mean-$\rho$ portfolio selection for positively homogeneous risk measures has been discussed in detail in \cite{herdegen2020dual}, and we refer the interested reader there.} Notwithstanding, the minimisation of convex risk measures has been studied by Ruszczy\'nski and Shapiro \cite{ruszczynski2006optimization}, and their results were later extended to \emph{quasiconvex} risk measures by Mastrogiacomo and Gianin~\cite{mastrogiacomo2015portfolio}.  These two papers are related to ours in the sense that they study the following question: Given a vector space $\mathcal{Z}$ representing the set of actions and a function $F:\mathcal{Z} \to L$ which maps each action to a payoff, when is $\min_{z \in C} \rho(F(z))$ well-posed for some given convex subset  $C$ of $\mathcal{Z}$. Whilst their setting is more general than ours ($\mathcal{Z}$ may be infinite dimensional), their assumptions on $\rho$ are stronger (a ``nice'' dual representation). As an application, they consider the mean-$\rho$ problem (1) and provide sufficient conditions that guarantee the existence of a solution to (1). In particular, these conditions imply the existence of optimal portfolios (in our sense) and thereby answer (Q1), at least partially. Nevertheless, their results do not contribute to answering (Q2) or (Q3). Indeed, even if the mean-$\rho$ problem (1) has a solution, there might still be $\rho$-arbitrage and the mean-$\rho$ problem (2) may not have any solutions; cf.~\cite[Corollary 3.21]{herdegen2020dual}.

We first address (Q1) and show in Theorem \ref{thm:WSTD equiv to boundedness and existence of rho optimal sets} that the crucial ingredient for the existence of optimal portfolios is that $\rho$ satisfies on the set of returns $\cX \subset L$ the new axiom \emph{sensitivity to large expected losses}:~For any $X \in \mathcal{X} \setminus \{0\}$ with $\mathbb{E}[X] \leq 0$, there exists $\lambda > 0$ such that $\rho(\lambda X) > 0$.  The economic meaning of this axiom is simple and intuitive:~Apart from the riskless portfolio, any portfolio with a nonpositive expected excess return has a positive risk if it is scaled by a sufficiently large amount. 

We then turn our attention to (Q2). A key methodological tool here is to consider $\rho^{\infty}$, the \emph{smallest positively homogeneous} risk functional that \emph{dominates} $\rho$.  A key result is that under sensitivity to large expected losses and the Fatou property, $\rho$-arbitrage is equivalent to $\rho^\infty$-arbitrage.  Therefore, if $\rho$ satisfies this and has a dual representation, then so does $\rho^{\infty}$, and we can lift the results from \cite{herdegen2020dual} on the dual characterisation of $\rho$-arbitrage for coherent risk measures to the case of convex risk measures; see~Theorem \ref{thm: no reg arb equivalence}.  This link fails in the case of \emph{strong} $\rho$-arbitrage and the results are more involved.  However, using methods from convex analysis, we are still able to provide a dual characterisation of \emph{strong} $\rho$-arbitrage for convex risk measures, see Theorem~\ref{thm: no strong reg arb}.

As a byproduct of our results on (Q1) and (Q2), we are able to answer (Q3).  We show in Theorem \ref{thm:well posedness} that under sensitivity to large expected losses and the Fatou property, the absence of $\rho$-arbitrage is equivalent to the well-posedness of the mean-$\rho$ problems (1) and (2).

\medskip
The notion of (strong) $\rho$-arbitrage gives rise to the new concept of \emph{(strong) $\rho$-consistent pricing}. %\footnote{Again, a more accurate name would be \emph{(strong) (mean-$\rho$)-consistent pricing}, but we drop the ``mean''; cf.~Remark \ref{rmk:(strong) rho arb definition}(b).}  
The second objective of this paper is to study (strong) $\rho$-consistent pricing by answering the following question for a given market:
\begin{enumerate}
	\item[(Q4)] \textbf{(Strong) $\rho$-consistent pricing.}~What are the set of (strong) $\rho$-consistent prices for a financial contract $Y$ that lives \emph{outside} the market, where consistency is determined through the absence of (strong) $\rho$-arbitrage?
\end{enumerate}

Pricing in a ``market consistent'' way has been an active area of research in mathematical finance for the past few decades.  The idea is simple.  Once we fix a property $\fE$ that translates to forbidding positions in the market that are ``too good to be true'', then the set of market consistent (with respect to $\fE$) prices for $Y$ is given by  
\begin{equation*}
    \textnormal{I}_{\fE}(Y):=\{ y \in \RR : \textnormal{the augmented market with $Y$ priced at $y$ satisfies $\fE$} \}.
\end{equation*}

When $\fE$ stands for no-arbitrage (or no-free lunch with vanishing risk), we are in the classical setting of arbitrage pricing; see e.g.~Delbaen and Schachermayer \cite{delbaen2006mathematics} and the references therein.  While the absence of arbitrage is universally accepted, its implications for pricing are often rather weak, since for incomplete markets, the pricing intervals can be too large to provide any useful information.  Sharper bounds can be obtained by requiring for $\fE$ more/something different than just absence of arbitrage.

There have been many proposals in the literature on how to define $\fE$. Cochrane and Saa-Requejo \cite{cochrane2000beyond} considered Sharpe ratios,
Bernardo and Ledoit \cite{bernardo2000gain} looked at gain-loss ratios, and \v{C}ern\`{y} and Hodges \cite{vcerny2002theory} used utility functions to define $\fE$.  Later contributions in multi-period and continuous-time settings include Kl\"oppel and Schweizer \cite{susanne2007dynamic} and Arai \cite{arai2011good}.

Not surprisingly, risk measures have played an important role for defining $\fE$, starting with Jaschke and K{\"u}chler \cite{jaschke2001coherent}, who studied the absence of \emph{good-deals of the second kind} in a topological framework.  Their results were generalised by Staum \cite{staum2004fundamental}, and sharpened by Cherny \cite{cherny2008pricing}.  More recently, Arduca and Munari \cite{arduca2020fundamental} have studied the absence of \emph{good-deals of the first kind} and \emph{scalable good-deals}.  We elaborate on these papers and discuss how they compare to our work in Remarks \ref{rmk:strong rho arb pricing bounds}(b) and \ref{rmk:applying dual char reg arb}(b).

In this paper, market consistency is defined through the absence of (strong) $\rho$-arbitrage.  This terminology is motivated by the observation that markets that admit (strong) $\rho$-arbitrage are not consistent with $\rho$ since the risk constraint becomes void.  We answer (Q4) by using our answers to (Q2).  In Theorem \ref{cor: elliptical market}, we are able to compute the (strong) $\rho$-consistent price bounds for a new asset in a market with elliptical returns.  Theorem (\ref{cor:strong rho arb pricing bounds}) \ref{cor:rho arb pricing bounds} gives \emph{explicit} (strong) $\rho$-consistent price bounds for $Y$ in a general market using (absolutely continuous) equivalent martingale measures.  

%We stress that -- unlike the extant literature on market consistent pricing, which often yields abstract price intervals --  (strong) $\rho$-consistent price bounds are computable, and sharp enough to be applied in practice.  

%Assuming $\rho$ is a convex risk measure that admits a dual representation $\rho(X) = \sup_{Z \in \mathcal{Q}^\alpha} (\mathbb{E}[-ZX] - \alpha(Z) )$ for some penalty function $\alpha$ with effective domain $\mathcal{Q}^\alpha$, we \emph{explicitly} derive no-strong-$\rho$-arbitrage/no-$\rho$-arbitrage price bounds.  This comes down to taking discounted expectations with respect to absolutely continuous/equivalent martingale measures for the discounted risky assets that lie within $\mathcal{Q}^{\overline{\textnormal{co}} \, \alpha}$/$\tilde{\mathcal{Q}}^\alpha$, where $\mathcal{Q}^{\overline{\textnormal{co}} 
 %\, \alpha}$ is the effective domain of $\overline{\textnormal{co}} \, \alpha$ and $\tilde{\mathcal{Q}}^\alpha$ is an ``interior'' of $\mathcal{Q}^\alpha$.  %For a wide range of risk measures it is not difficult to compute $\mathcal{Q}^{\overline{\textnormal{co}} \, \alpha}$ and $\tilde{\mathcal{Q}}^\alpha$, cf.~Section \ref{sec:Applications}, and so our results also sharpen and extend the literature on good-deal pricing.  

\medskip
The remainder of the paper is organised as follows.  Section \ref{section:setup} describes our modelling framework.  Section \ref{sec:sensitivity to large losses} is devoted to primal answers to (Q1)-(Q3), and also to (Q4) in the case of elliptical markets. Section \ref{sec:convex risk measures} provides a dual characterisation of (strong) $\rho$-arbitrage for convex risk measures and dual answers to (Q4).  In Section \ref{sec:Applications} we apply our theoretical results to two classes of examples: risk functionals based on loss functions and $g$-adjusted Expected Shortfall risk measures.  Section \ref{sec:conclusion} concludes. Appendix~\ref{app:examples} contains some counterexamples complementing the theory, Appendix \ref{app:convex analysis} contains key definitions and results on convex analysis (that are used in Section \ref{sec:convex risk measures}), and Appendix \ref{app:additional results} contains some additional technical results and all proofs.

\section{Modelling framework}
\label{section:setup}

\subsection{Risk framework}
We fix a probability space $(\Omega,\cF,\mathbb{P})$ and work on a Riesz space $L^\infty \subset L \subset L^1$ (for a background on Riesz spaces, see \cite[Chapter 8]{guide2006infinite}).  Key examples for $L$ include $L^p$-spaces for $p \in [1,\infty]$, or more generally Orlicz hearts and Orlicz spaces. We consider one period of uncertainty, where the elements in $L$ represent (discounted) payoffs at time $t=1$ of financial positions held at $t=0$.  The reward for any $X \in L$ is quantified by its expectation, $\mathbb{E}[X]$.  As for the associated risk, we consider a \emph{risk functional} $\rho:L \to (-\infty,\infty]$ satisfying the following axioms:
\begin{itemize}
\itemsep0em 
	\item \emph{Monotonicity:} For any $X_1, X_2 \in L$ such that $X_1 \leq X_2 \ \mathbb{P}$-a.s., $\rho(X_{1}) \geq \rho(X_{2})$;
	\item \emph{Normalisation:} $\rho(0)=0$;
		\item \emph{Star-shapedness:} For all $X \in L$ and $\lambda \in [1,\infty)$, $\rho(\lambda X) \geq \lambda \rho(X)$.
	\end{itemize}
Here, monotonicity means that higher payoffs have lower risk. Normalisation encodes that no investment means no risk.  These two axioms imply that $0$ lies in the \emph{acceptance set}  $\mathcal{A}_\rho:=\{X \in L : \rho(X) \leq 0\}$  of $\rho$ and  $\mathcal{A}_\rho + L_+ \subset \mathcal{A}_\rho$, where
\begin{equation*}
    L_+:=\{ X \in L : X \geq 0 \ \mathbb{P}\textnormal{-a.s.} \}.
\end{equation*}
Finally, star-shapedness captures the idea that a position's risk should increase \emph{at least} proportionally when scaled by a factor greater than one. This is economically sounder and strictly weaker than \emph{positive homogeneity}, where the inequality is replaced by an equality (and $\lambda$ valued in $(0, \infty)$).  Examples of risk functionals include the \emph{worst-case risk measure}, $\WC: L^1 \to (-\infty,\infty]$ given by
\begin{equation*}
    \WC(X):=\esssup(-X),
\end{equation*}
and the \emph{Value at Risk} (VaR) and \emph{Expected Shortfall} (ES) defined by
\begin{equation*}
    \textnormal{VaR}^{\alpha}(X) := \inf \{m \in \mathbb{R} :  \mathbb{P}[m+X < 0] \leq \alpha \} \quad \textnormal{and} \quad \textnormal{ES}^{\alpha}(X) := \frac{1}{\alpha} \int_{0}^{\alpha} \textnormal{VaR}^{u}(X) \dd u,
\end{equation*}
where $\alpha \in (0,1)$ denotes the confidence level and $X \in L^1$.  See Section \ref{sec:Applications} for more examples.

\medskip
Our definition of a risk functional is very general, but for some of our results, in particular for our dual characterisations, we also need some of the following axioms:
\begin{itemize}
\itemsep0em 
	\item \emph{Cash-invariance:} For all $X \in L$ and $c \in \RR$, $\rho(X+c) = \rho(X)-c$;
	\item \emph{Convexity:} For any $X,Y \in L$ and $\lambda \in [0,1]$, $\rho(\lambda X + (1-\lambda)Y) \leq \lambda \rho(X) + (1-\lambda) \rho(Y)$;
	 \item \emph{Fatou property on $\mathcal{Y} \subset L$:} If $X_n \to X$ $\as{\P}$ for $X_n, X \in \mathcal{Y}$ and $|X_n| \leq Y$ $\as{\P}$ for some $Y \in L$ then $\rho(X) \leq \liminf_{n \to \infty} \rho(X_n)$.
\end{itemize}
All three axioms are widely used in the literature.  Cash-invariance allows us to write $\rho(X) = \inf\{m \in \RR : X+m \in \mathcal{A}_\rho\}$ and interpret the value as the minimal amount of capital required to make the position $X$ acceptable.  Such risk functionals are fully characterised by their acceptance set.  Convexity represents the idea that diversification should not increase risk and implies $\mathcal{A}_\rho$ is convex.  Note that under normalisation, convexity implies star-shapedness but the converse is false.  Finally, the Fatou property ensures that risk is never underestimated by approximations; for our applications, it sometimes suffices to consider this on a subset $\mathcal{Y} \subset L$.  

It will be made clear whenever an additional axiom is assumed.  In line with the extant literature, we refer to cash-invariant risk functionals as \emph{risk measures} and positively homogeneous convex risk measures as \emph{coherent risk measures}; see  F{\"o}llmer and Schied \cite[Chapter 4]{follmerschied:2016} for an excellent overview of risk measures -- note however, that their definition of risk measure does not include normalisation.

\medskip
Whilst the key point of this paper is \emph{not} to insist on positive homogeneity of $\rho$, it turns out that its smallest positively homogeneous majorant $\rho^\infty: L \to (-\infty, \infty]$ plays a key role.  This is also a risk functional. %The notation is motivated by the fact that $\rho^\infty$ is the recession function %(see Appendix \ref{app:star-shapedness}) of $\rho$. 
It is explicitly given by
\begin{equation}
	\label{eq:rho bar}
	\rho^\infty(X) := \lim_{t \to \infty} \frac{\rho(tX)}{t}.
\end{equation}
For future reference, note that $\mathcal{A}_{\rho^\infty} = (\mathcal{A}_{\rho})^\infty$, where the latter is used to denote the largest cone contained in $\cA_\rho$, that is,
\begin{equation*}
    (\mathcal{A}_{\rho})^\infty := \{ X \in \cA_\rho : \lambda X \in \cA_\rho \ \textnormal{for all } \lambda \in (0,\infty) \}.
\end{equation*}
Moreover, if $\rho$ satisfies convexity, cash-invariance or the Fatou property on some $\mathcal{Y} \subset L$, then so does $\rho^{\infty}$.  

%\begin{proposition}
%	\label{prop:rho bar well-defined}
%Define the map $\rho^\infty: L \to (-\infty, \infty]$ by
%\begin{equation}
%	\label{eq:rho bar}
%	\rho^\infty(X) := \lim_{t \to \infty} \frac{\rho(tX)}{t}.
%\end{equation}
%Then $\rho^{\infty}$ is well-defined, monotone, normalised and the smallest positively homogeneous risk functional that dominates $\rho$. Moreover, if $\rho$ is convex, cash-invariant or satisfies the Fatou property on some $\mathcal{Y} \subset L$, then so does $\rho^{\infty}$.
%\end{proposition}
%
%\begin{proof}
%	Let $X \in L$.  First note that by star-shapedness of $\rho$, for any $k \geq 1$ and $s > 0$,
%	\begin{equation*}
%		\rho((ks)X)/(ks) \geq k\rho(sX)/(ks) = \rho(sX)/s.
%	\end{equation*}
%	Whence, $\rho(tX)/t$ is nondecreasing in $t$ and so $\rho^\infty$ is well-defined as a $(-\infty,\infty]$-valued map on $L$. Monotonicity and normalisation are directly inherited from $\rho$ (as are convexity, cash-invariance or the Fatou property on some $\mathcal{Y} \subset L$). Positive homogeneity follows directly from the definition. Finally, if $\tilde{\rho}$ is a positively homogeneous map that dominates $\rho$, then
%	\begin{equation*}
%		\tilde{\rho}(X) = \lim_{t \to \infty} \tilde{\rho}(tX)/t \geq \lim_{t \to \infty} \rho(tX)/t = \rho^\infty(X),
%	\end{equation*}
%whence $\rho^{\infty}$ is minimal.
	%Fatou property: if X_n \to X, then \rho^\infty(X) = \lim_{t \to \infty} \rho(tX)/t \leq \lim_{t \to \infty} \liminf_{n \to \infty} \rho(tX_n)/t \leq \liminf_{n \to \infty} \rho^\infty(X_n) because for any X and t, \rho^\infty(X) \geq \rho(tX)/t.
%\end{proof}

\subsection{Portfolio framework}
\label{section:portfolio optimisation}

We consider a one-period $(1+d)$-dimensional market $(S^0_t, \ldots, S^d_t)_{t \in \{0, 1\}}$.  We assume that $S^0$ is riskless and satisfies
$S^0_0 = 1$ and $S^0_1 = 1 + r$, where $r > -1$ denotes the riskless rate. We further assume that $S^1, \ldots, S^d$ are risky assets, where $S^1_0, \ldots, S^d_0 > 0$ and $S^1_1, \ldots, S^d_1 \in L$. We denote the (relative) \emph{return} of asset $i \in \{0, \ldots, d\}$ by 
\begin{equation*}
R^i := \frac{S^i_1 - S^i_0}{S^i_0},
\end{equation*}
and its expectation by $\mu^i := \mathbb{E}[R^i]$.  For notational convenience, we set $S:= (S^1, \ldots, S^d)$, $R:= (R^1, \ldots, R^d)$ and $\mu:= (\mu^1, \ldots, \mu^d) \in \RR^d$.  We may assume without loss of generality that the market is \emph{nonredundant}, i.e., $\sum_{i =0}^d \theta^i S^i =0$ $\as{\P}$ implies that $\theta^i = 0$ for all $i \in \{0, \ldots, d\}$.  We also assume that the risky returns are \emph{nondegenerate} in the sense that for at least one $i \in \{1, \ldots, d\}$, $\mu^i \neq r$. (If $\mu^i = r$ for all $i \in \{1, \ldots, d\}$, then every portfolio $\pi \in \RR^d$ has zero expected excess return.  There would be no incentive to invest and mean-risk portfolio optimisation becomes meaningless.)  Note that this implies that $\P$ itself is not an equivalent martingale measure for the discounted risky assets $S/S^0$.

\begin{remark}
\label{rmk:positive prices}
While the assumption $S^i_0 > 0$ for $i \in \{1, \ldots d\}$ is necessary in order to define the (relative) return $R^i$ in a meaningful way and to parametrise trading in fractions of wealth, it is without loss of generality.  Indeed, if $S^i_0 \leq 0$ for some $i 
\in \{1, \ldots, d\}$, then one can define $\tilde{S}^i := S^i + (1- S^i_0) S^0$ (so that $\tilde S^i_0 = 1$), and the original market is economically equivalent to the market with $S^i$ replaced by $\tilde{S}^i$.
\end{remark}

As $S^0_0, \ldots, S^d_0 > 0$, we can parametrise trading in \emph{fractions of wealth}, and we assume that trading is frictionless. More precisely, we fix an initial wealth $x_0 > 0$ and describe any portfolio (for this initial wealth) by a vector $\pi =(\pi^1, \ldots, \pi^d) \in \RR^d$, where $\pi^i$ denotes the fraction of wealth invested in asset $i \in \{1, \ldots, d\}$. The fraction of wealth invested in the riskless asset is in turn given by $\pi^0 := 1 - \sum_{i =1}^d \pi^i = 1 - \pi \cdot \1$, where $\1 :=(1, \ldots, 1) \in \RR^d$.  The \emph{return} of a portfolio $\pi \in \RR^d$ can be computed by $R_\pi := (1 - \pi \cdot \1)r + \pi \cdot R$, and its \emph{excess return} over the riskless rate $r$ is in turn given by
\begin{equation}
\label{eq:excess return}
X_{\pi} := R_\pi - r = (1 - \pi \cdot \1)r + \pi \cdot R - r =\pi \cdot(R -r\1).
\end{equation}
Thus, $\cX:=\{X_\pi:\pi \in \RR^d\}$ is a subspace of $L$, independent of the initial wealth $x_0$.  The \emph{expected excess return} of a portfolio $\pi \in \RR^d$ over the riskless rate $r$ can be calculated as
$\mathbb{E}[X_{\pi}] = \pi \cdot (\mu - r \1)$, and the set of all portfolios with expected excess return $\nu \in \RR$ is given by
\begin{equation}
\label{eq:def:Pi nu}
\Pi_{\nu} := \{\pi \in \RR^d: \mathbb{E}[X_{\pi}] = \nu\}.
\end{equation}
Then  $\Pi_{\nu}$ is nonempty, closed and convex for all $\nu \in \RR$. Moreover, 
\begin{equation}
\label{eq:Pi k}
\Pi_{\lambda \nu} = \lambda \Pi_\nu \quad \textnormal{for any } \nu \in \RR \textnormal{ and } \lambda \in \RR \setminus \{0\}.
\end{equation}
The \emph{risk} associated to a portfolio $\pi$ is given by $\rho(X_\pi)$.

\section{Mean-$\rho$ portfolio selection, $\rho$-arbitrage and $\rho$-consistent pricing}
\label{sec:sensitivity to large losses}

We start our discussion on mean-$\rho$ portfolio optimisation (concurrently, mean-$\rho^\infty$ portfolio optimisation) by introducing a preference preorder on the set of portfolios. This  preorder formalises the idea that return is ``desirable'' and risk is ``undesirable''.  

\begin{definition}
	\label{def:rho preferred}
A portfolio $\pi \in \RR^d$ is \emph{strictly $\rho$-preferred} over another portfolio $\pi^\prime \in \RR^d$ if $\mathbb{E}[X_{\pi}] \geq \mathbb{E}[X_{\pi^\prime}]$ and $\rho(X_{\pi}) \leq \rho(X_{\pi^\prime})$, with at least one inequality being strict.
\end{definition}

\begin{comment}
\begin{remark}
\label{rmk:excess returns cash invariant}
If $\rho$ is a cash-invariant risk measure, one could argue that it should be applied to the discounted net worth of the portfolio at $t=1$, $Y_\pi:=\tfrac{x_0X_\pi}{1+r}$, as $\rho(Y_\pi)$ can be interpreted as the \emph{additional capital} that is required (to be kept in reserve at $t=0$) for the agent's plans to be acceptable.  However, working with excess returns is mathematically more convenient and by Remark \ref{rmk:reparemtrise}(b), our results can easily be translated.
\end{remark}
\end{comment}

\noindent It gives rise to the mean-$\rho$ problems:
\begin{enumerate}[label=({\arabic*})] 
\item Given a minimal desired expected excess return $\nu^* \in [0,\infty)$, minimise $\rho(X_\pi)$ among all portfolios $\pi \in \RR^d$ that satisfy $\mathbb{E}[X_{\pi}] \geq \nu^*$.%\nu_\min is reserved for something else
    \item Given a maximal risk threshold $\rho^* \in [0,\infty)$, maximise $\mathbb{E}[X_\pi]$ among all portfolios $\pi \in \RR^d$ that satisfy $\rho(X_\pi) \leq \rho^*$.
   
\end{enumerate}

\noindent The way to tackle mean-$\rho$ portfolio selection is to first find the set of $\rho$\emph{-optimal} portfolios, and to then use the \emph{$\rho$-optimal boundary} to find the solutions to (1) and (2).

\begin{definition}
\label{def:optimal portfolio}
Let $\nu \geq 0$. A portfolio $\pi \in \Pi_\nu$ is called \emph{$\rho$-optimal} for $\nu$ if $\rho(X_\pi) < \infty$ and $\rho(X_\pi) \leq \rho(X_{\pi'})$ for all $\pi' \in \Pi_\nu$. We denote the set of all $\rho$-optimal portfolios for $\nu$ by $\Pi^\rho_\nu$. Moreover, we set
\begin{equation}
\label{eq:def:optimal portfolio:rho nu}
    \rho_\nu := \inf \{ \rho(X_\pi) : \pi \in \Pi_{\nu} \},
\end{equation}
and define the \emph{$\rho$-optimal boundary} by
\begin{equation}
\label{eq:def:optimal portfolio:o rho}
\mathcal{O}_{\rho} := \{ (\rho_\nu, \nu): \nu \geq 0 \}.
\end{equation}
\end{definition}

%The way to tackle these two problems is to first study problem (1) with an equality constraint, i.e., for fixed $\nu \geq 0$, find the minimal risk among the portfolios in $\Pi_\nu$:
%\begin{enumerate} 
%    \item [(1')] For $\nu \geq 0$, minimise $\rho(X_\pi)$ among all portfolios $\pi \in \Pi_\nu$.
%\end{enumerate}

In classical mean-variance portfolio selection, optimal portfolios exist for all $\nu \geq 0$ and they form the solutions to problems (1) and (2). By contrast, this can all break down in the mean-$\rho$ framework.  Optimal portfolios need not exist, cf.~\cite[Example A.1]{herdegen2020dual}.  Even if $\rho$-optimal portfolios do exist, the mean-$\rho$ problems may still be ill-posed, for example if the market admits  \emph{(strong) $\rho$-arbitrage} (defined below), cf.~Example \ref{exa:binomial} in Appendix \ref{app:examples}.  %This motivates the following definition.

\begin{definition}
\label{def:rho arbitrage}
The market $(S^0, S)$ is said to admit \emph{$\rho$-arbitrage} if there exists a sequence of portfolios $(\pi_{n})_{n \geq 1} \subset \RR^d$ such that 
\begin{equation*}
    \mathbb{E}[X_{\pi_{n}}] \uparrow \infty \quad \textnormal{and} \quad \rho(X_{\pi_{n}}) \leq 0 \ \textnormal{for all } n.
\end{equation*}
It is said to admit \emph{strong $\rho$-arbitrage} if there is a sequence of portfolios $(\pi_{n})_{n \geq 1} \subset \RR^d$ such that 
\begin{equation*}
    \mathbb{E}[X_{\pi_{n}}] \uparrow \infty \quad \textnormal{and} \quad \rho(X_{\pi_{n}}) \downarrow -\infty.
\end{equation*}
\end{definition}

\begin{remark}
\label{rmk:(strong) rho arb definition}
(a)  This definition of strong $\rho$-arbitrage is equivalent to \cite[Definition 3.17]{herdegen2020dual} when $\rho$ is positively homogeneous; see \cite[Remark 3.19]{herdegen2020dual}.  Aside from the (very unusual) scenario where $\Pi^\rho_\nu = \emptyset$ for all $\nu \geq 0$ and $\rho_\nu \geq 0$ for all $\nu >0$, this definition of $\rho$-arbitrage is equivalent to \cite[Definition 3.17]{herdegen2020dual} when $\rho$ is positively homogeneous.  
Note, however, that the existence of a portfolio $\pi$ with $\mathbb{E}[X_\pi] > 0$ and $\rho(X_\pi) \leq 0$ alone does \emph{not} necessarily constitute a $\rho$-arbitrage since this “riskless profit” (in terms of $\rho$) may not be scaled.

%(b) Technically speaking, it would be more precise to refer to (strong) $\rho$-arbitrage as \emph{(strong) (mean-$\rho$)-arbitrage}.  Indeed, we are using the expectation to measure the reward of a portfolio.  In a follow-up paper, we aim to study \emph{$u$-$\rho$ portfolio selection}, which is a generalisation of mean-$\rho$ portfolio selection where $u:L \to [-\infty,\infty)$ is a \emph{reward measure}.  There, the notion of \emph{(strong) ($u$-$\rho$)-arbitrage} will appear.  However, due to the importance of mean-$\rho$ portfolio selection, and since in this paper we only use the mean to quantify the reward of a portfolio, we drop the ``mean'' and just use (strong) $\rho$-arbitrage.

(b) If $\rho$ is a positively homogeneous risk functional that is \emph{expectation bounded} ($\rho(X) \geq \mathbb{E}[-X]$ for all $X \in L$), the market admits strong $\rho$-arbitrage if and only if there exists a portfolio $\pi \in \RR^d$ (in fractions of wealth) such that $\rho(X_\pi) < 0$.  This is equivalent to the existence of a portfolio $(\theta^0,\theta) \in \RR^{1+d}$ (in numbers of shares) and $\epsilon > 0$ such that
\begin{equation*}
    \theta^0 S^0_0 + \theta \cdot S_0 \leq 0 \quad \textnormal{and} \quad \theta^0 S^0_1 + \theta \cdot S_1 - \epsilon \in \mathcal{A}_{\rho},
\end{equation*}
which is referred to as a \emph{good-deal of the second kind}, see e.g.~\cite{jaschke2001coherent}.  A \emph{good-deal of the first kind} is a portfolio $(\theta^0,\theta) \in \RR^{1+d} \setminus \{\mathbf{0}\}$ (in numbers of shares) such that
\begin{equation*}
    \theta^0 S^0_0 + \theta \cdot S_0 \leq 0 \quad \textnormal{and} \quad \theta^0 S^0_1 + \theta \cdot S_1 \in \mathcal{A}_{\rho}.
\end{equation*}
In our setting, this corresponds to a portfolio $\pi \in \RR^d \setminus \{\mathbf{0}\}$ (in fractions of wealth) with $X_\pi \in \mathcal{A}_{\rho}$. Thus, when $\rho$ is a positively homogeneous risk measure that satisfies the Fatou property and \emph{strict expectation boundedness} ($\rho(X) > \mathbb{E}[-X]$ for all non-constant $X \in L$), the existence of a good-deal of the first kind is equivalent to the market admitting $\rho$-arbitrage by \cite[Theorem 3.11 and Theorem  3.20]{herdegen2020dual}.
\end{remark}

(Strong) $\rho$-arbitrage is an extension of \emph{arbitrage of the first (second) kind}.  They coincide when $\rho$ is the worst-case risk measure, cf.~\cite[Proposition 3.22]{herdegen2020dual}. 

\begin{definition}
We say the market $(S^0,S)$ admits \emph{arbitrage of the first kind} if there exists a trading strategy $(\theta^0,\theta) \in \mathbb{R}^{1+d}$ (parametrised in numbers of shares) such that
	\begin{equation*}
	\theta^0 S^0_0 + \theta \cdot S_0 \leq 0, \quad \theta^0 S^0_1 + \theta \cdot S_1 \geq 0 \; \mathbb{P}\textnormal{-a.s.} \quad \textnormal{and} \quad \mathbb{P}[\theta^0 S^0_1 + \theta \cdot S_1 > 0] > 0.
\end{equation*}
It admits \emph{arbitrage of the  second kind} if there exists a trading strategy $(\theta^0,\theta) \in \mathbb{R}^{1+d}$ such that 
	\begin{equation*}
	\theta^0 S^0_0 + \theta \cdot S_0 < 0, \quad \textnormal{and} \quad \theta^0 S^0_1 + \theta \cdot S_1 \geq 0 \; \mathbb{P}\textnormal{-a.s.}
	\end{equation*} 
\end{definition}

\noindent It is not difficult to show that in general, arbitrage of the first kind implies $\rho$-arbitrage, and when $\rho$ is unbounded from below (e.g.~when $\rho$ is cash-invariant), arbitrage of the second kind implies strong $\rho$-arbitrage.  For more connections between classical arbitrage and (strong) $\rho$-arbitrage, see Corollary \ref{cor:rho infinity WC} and Theorem \ref{prop:A rho infinity larger than L+}.
 
\subsection{Sensitivity to large expected losses}

We first seek to understand under which conditions $\rho$-optimal portfolios exist (i.e., address (Q1)) and what properties $\rho$-optimal sets have.  To that end, we introduce the following axiom.

\begin{definition}
\label{def:WSLL}
The risk functional $\rho$  is said to satisfy \emph{sensitivity to large expected losses} on $\mathcal{Y} \subset L$ if for each $X \in \mathcal{Y} \setminus \{ 0 \}$ with $\mathbb{E}[X] \leq 0$, there exists $\lambda \in (0,\infty)$ such that $\lambda X \not\in \mathcal{A}_\rho$.
\end{definition}

The financial interpretation of sensitivity to large expected losses is clear:
For any nonzero position (in $\mathcal{Y}$) that has a nonpositive expectation, there is eventually a point where the scaled position is considered unacceptable by the person choosing the risk functional, in our case the regulator. It reflects a new notion of loss aversion and is a special case of the more general concept of \emph{sensitivity to large losses} which is studied in detail in \cite{HKM2024}.\footnote{For other existing notions of risk aversion in financial regulation, see \cite{mao2020risk}.}

\begin{remark}
\label{rmk:rho WSLL iff rho bar WSLL}

(a) We will often take $\mathcal{Y} \in \{ \cX, L\}$. Note that $\rho$ satisfies sensitivity to large expected losses on $\mathcal{Y} \subset L$ if and only if $\rho^\infty$ does.  When $\mathcal{Y} = L$, this is equivalent to $\mathcal{A}_{\rho^\infty}  \setminus \{0\} \subset \{X \in L : \mathbb{E}[X] > 0\}$, from which it follows that $\mathcal{A}_{\rho^\infty}$ is \emph{pointed}, i.e., $\mathcal{A}_{\rho^\infty} \cap (-\mathcal{A}_{\rho^\infty}) = \{0\}$.  Pointedness of $\mathcal{A}_{\rho^\infty}$ plays an important role in \cite[Section 4]{arduca2020fundamental}; cf.~also Remark \ref{rmk:applying dual char reg arb}(b).

(b)  It is often the case (see the examples in Section \ref{sec:Applications}) that $\rho$ is sensitive to large expected losses on the entire space $L$.  This is a more general concept than \emph{strict expectation boundedness} ($\rho(X) > \mathbb{E}[-X]$ for all non-constant $X \in L$), which played an important role in \cite{herdegen2020dual}.  The two properties are equivalent when $\rho$ is a positively homogeneous risk measure.
\end{remark}

By \cite[Theorem 3.11]{herdegen2020dual}, if $\rho$ is a positively homogeneous risk measure, then sensitivity to large expected losses together with the Fatou property implies that $\Pi^\rho_{\nu}$ is nonempty and compact for all $\nu$ with $\rho_\nu < \infty$. The same result also holds for risk functionals and this answers (Q1).

\begin{theorem}
\label{thm:WSTD equiv to boundedness and existence of rho optimal sets}
Assume $\rho$ is a risk functional that satisfies the Fatou property on $\mathcal{X}$ and sensitivity to large expected losses on $\cX$.  Then for any $\nu \geq 0$ with $\rho_\nu < \infty$, $\Pi^\rho_\nu$ is nonempty and compact. 
 \end{theorem}
 
%note that proof does not make use of monotonicity (deviation risk measures).

\begin{comment}
\begin{remark}
\label{rmk:stict convexity}
(a) If $\rho$ is convex, then so is $\Pi^\rho_\nu$.  If $\rho$ is even \emph{strictly convex on} $\Pi_\nu$ (i.e., $\rho(\lambda X_\pi + (1-\lambda) X_{\pi'}) < \lambda \rho(X_\pi) + (1-\lambda) \rho(X_{\pi'})$ for all $\lambda \in (0,1)$ and $\pi,\pi' \in \Pi_\nu$ with $\rho(X_\pi),\rho(X_{\pi'}) < \infty$), then $\Pi^\rho_\nu$ is a singleton.
%relevant for OCE risk measures defined on ORLICZ HEART where l is strictly convex as they are strictly convex modulo translation, and hence strictly convex on \Pi_\nu.

(b) One might wonder what happens when $\rho$ is not weakly sensitive to large losses on $\cX$.  Then $A_0$ may be unbounded, in which case we lose the boundedness of the sublevel sets of $f_\rho$.  Then for $\nu > 0$, even if $\rho_\nu < \infty$, $\Pi^\rho_\nu$ can be empty; see \cite[Example A.1]{herdegen2020dual} for a counterexample with a coherent risk measure. %which further justifies the introduction of sensitivity to large expected losses.
\end{remark}
\end{comment}

\subsection{Optimal boundary}

We next answer (Q2).  A big step towards this is to understand the map $\nu \mapsto \rho_\nu$ from $\mathbb{R}_+$ to $[-\infty,\infty]$ whose graph corresponds exactly to the $\rho$-optimal boundary (but with the axes reversed). To this end, it turns out useful to relate the map $\nu \mapsto \rho_\nu$ to the map $\nu \mapsto \rho^{\infty}_\nu$.  We start by stating some basic properties. 

\begin{proposition}
	\label{prop:basic properties of rho optimal boundary}
	For a risk functional $\rho$, the map $\nu \mapsto \rho_\nu$ from $\mathbb{R}_+$ to $[-\infty,\infty]$ is star-shaped, i.e., for all $\nu \in \RR_+$ and $\lambda \in [1,\infty)$, $\rho_{\lambda \nu} \geq \lambda \rho_\nu$.  Moreover, $\rho_0 \leq 0$ and the map $\nu \to \rho^\infty_\nu$ from $\RR_+$ to $[-\infty,\infty]$ is a positively homogeneous majorant.
\end{proposition}

%\begin{proof}
%	By the expectation boundedness of $\rho$, we have $\rho_\nu \in [-\nu,\infty]$ for all $\nu \in \RR_+$.  Since the riskless portfolio has zero risk, $\rho_0=0$.  Now let $\nu \in \RR_+$ and $\lambda \geq 1$.  Then by the star-shapedness of $\rho$, \eqref{eq:Pi k}, and \eqref{eq:def:optimal portfolio:rho nu}, it follows that
%	\begin{equation*}
%		\rho_{\lambda \nu} = \inf_{\pi \in \Pi_{\lambda \nu}} \rho(X_\pi) = \inf_{\pi \in \Pi_{\nu}} \rho(X_{\lambda \pi}) =\inf_{\pi \in \Pi_{\nu}} \rho(\lambda X_\pi) \geq \lambda \inf_{\pi \in \Pi_{\nu}} \rho(X_\pi) = \lambda \rho_\nu. 
%	\end{equation*} 
%	The final claim follows from the fact that $\rho^\infty$ is a positively homogeneous risk measure that dominates $\rho$.
%\end{proof}

As a consequence of this result, $\cO_{\rho}$ lies to the left of $\cO_{\rho^\infty}$ in the mean-risk plane.  Moreover, the function $\nu \mapsto \rho_\nu$ is increasing on the interval $\{ \nu \in [\nu^+,\infty) : \rho_\nu < \infty \}$ where
\begin{equation*}
    \nu^+ := \inf\{ \nu \geq 0 : \rho_\nu > 0 \} \in [0,\infty].
\end{equation*}
However, we lack knowledge concerning its behaviour on $(0,\nu^+)$.  The next result shows that sensitivity to large expected losses together with the Fatou property yields a stronger connection between $\cO_\rho$ and $\cO_{\rho^\infty}$ and gives us information concerning
\begin{equation*}
    \nu_{\min}:=\inf\{\nu \geq 0: \rho_{\nu'} > \rho_{\nu} \textnormal{ for all } \nu' > \nu \} \in [0,\nu^+] \quad \textnormal{and} \quad \rho_{\min}:=\inf\{\rho_\nu:\nu \geq 0\} \in [-\infty,0].
\end{equation*}
 
\begin{proposition}
\label{prop:properties of rho optimal boundary}
Let $\rho$ be a risk functional that satisfies the Fatou property on $\mathcal{X}$ and sensitivity to large expected losses on $\cX$.  Then $\nu \mapsto \rho_\nu$ is $(-\infty,\infty]$-valued and lower semi-continuous. Its smallest positively homogeneous majorant is given by $\nu \mapsto \rho_\nu^\infty$, and $\rho^\infty_1 > 0$ if and only if $\nu^+ < \infty$. Moreover, we have the following three cases:
\begin{enumerate}
    \item If $\rho^\infty_1 > 0$, then $\nu_{\min}< \infty$ and $\rho_{\min} = \rho_{\nu_{\min}} \in (-\infty,0]$.
    \item If $\rho^\infty_1=0$, then $\nu_{\min} \in [0,\infty]$ and $\rho_{\min} \in [-\infty,0]$.
    \item If $\rho^\infty_1 < 0$, then $\nu_{\min} = \infty$ and $\rho_{\min} = -\infty$.
\end{enumerate}
 \end{proposition}

\begin{remark}
\label{rmk:smallest positively homogeneous majorant}
One might be inclined to think that $\nu \mapsto \rho_\nu^\infty$ is \emph{always} the smallest positively homogeneous majorant of $\nu \mapsto \rho_\nu$.  However, we need both the Fatou property on $\mathcal{X}$ and sensitivity to large expected losses on $\mathcal{X}$ in order for this to hold; cf.~Example \ref{exa:no WSTLL} in Appendix \ref{app:examples}.
\end{remark}

The next result shows that for a \emph{convex} risk functional $\rho$ satisfying the Fatou property and sensitivity to large expected losses, $\mathcal{O}_\rho$ is continuous (except where it possibly jumps to $\infty$) and convex.  It has a strong connection with $\mathcal{O}_{\rho^{\infty}}$ by Proposition \ref{prop:properties of rho optimal boundary}, and by Theorem \ref{thm:WSTD equiv to boundedness and existence of rho optimal sets} every point on the $\rho$-optimal boundary (with finite risk) corresponds to a $\rho$-optimal portfolio. Figure \ref{optimal boundary picture} gives a graphical illustration.

\begin{proposition}
\label{prop:convex and WSTD optimal boundary properties}
Let $\rho$ be a convex risk functional that satisfies the Fatou property on $\mathcal{X}$ and sensitivity to large expected losses on $\cX$.  Then the map $\nu \mapsto \rho_\nu$ from $\RR_+$ to $(-\infty,\infty]$ is convex, continuous on the closed set $\{\nu \in \RR_+ : \rho_{\nu} < \infty\}$ and
\begin{equation*}
    \rho^\infty_1 > 0 \iff \nu^+ < \infty \iff \nu_{\min} < \infty.
\end{equation*}
Moreover, we have the following three cases:
\begin{enumerate}
    \item If $\rho^\infty_1  > 0$, the map $\nu \mapsto \rho_\nu$ is nonincreasing on $[0,\nu_{\textnormal{min}}]$, increasing on the closed interval $\{ \nu \in [\nu_{\textnormal{min}},\infty) : \rho_\nu < \infty \}$ and bounded below by $\rho_{\min} = \rho_{\nu_{\min}} \in (-\infty,0]$.
    \item If $\rho^\infty_1=0$, the map $\nu \mapsto \rho_\nu$ is nonincreasing on $\RR_+$ and $\rho_{\min} \in [-\infty,0]$. %map is nonincreasing so global minimum may exist, but not with the properties in (a)  -- This can correspond to strong rho arbitrage eg the smallest positively homogeneous majorant of -ln(x+1) for x \geq 0 is y=0.
    \item If $\rho^\infty_1 < 0$, the map $\nu \mapsto \rho_\nu$ is decreasing on $\RR_+$ and $\rho_{\min} = -\infty$.
\end{enumerate}
\end{proposition}

%\begin{remark}
%In particular, note the upgrade in Proposition \ref{prop:convex and WSTD optimal boundary properties}(b) from Proposition \ref{prop:properties of rho optimal boundary}(b).
%\end{remark}

\begin{figure}[H]
	\centering
	\includegraphics[width=1.0\textwidth]{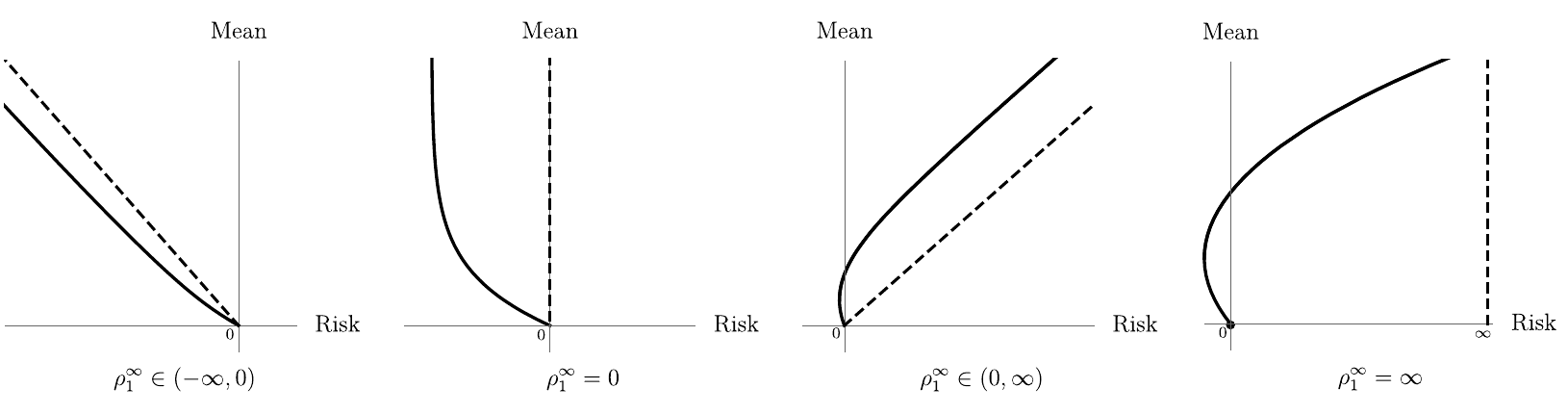}
	\caption{The $\rho$-optimal boundary (solid) and $\rho^\infty$-optimal boundary (dashed) when $\rho$ satisfies convexity, the Fatou property on $\cX$, sensitivity to large expected losses on $\cX$, and $\rho_\nu < \infty$ for all $\nu \geq 0$.  The black dot in the right lower panel is part of the $\rho^{\infty}$-optimal boundary.}
	\label{optimal boundary picture}
\end{figure}

\subsection{(Strong) $\rho$-arbitrage}
\label{subsec:relating rho and rho bar arbitrage}

We can now give primal characterisations of (strong) $\rho$-arbitrage in terms of the sign of $\rho^\infty_1$.  In particular, note that $\rho$-arbitrage is fully characterised by the sign of $\rho^\infty_1$ when $\rho$ satisfies the Fatou property on $\cX$ and sensitivity to large expected losses on $\cX$.  

\begin{theorem}
\label{prop:strong reg arb:first characterisation}
For a risk functional $\rho$, we have \textnormal{(a)} $\iff$ \textnormal{(b)} $\implies$ \textnormal{(c)} for the statements:
\begin{enumerate}
    \item $\rho^\infty_1 < 0$.
    \item The market $(S^0,S)$ admits strong $\rho^\infty$-arbitrage.
    \item The market $(S^0,S)$ admits strong $\rho$-arbitrage.
\end{enumerate}
\end{theorem}

\begin{remark}
\label{rmk:strong rho arb not equivalent to strong rho infinity arb}
Remark  \ref{rem:strong rho bar arbitrage} shows that the implication ``(c) $\implies$ (b)'' does not hold even if $\rho$ satisfies the Fatou property and sensitivity to large expected losses on $\cX$.  That being said, Theorem \ref{thm: no strong reg arb} provides conditions for when strong $\rho$-arbitrage \emph{is equivalent} to strong $\rho^\infty$-arbitrage.
\end{remark}

\begin{theorem}
\label{thm:reg arb:first characterisation}
Assume $\rho$ is a risk functional that satisfies the Fatou property on $\mathcal{X}$ and sensitivity to large expected losses on $\cX$.  Then the following are equivalent:
\begin{enumerate}
    \item $\rho^\infty_1 > 0 $.
    \item The market $(S^0,S)$ does not admit $\rho^\infty$-arbitrage.
    \item The market $(S^0,S)$ does not admit $\rho$-arbitrage.
\end{enumerate}
\end{theorem}

\begin{remark}
\label{rmk:rho arb scalable acceptable deal relation}
(a) Note that the implication ``(c) $\implies$ (b)'' remains true even without the Fatou property on $\mathcal{X}$ or sensitivity to large expected losses on $\cX$. However, Example \ref{exa:no WSTLL} in Appendix~\ref{app:examples} shows that in order for the reverse implication to hold true we need both these properties simultaneously (see also Remark \ref{rmk:smallest positively homogeneous majorant}).

(b) By Theorem \ref{thm:reg arb:first characterisation} and Remark \ref{rmk:rho WSLL iff rho bar WSLL}(a), when $\rho$ is sensitive to large expected losses and satisfies the Fatou property, then the market admits $\rho$-arbitrage if and only if there exists a portfolio $\pi \in \RR^d \setminus \{\mathbf{0}\}$ (in fractions of wealth) with $X_\pi \in \mathcal{A}_{\rho^\infty}$.  This is equivalent to the existence of a portfolio $(\theta^0,\theta) \in \RR^{1+d} \setminus \{\mathbf{0}\}$ (in numbers of shares) such that
\begin{equation*}
    \theta^0 S^0_0 + \theta \cdot S_0 \leq 0 \quad \textnormal{and} \quad \theta^0 S^0_1 + \theta \cdot S_1 \in \mathcal{A}_{\rho}^\infty.
\end{equation*}
This is referred to as a \emph{scalable good-deal} in \cite{arduca2020fundamental}. 
\end{remark}

\begin{remark}
(a) Combining Theorems \ref{prop:strong reg arb:first characterisation} and \ref{thm:reg arb:first characterisation} allows us to conclude the following when $\rho$ satisfies the Fatou property on $\cX$ and sensitivity to large expected losses on $\cX$: if $\rho^\infty_1 < 0$, then the market admits strong $\rho$-arbitrage; if $\rho^\infty_1 = 0$, then the market admits $\rho$-arbitrage and it may or may not admit strong $\rho$-arbitrage; if $\rho^\infty_1 > 0$, then the market does not admit $\rho$-arbitrage.  %In particular, note that when $\rho^\infty_1 = 0$ we cannot determine whether or not the market admits strong $\rho$-arbitrage.

(b) Fix $\alpha \in (0,1)$.  Let $\rho \equiv \textnormal{OCE}^l$ (see Section \ref{subsec:loss function risk measures}) where $l$ is a \emph{loss function} satisfying
   \begin{equation*}
       \lim_{x \to -\infty} \frac{l(x)}{x} = 0 \quad \textnormal{and} \quad \lim_{x \to \infty} \frac{l(x)}{x} = \frac{1}{\alpha}.
   \end{equation*}
One can verify that $\rho^\infty \equiv \ES^\alpha$ by Proposition 5.6, together with the dual representation of Expected Shortfall and the fact $[0,\tfrac{1}{\alpha}] \supset \textnormal{dom} \, l^* \supset (0,\tfrac{1}{\alpha})$.  Performing mean-$\rho$ portfolio selection in the binomial market from Example \ref{exa:binomial} in Appendix \ref{app:examples} with $\alpha < \alpha^*$, $\alpha = \alpha^*$ and $\alpha > \alpha^*$ illustrates simple cases where, respectively, $\rho^\infty_1 > 0$, $\rho^\infty_1 = 0$ and $\rho^\infty_1 < 0$ can occur.
\end{remark}

Theorems \ref{prop:strong reg arb:first characterisation} and \ref{thm:reg arb:first characterisation} go a long way in answering (Q2) from the introduction. When $\cA_{\rho^\infty}$ is equal to the positive cone, we can further establish a useful relationship between $\rho$-arbitrage and arbitrage of the first kind.  This is a consequence of Theorem \ref{thm:reg arb:first characterisation} and Remark \ref{rmk:rho WSLL iff rho bar WSLL}(a).

\begin{corollary}
\label{cor:rho infinity WC}
    Assume $\rho$ is a risk functional that satisfies the Fatou property on $\mathcal{X}$ and $\cA_{\rho^\infty} = L_+$.  Then the following are equivalent:
\begin{enumerate}
    \item The market $(S^0,S)$ does not admit $\rho$-arbitrage.
    \item The market $(S^0,S)$ does not admit arbitrage of the first kind.
\end{enumerate}
\end{corollary}

\noindent The following result shows that when $\cA_{\rho^\infty} \supsetneq L_+$, this equivalence breaks down. It can be viewed as a generalisation of \cite[Theorem 3.23]{herdegen2020dual}.

\begin{theorem}
\label{prop:A rho infinity larger than L+}
Assume $\rho$ is a (cash-invariant) risk functional such that $\cA_{\rho^\infty} \supsetneq L_+$.  Then there exists a market with returns in $L$ that does not admit arbitrage of the first kind, but admits (strong) $\rho$-arbitrage.  %Moreover, if $\rho$ is cash-invariant, then there exists a market with returns in $L$ that does not admit arbitrage of the first kind, but admits strong $\rho$-arbitrage. 
\end{theorem}

\indent The primal characterisations of (strong) $\rho$-arbitrage in Theorems \ref{prop:strong reg arb:first characterisation} and \ref{thm:reg arb:first characterisation} are also useful when returns are \emph{elliptically distributed}.  We will apply them to that setting in Section \ref{subsec:consistent pricing elliptical}, but before this we turn to (Q3).

\begin{theorem}
    \label{thm:well posedness}
Assume $\rho$ is a risk functional that satisfies the Fatou property on $\mathcal{X}$ and sensitivity to large expected losses on $\cX$.  The following are equivalent:
\begin{enumerate}
    \item The market $(S^0,S)$ does not admit $\rho$-arbitrage.
    \item For any $\nu^*, \rho^* \geq 0$, the mean-$\rho$ problems (1) and (2) admit at least one solution.
\end{enumerate}
\end{theorem}

\subsection{(Strong) $\rho$-consistent pricing for elliptical returns}
\label{subsec:consistent pricing elliptical}
The criteria in Theorems \ref{prop:strong reg arb:first characterisation} and \ref{thm:reg arb:first characterisation} are rather indirect as they rely on computing the sign of $\rho^\infty_1$.  This is nontrivial in general.  However, one can easily extend \cite[Section 3.4]{herdegen2020dual} to give a characterisation of (strong) $\rho$-arbitrage in terms of the \emph{maximal Sharpe ratio} when returns are elliptical.  As a consequence, in this setting, we can answer (Q4) from the introduction. To this end, we introduce the concept of \emph{(strong) $\rho$-consistent pricing}.

\begin{definition}
\label{def:rho consistent pricing}
Let $X \in L$.  We say $x \in \RR$ is a \emph{(strong) $\rho$-consistent} price for $X$ if the augmented market $(S^0,S,S^{d+1})$ with $S^{d+1}_0=x$ and $S^{d+1}_1 = X$ does not admit (strong) $\rho$-arbitrage.  We denote the set of (strong) $\rho$-consistent prices for $X$ by $I_\rho(X)$ ($I_{\rho}^s(X)$).
\end{definition}

The idea behind the above definition is as follows: If the original market $(S^0,S)$ is free of (strong) $\rho$-arbitrage, but the augmented market $(S^0,S,S^{d+1})$ is not, then the new asset is priced ``incorrectly'' because a situation arises that is ``too good to be true'', in the sense that the augmented market allows for arbitrary high returns with non-positive risk (with risk tending to $-\infty$).  

\begin{remark}
\label{rem:rho consistent pricing}
(a) Since (strong) $\WC$-arbitrage coincides with arbitrage of the first (second) kind, (cf.~\cite[Proposition 3.22]{herdegen2020dual}) we have that $I_\WC(X)$ coincides with the set of arbitrage-of-the-first-kind-free prices ($I^s_\WC(X)$ coincides with the set of arbitrage-of-the-second-kind-free prices). 

(b) Due to the fact that arbitrage of the first kind implies $\rho$-arbitrage, we have $I_\rho(X) \subset I_\WC(X)$.  Similarly,  $I^s_\rho(X) \subset I^s_\WC(X)$ when $\rho$ is unbounded from below.  Thus, $\rho$-consistent pricing can be seen as a generalisation of classical arbitrage pricing that may result in sharper and more acceptable bounds, cf.~Theorem \ref{cor:rho arb pricing bounds} and Example \ref{exa:2 dim market}.

(c) Since (strong) $\rho$-arbitrage is closely linked to good deals of the first (second) kind, cf.~Remark \ref{rmk:(strong) rho arb definition}, (strong) $\rho$-consistent pricing is closely linked to good-deal pricing, see e.g.~\cite{vcerny2002theory, jaschke2001coherent}. Indeed, the main goal is to price in a way that ensures the absence of ``attractive'' investment opportunities.
\end{remark}

We begin by outlining some basic properties of (strong) $\rho$-consistent pricing.

\begin{proposition}
\label{prop:basic pricing properties}
    Assume $\rho$ is a risk functional, let $X \in L$ and $(S^0,S)$ the original market.
\begin{enumerate}
    \item We have $I_\rho(X) \subset I^s_{\rho}(X) \subset I^s_{\rho^\infty}(X)$.  If $\rho$ satisfies the Fatou property and sensitivity to large expected losses, then $I_\rho(X)=I_{\rho^\infty}(X)$.  In this case, if $\cA_{\rho^\infty} = L_+$, then the set of $\rho$-consistent prices coincides with the set of arbitrage-(of the first kind)-free prices. 
    \item If the market $(S^0,S)$ admits (strong) $\rho$-arbitrage, then $I_\rho(X) = \emptyset$ ($I_{\rho}^s(X) = \emptyset$).
    \item If the market $(S^0,S)$ does not admit (strong) $\rho$-arbitrage and $X = a \cdot (S^0_1,S^1_1,\dots,S^d_1)$ for some $a \in \RR^{d+1}$, then $I_\rho(X) = \{a \cdot (S^0_0,S^1_0,\dots,S^d_0)\}$ ($I^{s}_{\rho}(X) = \{a \cdot (S^0_0,S^1_0,\dots,S^d_0)\}$, if $\rho$ is unbounded from below, and $I^{s}_{\rho}(X) = \RR$ otherwise).
\end{enumerate}
\end{proposition}
Using only the primal characterisations of Theorems \ref{prop:strong reg arb:first characterisation} and \ref{thm:reg arb:first characterisation}, not much can be said about the set of (strong) $\rho$-consistent prices for an external contract beyond Proposition \ref{prop:basic pricing properties} in general. Notwithstanding, if asset prices are \emph{elliptically distributed} and $\rho$ is \emph{law-invariant} we can give explicit pricing intervals.

\begin{definition}\label{Elliptical Distribution Definition}
An $\mathbb{R}^d$-valued random vector $X = (X_{1},\dots,X_{d})$ has an \emph{elliptical distribution} if there exists a \emph{location vector} $a \in \RR^d$, a $d \times d$ nonnegative definite \emph{dispersion matrix} $B \in \RR^{d \times d}$, and a \emph{characteristic generator} $\psi:  \left[0,\infty \right) \rightarrow \mathbb{R}$ such that the characteristic function of $X$, $\phi_{X}$ can be expressed as
\begin{equation*}
\phi_{X}(t) = e^{i t^\top a}\psi(t^{T}Bt) \quad \textnormal{for all} \ t \in \mathbb{R}^{d}.
\end{equation*}
In this case we write $X \sim \tilde E_{d}(a, B, \psi)$. 
\end{definition}

\begin{remark}\label{RemarkAboutElliptical} 
If $X$ has an elliptical distribution with finite second moments, $X$ is also characterised by its mean vector $\mu \in \RR^d$, covariance matrix $\Sigma \in \RR^{d \times d}$  and characteristic generator~$\psi$. Therefore, we may write $X \sim E_{d}\left(\mu, \Sigma, \psi\right)$; see \cite[Remark 3.27]{QRMbook} for details.
\end{remark}

\begin{definition}
\label{defn:law invariance}
A risk functional $\rho : L \to (-\infty, \infty]$ is called \emph{law-invariant} if $\rho(X_1)=\rho(X_2)$ whenever $X_1,X_2 \in L$ have the same law.
\end{definition}

The next result gives the precise set of (strong) $\rho$-consistent prices for an external contract in an elliptical market setting. To this end, in order to easily deal with the case that $x \leq 0$, we replace the asset $S^{d+1}$ by $\tilde{S}^{d+1} := S^{d+1} + (1-x) S^0$ as described in Remark \ref{rmk:positive prices} and similarly set $\tilde S^i := S^{i} + (1- S^i_0) S^0$, and then check if the market $(S^0, \tilde S,\tilde{S}^{d+1})$ admits (strong) $\rho$-arbitrage.

\begin{theorem}
\label{cor: elliptical market}
    Let $(S^0,S)$ be a market, $X \in L$ and assume $(S_1,X)$ has an elliptical distribution with mean vector $\mu$, positive definite covariance matrix $\Sigma$ and characteristic generator $\psi$.  Assume $\{Y \sim E_{1}(\mu_{Y}, \sigma^2_Y, \psi):\mu_Y \in \RR, \ \sigma^{2}_Y \geq 0 \} \subset L$ and let $Z \sim E_{1}\left(0, 1, \psi\right)$.  For $x \in \RR$, define $\tilde \mu(x) \in \RR^{d+1}$ and $\SR_{\max}(x)$ by
    \begin{equation*}
        \tilde \mu^i(x) := 
\begin{cases}
\mu^i - S^i_0(1+r), &i \in \{1, \ldots, d\}, \\
\mu^{d+1} - x(1+r), &i = d+1,
\end{cases} \quad  \text{and} \quad \SR_{\max}(x) := \sqrt{\tilde \mu(x)^\top \Sigma^{-1} \tilde \mu(x)}.
    \end{equation*} 
\begin{enumerate}
    \item If $\rho$ is a law-invariant risk measure and $\rho$-arbitrage is equivalent to $\rho^\infty$-arbitrage, then
    \begin{equation*}
        I_\rho(X) = I_{\rho^\infty}(X) = \{ x \in \RR : \SR_{\max}(x) < \rho^{\infty}(Z) \}. %\cup \{ x^* \in \RR : \SR_{\max}(x^*) = 0 = \rho^{\infty}(Z) \}.
    \end{equation*}
    \item If $\rho$ is a law-invariant risk measure and strong $\rho$-arbitrage is equivalent to strong $\rho^\infty$-arbitrage, then
    \begin{equation*}
        I^s_{\rho}(X) = I^s_{\rho^\infty}(X) = \{ x \in \RR : \SR_{\max}(x) \leq \rho^{\infty}(Z) \}. %\cup \{ x^* \in \RR : \SR_{\max}(x^*) = 0 \}.
    \end{equation*}
\end{enumerate}
\end{theorem}

\begin{remark}
\label{rmk: elliptical market pricing}
(Strong) $\rho$-arbitrage and (strong) $\rho^\infty$-arbitrage are in particular equivalent when $\rho \equiv \rho^\infty$, i.e., when $\rho$ is positively homogeneous.  More generally, Theorem \ref{thm:reg arb:first characterisation} (Theorem \ref{thm: no strong reg arb}) give conditions when (strong) $\rho$-arbitrage is equivalent to (strong) $\rho^\infty$-arbitrage.
\end{remark}

If $\rho^\infty (Z) < \infty$ (which is generically the case if $\rho^\infty \not\equiv \WC$),  Theorem \ref{cor: elliptical market} implies that the set of (strong) $\rho$-consistent prices is a bounded interval and hence meaningful --
unlike the set of no-arbitrage prices (which is $(-\infty,\infty)$ in this situation). We illustrate this with an example, where $\rho^\infty = \rho \equiv \ES^\alpha$, the Expected Shortfall at level $\alpha \in (0,1)$.

\begin{example}
\label{exa:2 dim market}
Assume $(S^1_1,X)$ has an elliptical distribution with characteristic generator $\psi$, mean vector $\mu = (\mu_1,\mu_2)$ and covariance matrix 
    \begin{equation*}
\Sigma = \begin{pmatrix}
\sigma_1^2 & \gamma \sigma_1 \sigma_2 \\
\gamma \sigma_1 \sigma_2 & \sigma_2^2 ,
\end{pmatrix}
    \end{equation*}
where $\gamma \in (-1,1)$ denotes the correlation between $S^1_1$ and $X$. Theorem \ref{cor: elliptical market} then gives
\begin{equation*}
    \SR_{\max}(x) := \sqrt{\tilde \mu(x) ^\top \Sigma^{-1}\tilde \mu(x)} = \sqrt{\tfrac{1}{1-\gamma^2}(\SR_{2}^2(x) - 2\gamma\SR_1\SR_2(x)+\SR_1^2)}
\end{equation*}
where $\SR_1 :=(\mu^1 - S^1_0 (1+r))/\sigma_1$ and $\SR_2(x):=(\mu^2 - x(1+r))/\sigma_2$ denote the Sharpe ratios of $\tilde S^1 = S^1 + (1-S^0) S^0$ and $\tilde{S}^2=X+(1-x)S^1$, respectively.

Fix $\alpha \in (0,1)$ and let $Z \sim E_1(0,1,\psi)$.  Note that by strict expectation boundedness $\ES^\alpha(Z) \in (0, \infty)$.\footnote{If $Z \sim \cN(0, 1)$, then $\ES^\alpha(Z)=\phi(\Phi^{-1}(\alpha))/\alpha$ where $\phi$ and $\Phi$ denote the pdf and cdf of a standard normal distribution, respectively.}  By Theorem \ref{cor: elliptical market}, the original market $(S^0,S^1)$ is free of $\ES^\alpha$-arbitrage (and hence free of strong $\ES^\alpha$-arbitrage) if and only if $\SR_1 \in (-\ES^\alpha(Z),\ES^\alpha(Z))$, which is equivalent to
\begin{equation}
\label{eq:exa:2 dim market:or mark}
S^1_0 \in \left( \frac{\mu_1-\sigma_1\ES^\alpha(Z)}{1+r}, \frac{\mu_1+\sigma_1\ES^\alpha(Z)}{1+r} \right).
\end{equation}
So assume \eqref{eq:exa:2 dim market:or mark} and set
\begin{align*}
(S^2_0)^* := \frac{\mu_2-\gamma \SR_1 \sigma_2}{1+r} \quad \text{and} \quad \kappa := \frac{\sigma_2}{1+r}\sqrt{(1-\gamma^2)(\ES^\alpha(Z)^2-\SR_1^2)}. %\quad \textnormal{and} \\ \SR_2^+ & := \gamma\SR_1 + \sqrt{(1-\gamma^2)((\ES^\alpha(Z))^2-\SR_1^2)},
\end{align*}
Then Theorem \ref{cor: elliptical market} gives, after some algebra,
\begin{align*}
    I_{\ES^\alpha}(X) &=
        \left( (S^2_0)^* - \kappa,(S^2_0)^* + \kappa \right)  \quad \text{and} \quad   I^s_{\ES^\alpha}(X) =
        \left[ (S^2_0)^* - \kappa,(S^2_0)^* + \kappa \right].
\end{align*}
In particular, the length of the (strong) $\rho$-consistent price interval for $X$ is given by $2 \kappa$.

If we allow $\gamma$ to vary (cf.~Figure \ref{interval length picture}, left panel), then the length of the interval is largest when $\gamma = 0$ (perfectly uncorrelated), and decreases to $0$ as $\gamma \uparrow 1$ or $\gamma \downarrow -1$ (perfectly correlated).  

 If we allow $\alpha$ to vary (cf.~Figure \ref{interval length picture}, right panel), then the length of the interval diverges to $\infty$ and the interval itself to $(-\infty, \infty)$ as $\alpha \downarrow 0$. Moreover, the length of the interval decreases to zero and the endpoints of the intervals converge to the midpoint $(S^2_0)^*$ as $\alpha \uparrow \alpha^*$,  where $\alpha^*$ is the critical threshold for the original market to be free of $\ES^\alpha$-arbitrage.
%    \item One could say $(S^2_0)^*$ is the most suitable price to choose for $X$ since regardless of the law-invariant risk measure $\rho$ (assuming (strong) $\rho$-arbitrage is equivalent to (strong) $\rho^\infty$-arbitrage, and replacing $\ES^\alpha(Z)$ by $\rho^\infty(Z)$ in the above analysis), if the original market $(S^0,S^1)$ does not admit (strong) $\rho$-arbitrage, then under this price, the augmented market will also not admit (strong) $\rho$-arbitrage.
\end{example}

\begin{figure}[H]
	\centering
	\includegraphics[width=1.0\textwidth]{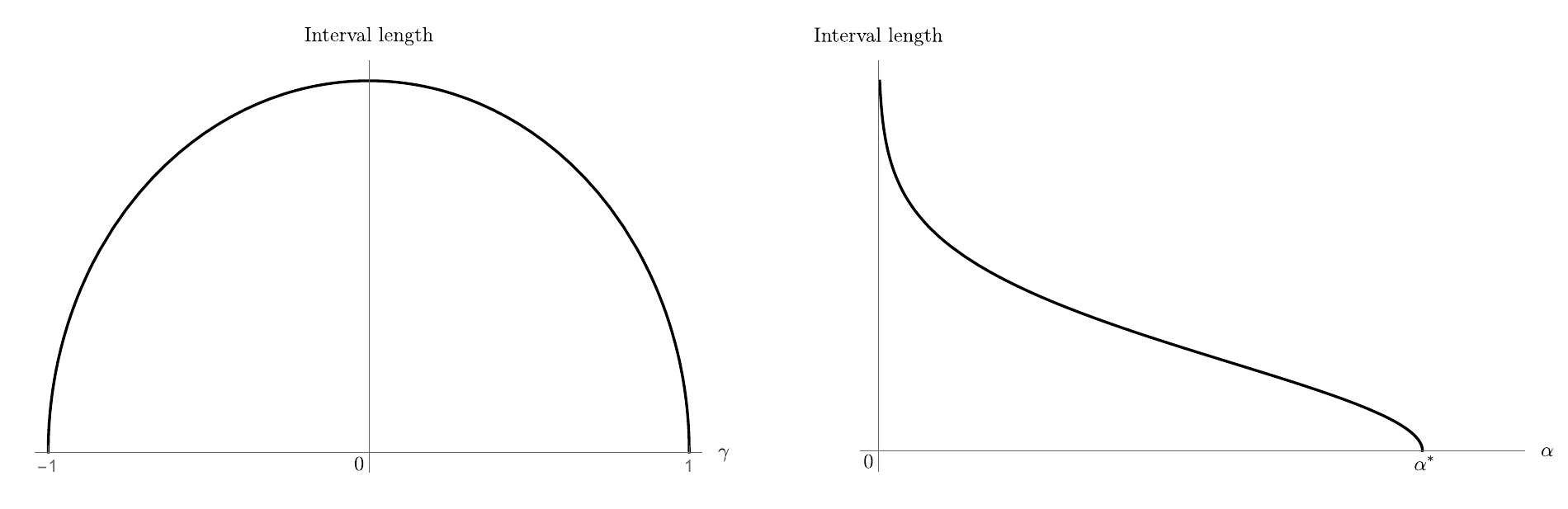}
	\caption{Interval length of  $I_{\ES^\alpha}(X)$. The left panel shows the dependence on the correlation $\gamma$ between $S^1_1$ and $X$; the right panel shows the dependence on the ES sensitivity parameter $\alpha$.}
	\label{interval length picture}
\end{figure}

%When returns are not elliptical... Motivated by \cite[Section 4]{herdegen2020dual}, there is hope that more can be deduced under the assumption that $\rho$ is a convex risk measure that admits a dual representation.  We turn to this in the next section 

\section{Dual characterisation of $\rho$-arbitrage and $\rho$-consistent pricing}\label{sec:convex risk measures}

In this section, we consider the case that $\rho$ is a convex risk measure on $L$ that admits a dual representation.  There are many relevant examples that fall into this  category, cf.~Section \ref{sec:Applications}.

Let $\mathcal{D} := \{Z \in L^1 : Z \geq 0 \ \mathbb{P}\textnormal{-a.s.} \textnormal{ and }  \mathbb{E}[Z]=1  \}$ be the set of all Radon-Nikod{\'y}m derivatives of probability measures that are absolutely continuous with respect to $\mathbb{P}$. Throughout this section, we assume that $\rho: L \xrightarrow{} (-\infty, \infty]$ admits a dual representation
\begin{equation} \label{eq:dual_char_of_rho}
    \rho(X) = \sup_{Z \in \mathcal{Q}^\alpha} \{\mathbb{E}[-ZX] - \alpha(Z)\}
\end{equation}
for some \emph{penalty function} $\alpha:\mathcal{D} \to [0,\infty]$ with effective domain 
$\mathcal{Q}^\alpha := \mathrm{dom\,} \alpha = \{ Z \in \mathcal{D} : \alpha(Z) < \infty \} \neq \emptyset$. The penalty function determines how seriously we treat probabilistic models in $\mathcal{D}$.  Since $\rho$ is normalised, $\inf \alpha = 0$.  Moreover, replacing $\alpha$ if necessary by its quasi-convex hull, we may assume without loss of generality that %$\alpha$ is quasi-convex and 
$\mathcal{Q}^\alpha$ is convex; see Remark \ref{rem:dual char} for details. 

\begin{remark}
\label{rem:dual char}
(a) The supremum in \eqref{eq:dual_char_of_rho} is over the effective domain of $\alpha$. This captures the idea that only the measures ``contained'' in $\mathcal{Q}^\alpha$ are seen as plausible. Since $-ZX$ may not be integrable, we define $\mathbb{E}[-ZX] := \mathbb{E}[ZX^-] - \mathbb{E}[ZX^+]$, with the conservative convention $\mathbb{E}[-ZX] :=\infty$ if $\mathbb{E}[ZX^{-}] = \infty$. 

(b) The class of risk measures satisfying \eqref{eq:dual_char_of_rho} is very large. In particular, we do not impose lower semi-continuity on $\alpha$, or $L^1$-closedness or uniform integrability on $\cQ^\alpha$.

(c)  The representation in \eqref{eq:dual_char_of_rho} is not unique. However, it is not difficult to check that the \emph{minimal} penalty function for which \eqref{eq:dual_char_of_rho} is satisfied is given by 
\begin{equation*}
%	\label{def:alpha rho}
    \alpha^\rho(Z) := \sup\{ \mathbb{E}[-ZX] - \rho(X) : X \in L \textnormal{ and } \rho(X) < \infty \}.
\end{equation*}
\begin{comment}
Indeed, by \eqref{eq:dual_char_of_rho}, for any $X \in L$ and $Z \in \mathcal{D}$ we have $\rho(X) \geq \mathbb{E}[-ZX]-\alpha(Z)$ (where the right side is $-\infty$ if $\alpha(Z)=\infty$), i.e., $\alpha(Z) \geq \mathbb{E}[-ZX]-\rho(X)$ (where the right side is $-\infty$ if $\rho(X) = \infty$).  Whence, $\alpha(Z) \geq \alpha^\rho(Z)$ and so $\rho(X) = \sup_{Z \in \mathcal{D}} \{\mathbb{E}[-ZX] - \alpha(Z)\} \leq \sup_{Z \in \mathcal{D}} \{\mathbb{E}[-ZX] - \alpha^\rho(Z)\}$ for all $X \in L$.  For the reverse inequality, if $\rho(X) < \infty$ we have 
\begin{equation*}
    \sup_{Z \in \mathcal{D}} \{\mathbb{E}[-ZX] - \alpha^\rho(Z)\} \leq \sup \{\mathbb{E}[-ZX] - (\mathbb{E}[-ZX]-\rho(X)): Z \in \mathcal{D}, \  \mathbb{E}[-ZX] \in \mathbb{R}\}\footnote{If $\mathbb{E}[-ZX] = \infty$, then since $\rho(X) < \infty$,  $\alpha^\rho(Z) = \infty$.  By our convention, it would follow that $\mathbb{E}[-ZX] - \alpha^\rho(Z) = -\infty$.  And if $\mathbb{E}[-ZX] = -\infty$, then $\mathbb{E}[-ZX] - \alpha^\rho(Z) = -\infty$.  Thus, we only need to consider the densities that satisfy $\mathbb{E}[-ZX] \in \mathbb{R}$, in which case $\alpha^\rho(Z) \geq \mathbb{E}[-ZX] - \rho(X)$.} = \rho(X),
\end{equation*}
and if $\rho(X) = \infty$, then of course $\sup_{Z \in \mathcal{D}} \{\mathbb{E}[-ZX] - \alpha^\rho(Z)\} \leq \rho(X)$.
\end{comment}
Note that $\alpha^\rho$ is automatically convex. Moreover, its effective domain
$\cQ^{\rho}:=\{Z \in \cD : \alpha^\rho(Z) < \infty\}$ is also convex and the \emph{maximal} dual set. Notwithstanding, it is sometimes useful not to consider $\alpha^\rho$ or $\mathcal{Q}^\rho$; cf.~\cite[Remark 4.1(c)]{herdegen2020dual}.

(d) If $\alpha: \cD \to [0, \infty]$ represents $\rho$ and $\alpha': \cD \to [0, \infty]$ satisfies $\alpha \geq \alpha' \geq \alpha^\rho$, then $\alpha'$ represents $\rho$ and  $\cQ^\alpha \subset \cQ^{\alpha'} \subset \cQ^\rho$. This follows directly by comparing \eqref{eq:dual_char_of_rho} for $\alpha, \alpha'$ and  $\alpha^\rho$.

%(e) If $\cQ^\alpha$ is not convex, we may replace $\alpha$ by its quasi-convex hull $\mathrm{qco\,} \alpha$, i.e., the largest quasi-convex function dominated by $\alpha$; see Appendix \ref{app:convex analysis} for details. Since $\alpha^\rho$ is convex and dominated by $\alpha$, we have $\alpha^\rho \leq \mathrm{qco\,} \alpha \leq \alpha$. Hence, by part (d), $\rho$ can be represented by $\mathrm{qco\,}\alpha$. %and $\cQ^\alpha \subset \cQ^{ \mathrm{qco\,}\alpha} \subset \cQ^\rho$. It follows from the definition of quasi-convexity that $\cQ^{\mathrm{qco\,}\alpha}$ is convex. Moreover, note that $\cQ^{\mathrm{qco\,}\alpha}$ is the convex hull of $\cQ^\alpha$.
\begin{comment}
(f) If we define $\rho$ through \eqref{eq:dual_char_of_rho}  for some function $\alpha: \cD \to [0, \infty]$ satisfying $\inf \alpha = 0$ and for which there exists $Z \in \cQ^{\alpha}$ such that $ZX \in L^1$ for all $X \in L^1$ (often there is no penalty associated with the real-world measure $\mathbb{P}$, i.e., $\alpha(1)=0$), then $\rho$ is a $(-\infty, \infty]$-valued convex risk measure.
\end{comment}
\end{remark}

If $\rho$ admits a dual representation as in \eqref{eq:dual_char_of_rho}, then also $\rho^\infty$ admits the dual representation
\begin{equation}
\label{eq:dual char of bar rho}
\rho^\infty(X) = \sup_{Z \in \cQ^\alpha} (\mathbb{E}[-ZX]).
\end{equation}
%Indeed, the right hand side of \eqref{eq:dual char of bar rho} clearly defines a positively homogeneous risk measure and if $\tilde{\rho}$ is a positively homogeneous risk measure that dominates $\rho$, then for any $X \in L$ and $Z \in \cQ^\alpha$
%	\begin{equation*}
%		\tilde{\rho}(X) = \lim_{t \to \infty}\frac{\tilde{\rho}(tX)}{t} \geq \lim_{t \to \infty}\frac{\rho(tX)}{t} \geq\lim_{t \to \infty}\frac{\mathbb{E}[-ZtX] - \alpha(Z)}{t} = \mathbb{E}[-ZX],
%	\end{equation*}
%which gives $\tilde{\rho}(X) \geq \sup_{Z \in \cQ^\alpha} (\mathbb{E}[-ZX])$.

\subsection{Preliminary considerations and conditions}

For \emph{coherent} risk measures, the dual characterisation of (strong) $\rho$-arbitrage has been studied in~\cite{herdegen2020dual}. We recall the key conditions that were introduced in \cite[Section 4]{herdegen2020dual} and extend some of their consequences to the current setup.

\medskip
\noindent\textbf{Condition I.}~For all $i \in \{1,\dots,d \}$ and any $Z \in \mathcal{Q}^\alpha$, $Z R^i \in L^1$.

\medskip
\noindent\textbf{Condition UI.}~$\mathcal{Q}^\alpha$ is uniformly integrable, and for all $i \in \{1,\dots,d \}$, $R^i\mathcal{Q}^\alpha$ is uniformly integrable, where $R^i\mathcal{Q}^\alpha := \{ R^iZ : Z \in \mathcal{Q}^\alpha \}.$

\medskip
Condition I is weak, but has some important consequences.  Arguing as in \cite[Proposition 4.3]{herdegen2020dual}, we obtain the following result.

\begin{proposition}
	\label{prop:cond I}
	Suppose that Condition I is satisfied. Then for any portfolio $\pi \in \RR^d$, \begin{equation}
\label{eq:relate rho to C}
    \rho(X_{\pi}) = \sup_{c \in C_{{\mathcal{Q}}^\alpha}} (\pi \cdot c - f_{\alpha}(c)),
\end{equation}
where $C_{\mathcal{Q}^{\alpha}} := \{\mathbb{E}[-Z(R-r\mathbf{1})] : Z \in \mathcal{Q}^\alpha \} \subset \mathbb{R}^d$ is convex and $f_\alpha:\mathbb{R}^d \to [0,\infty]$, defined by
\begin{equation*}
    f_\alpha(c) = 
    \inf \{\alpha(Z):Z \in \cQ^\alpha \textnormal{ and } \mathbb{E}[-Z(R-r\mathbf{1})] = c\}
\end{equation*}
%is quasi-convex and 
satisfies $\mathrm{dom}\, f_\alpha = C_{{\mathcal{Q}}^\alpha}$. Moreover, $\rho$ satisfies the Fatou property on  $\cX$.
\end{proposition}

Condition UI is a uniform version of Condition I.  For $X \in \cX$, it allows us to replace $\alpha$ in \eqref{eq:dual_char_of_rho} by its $L^1$-lower semi-continuous convex hull $\overline{\textnormal{co}} \,\alpha$, and the supremum by a maximum.

\begin{proposition}
	\label{prop:Cond UI preliminary result}
	Suppose that Condition UI is satisfied. Denote by $\overline{\textnormal{co}} \,\alpha: \cD \to [0, \infty]$ the $L^1$-lower semi-continuous convex hull of $\alpha$. Then for $X \in \cX$, 
	\begin{equation}
		\label{eq:preliminary representation}
		\rho(X) = \max_{Z \in \cQ^{\overline{\textnormal{co}} \,\alpha}} \{ \mathbb{E}[-ZX]-\overline{\textnormal{co}} \,\alpha(Z) \}.
	\end{equation}
\end{proposition}

As a consequence of Proposition \ref{prop:Cond UI preliminary result}, we obtain the following result which is crucial in establishing the dual characterisation of strong $\rho$-arbitrage.

 \begin{proposition}
	\label{prop:cond UI}
	Suppose that Condition UI is satisfied. Then for any portfolio $\pi \in \RR^d$, 
	\begin{equation}
		\label{eq:relate rho to C hat}
		\rho(X_{\pi}) =  \max_{c \in C_{\cQ^{\overline{\textnormal{co}} \,\alpha}}} (\pi \cdot c - f_{\overline{\textnormal{co}} \,\alpha}(c)),
	\end{equation}
	where $C_{\cQ^{\overline{\textnormal{co}} \,\alpha}} := \{\mathbb{E}[-Z(R-r\mathbf{1})] : Z \in \cQ^{\overline{\textnormal{co}} \,\alpha} \} \subset \RR^d$ is convex and bounded, and $f_{\overline{\textnormal{co}} \,\alpha}:\mathbb{R}^d \to [0,\infty]$, defined by
	\begin{equation}
		\label{eq:prop:cond UI}
		f_{\overline{\textnormal{co}} \,\alpha}(c) = 
		\inf \{\overline{\textnormal{co}} \,\alpha(Z):Z \in \cQ^{\overline{\textnormal{co}} \,\alpha} \textnormal{ and } \mathbb{E}[-Z(R-r\mathbf{1})] = c\}
	\end{equation}
	is the lower semi-continuous convex hull of $f_\alpha$ defined in Proposition \ref{prop:cond I} and satisfies  $\mathrm{dom}\, f_{\overline{\textnormal{co}} \,\alpha} = C_{\cQ^{\overline{\textnormal{co}} \,\alpha}}$. Moreover, the infimum in \eqref{eq:prop:cond UI} is a minimum if $c \in C_{\cQ^{\overline{\textnormal{co}} \,\alpha}}$.
	
	%The key thing is that in (4.5), f_{\alpha^{**}} is lsc and convex.
\end{proposition}

\begin{remark}
	\label{rmk:Q subset Q start start subset Q bar}
	(a) By $\overline{\textnormal{co}} \,\alpha \leq \alpha$ and \eqref{eq:lsc co hull} it follows that $\cQ^\alpha \subset \cQ^{\overline{\textnormal{co}} \,\alpha} \subset \bar{\cQ}^\alpha$, where $\bar{\cQ}^\alpha$ is the $L^1$-closure of $\cQ^\alpha$. Moreover,  if $\alpha$ is bounded from above on its effective domain, then $\cQ^{\overline{\textnormal{co}} \,\alpha} = \bar{\cQ}^\alpha$.

	(b) Since the $L^1$-lower semi-continuous convex hull of $\alpha$ coincides with its $\sigma(L^1,L^\infty)$-lower semi-continuous convex hull  by \cite[Theorem 2.2.1]{zalinescu2002convex}, $\overline{\textnormal{co}} \,\alpha$ coincides with $\alpha^{**}$, the biconjungate of $\alpha$ under the pairing $\langle \cdot, \cdot \rangle: L^1 \times L^\infty \to \RR$, $\langle Z, X \rangle \mapsto \mathbb{E}[-ZX]$, by the Fenchel-Moreau theorem (and the fact that $\alpha$ is nonnegative); see Appendix \ref{app:convex analysis} for details.
\end{remark}

The final object that we need to recall is the ``interior'' of $\cQ^\alpha$, which will be crucial in the dual characterisation of $\rho$-arbitrage.  This is done in an abstract way.  More precisely, we look for (nonempty) subsets $\tilde \cQ^\alpha \subset \cQ^\alpha$ satisfying:

\medskip
\noindent\textbf{Condition POS.}~$\tilde Z > 0$ $\as{\P}$ for all $\tilde Z \in \tilde \cQ^\alpha$.

\medskip
\noindent\textbf{Condition MIX.}~$\lambda Z + (1-\lambda) \tilde Z \in \tilde \cQ^\alpha$ for all $Z \in \cQ^\alpha$, $\tilde Z \in \tilde \cQ^\alpha$ and $\lambda \in (0, 1)$.
\medskip

\noindent\textbf{Condition INT.}~For all $\tilde Z \in \tilde \cQ^\alpha$, there is an $L^\infty$-dense subset $\cE$ of $\cD \cap L^\infty$ such that for all $Z \in \cE$,  there is $\lambda \in (0, 1)$ such that $\lambda Z + (1-\lambda) \tilde Z \in \cQ^\alpha$.

\medskip
By \cite[Proposition 4.9]{herdegen2020dual}, the maximal ``interior'' of $\cQ^\alpha$ can be described explicitly by
\begin{align*}
	\tilde{\mathcal{Q}}^\alpha_{\max} :=\{ 0 < \tilde Z \in \mathcal{Q}^\alpha :  \,& \text{there is an $L^\infty$-dense subset $\cE$ of $\cD \cap L^\infty$ such that for all } Z \in \cE,\\
	& \text{there is } \lambda \in (0,1)  \text{ such that } \lambda Z + (1-\lambda)\tilde Z \in \mathcal{Q}^\alpha \}.
\end{align*}
%However, it is rather of theoretical than practical interest.  

\subsection{Dual characterisation of (strong) $\rho$-arbitrage}
We are now in a position to state and prove the dual characterisation of strong $\rho$-arbitrage in terms of absolutely continuous martingale measures (ACMMs) for the discounted risky assets,
\begin{equation*}
\mathcal{M} := \{Z \in \mathcal{D} : \mathbb{E}[Z(R^{i}-r)] =0 \textnormal{ for all } i=1,\dots ,d \};
\end{equation*}
and the dual characterisation of $\rho$-arbitrage in terms of equivalent martingale measures (EMMs) for the discounted risky assets,
\begin{equation*}
\mathcal{P}  := \{Z \in \mathcal{M} :  Z > 0 \ \mathbb{P}\textnormal{-a.s.} \}.
\end{equation*}

We first consider the dual characterisation of strong $\rho$-arbitrage. Since strong $\rho$-arbitrage is in general not equivalent to strong $\rho^\infty$-arbitrage (cf. Remark \ref{rem:strong rho bar arbitrage}), the following result and its proof are rather delicate and require some more advanced notions of convex analysis; see Appendix \ref{app:convex analysis} for a brief review on some key notions of details.
\begin{theorem}
\label{thm: no strong reg arb}
Assume Condition UI is satisfied and $1 \in \cQ^\alpha$. Denote by $\overline{\textnormal{co}} \,\alpha: \cD \to [0, \infty]$ the $L^1$-lower semi-continuous convex hull of $\alpha$. Then the following are equivalent:
\begin{enumerate}
    \item The market $(S^0,S)$ does not admit strong $\rho$-arbitrage.
    \item $\cQ^{\overline{\textnormal{co}}\,\alpha} \cap \mathcal{M} \neq \emptyset$.
\end{enumerate}
When $\alpha$ is bounded on its effective domain, then $\cQ^{\overline{\textnormal{co}}\,\alpha} = \bar{\cQ}^\alpha$ and strong $\rho$-arbitrage is equivalent to strong $\rho^\infty$-arbitrage.
\end{theorem}

\begin{remark}
\label{rmk:strong rho arb pricing bounds}
(a) Note that in order to apply Theorem \ref{thm: no strong reg arb}, we do not necessarily need to find $\overline{\textnormal{co}} \,\alpha$ but rather only its effective domain $\cQ^{\overline{\textnormal{co}} \,\alpha}$. 

(b) When $\rho$ is coherent, then $\cQ^{\overline{\textnormal{co}} \,\alpha} = \bar{\cQ}^\alpha$ by Remark \ref{rmk:Q subset Q start start subset Q bar}(a), and Theorem \ref{thm: no strong reg arb} reduces to \cite[Theorem 4.15]{herdegen2020dual}.  By Remark \ref{rmk:(strong) rho arb definition}(b), this is comparable to Cherny's result \cite[Theorem 3.1]{cherny2008pricing} concerning the absence of good-deals of the second kind for coherent risk measures.  Good-deals of the second kind is a concept that was first explored by Jaschke and K{\"u}chler \cite{jaschke2001coherent} in a topological framework.  They proved that its absence is equivalent to the existence of a \emph{consistent price system}, and their results were generalised by Staum \cite{staum2004fundamental}.  \cite{jaschke2001coherent} and \cite{staum2004fundamental} are theoretical, and unlike \cite{cherny2008pricing}, are difficult to apply.  On the other hand, Theorem \ref{thm: no strong reg arb} (which leads to Theorem \ref{cor:strong rho arb pricing bounds}) can be used in practice, and also holds for convex risk measures.   

%This was first explored by Jaschke and K{\"u}chler \cite{jaschke2001coherent}, who studied the situation where $D=\{ X \in L : \rho(X) < 0 \}$ and $\rho$ is a coherent risk measure.  They worked in a topological framework and proved that the absence of good-deals is equivalent to the existence of a \emph{consistent price system}.  Their results were generalised by Staum \cite{staum2004fundamental}, and sharpened by Cherny \cite{cherny2008pricing} who provided explicit no-good-deal price bounds for $Y$ in terms of \emph{risk neutral measures} in a less abstract setting. 
\end{remark}

We next consider the dual characterisation of $\rho$-arbitrage.

\begin{theorem}
\label{thm: no reg arb equivalence}
Suppose that Condition I is satisfied, $\rho$ satisfies sensitivity to large expected losses on $L$ and $\tilde \cQ_{\max}^\alpha \neq \emptyset$. Then the following are equivalent:
\begin{enumerate}
    \item The market $(S^0,S)$ does not admit $\rho$-arbitrage.
    
     %\item \NK{The market $(S^0,S)$ does not admit $\rho^\infty$-arbitrage.}
    
    \item There exists $\tilde \cQ^\alpha \subset \cQ^\alpha$ satisfying Conditions POS, MIX and INT, such that $\tilde{\mathcal{Q}}^\alpha \cap \cP \neq \emptyset$.
    
    \item For every nonempty $\tilde \cQ^\alpha \subset \cQ^\alpha$ satisfying Conditions POS, MIX and INT, we have  $\tilde{\mathcal{Q}}^\alpha \cap \cP \neq \emptyset$.
\end{enumerate}
\end{theorem}

\begin{remark}
\label{rmk:applying dual char reg arb}
(a) Usually (at least in all the examples we consider) there is an ``interior'' of $\cQ$ which contains $1$ (the real world measure).  This implies $\tilde{\cQ}_{\max}^\alpha \neq \emptyset$.  Furthermore, by \cite[Proposition 4.11]{herdegen2020dual} it follows that $\rho^\infty$ is strictly expectation bounded and so by Remark \ref{rmk:rho WSLL iff rho bar WSLL}(b), $\rho$ is sensitive to large expected losses on the entire space $L$.  In such cases, we only need to check when Condition I holds in order to apply Theorem \ref{thm: no reg arb equivalence}.  

(b) When $\rho$ is coherent, Theorem \ref{thm: no reg arb equivalence} reduces to \cite[Theorem 4.20]{herdegen2020dual}.  By Remarks \ref{rmk:rho WSLL iff rho bar WSLL}(a) and \ref{rmk:rho arb scalable acceptable deal relation}(b), this is related to Arduca and Munari's result \cite[Theorem 4.14]{arduca2020fundamental} on the absence of good-deals of the first kind.  They work in a one-period model that accounts for frictions and derive a fundamental theorem of asset pricing, which just like Theorem \ref{thm: no reg arb equivalence}, subsumes the classical fundamental theorem of asset pricing.  This is the only result in the market consistent pricing (good-deal pricing) literature we are aware of that does this.  They relate the absence of good-deals with the existence of a \emph{strictly consistent price deflator}, however, it is difficult to apply this in practice.  Theorem \ref{thm: no reg arb equivalence} (which leads to Theorem \ref{cor:rho arb pricing bounds}) tells us \emph{exactly} which strictly consistent price deflators (EMMs in the frictionless case) to consider. 

%The only notion of good deals (aside from $\rho$-arbitrage) we are aware of that truly subsumes classical arbitrage is the recent work by Arduca and Munari \cite{arduca2020fundamental}.  They work in a one-period model that accounts for frictions and derive a fundamental theorem of asset pricing for \emph{pointed convex cones} $D$ that contain $L_+$.  Their main result relates the absence of good-deals with the existence of a \emph{strictly consistent price deflator}.  It holds a strong connection with \cite[Theorem 4.20]{herdegen2020dual}; cf.~Remark \ref{rmk:applying dual char reg arb}.
\end{remark}

\subsection{Dual characterisation of (strong) $\rho$-consistent pricing}

As a consequence of Theorem \ref{thm: no strong reg arb}/Theorem \ref{thm: no reg arb equivalence}, one can obtain price bounds for a financial contract $X$ that lies \emph{outside} a given market based on the absence of strong $\rho$-arbitrage/$\rho$-arbitrage.  This involves running through the ACMMs that lie in $\cQ^{\overline{\textnormal{co}}\,\alpha}$/the EMMs that lie in an ``interior''  $\tilde{\cQ}^\alpha$  of $\mathcal{Q}^\alpha$ (note that the choice of ``interior'' does not matter) and taking the corresponding discounted expectation.  This is what the next two results show.  %In particular, since $\cQ^{\overline{\textnormal{co}}\,\alpha}$/$\tilde{\cQ}^\alpha$ is not difficult to compute in general (cf.~Section \ref{sec:Applications}), this means that $I_\rho^s(X)$/$I_\rho(X)$ can be computed \emph{explicitly} once $\mathcal{M}$/$\mathcal{P}$ is known. 

\begin{theorem}
\label{cor:strong rho arb pricing bounds}
Suppose $\mathcal{Q}^\alpha$ is uniformly integrable and $1 \in \mathcal{Q}^\alpha$.  Let
\begin{equation*}
    \tilde{L}:=\{Y \in L^1 : \lim_{a \to \infty} \sup_{Z \in \mathcal{Q}^\alpha} \mathbb{E}[Z|Y|\mathds{1}_{|Y|>a}]=0\}
\end{equation*}
and $(S^0,S)$ be a $(1+d)$-dimensional market with returns in $\tilde{L}$.  Then, for \emph{any} financial contract $X \in \tilde{L}$, we have
\begin{equation}
\label{eq:no strong rho arb price bounds}
    I^s_{\rho}(X) = \{\mathbb{E}[ZX/(1+r)] : Z \in \mathcal{Q}^{\overline{\textnormal{co}}\,\alpha} \cap \mathcal{M} \}
\end{equation}
where $\mathcal{M}$ is the set of ACMMs for the original market.
\end{theorem}

\begin{theorem}
\label{cor:rho arb pricing bounds}
Suppose $\rho$ satisfies sensitivity to large expected losses and admits a dual representation \eqref{eq:dual_char_of_rho}, where $\emptyset \neq \tilde{\mathcal{Q}}^\alpha \subset \mathcal{Q}^\alpha$ satisfies Conditions POS, MIX and INT.  Let
\begin{equation*}
    \tilde{L}:=\{Y \in L^1 : ZY \in L^1 \textnormal{ for all } Z \in \mathcal{Q}^\alpha \}
\end{equation*}
and $(S^0,S)$ be a $(1+d)$-dimensional market with returns in $\tilde{L}$.  Then, for \emph{any} financial contract $X \in \tilde{L}$, we have
\begin{equation}
\label{eq:no rho arb price bounds}
    I_{\rho}(X) = \{\mathbb{E}[ZX/(1+r)] : Z \in \tilde{\mathcal{Q}}^\alpha \cap \mathcal{P} \}
\end{equation}
where $\mathcal{P}$ is the set of EMMs for the original market.
\end{theorem}

\begin{remark}
While there is a strong link between good-deals of the first (second) kind and (strong) $\rho$-arbitrage, cf.~Remark \ref{rem:rho consistent pricing}(c), let us stress that -- unlike the extant literature on good-deal pricing, which often yields only abstract price intervals --  (strong) $\rho$-consistent price bounds are \emph{explicitly computable} once ($\mathcal{M}$) $\mathcal{P}$ is known since $(\mathcal{Q}^{\overline{\textnormal{co}}\,\alpha})$ $\tilde{\mathcal{Q}}^\alpha$ is not difficult to find in general, cf.~Section \ref{sec:Applications}.
\end{remark}

\section{Examples}
\label{sec:Applications}

In this section, we apply our theory to various examples.  Our main focus is on convex risk measures that are not coherent since the latter have been discussed in \cite[Section 5]{herdegen2020dual}.  We do not make any assumptions on the returns, other than our standing assumptions that they are contained in a Riesz space and that the market is nonredundant and nondegenerate.

\subsection{Risk functionals based on loss functions}
\label{subsec:loss function risk measures}

The examples in this section are based around the theme of \emph{loss functions}.  

\begin{definition}
\label{def:loss function}
A function $l:\mathbb{R} \to \RR$ is called a \emph{loss function}  if it is nondecreasing, convex, $l(0)=0$ and $l(x) \geq x$ for all $x \in \RR$. 
\end{definition}

A loss function $l$ reflects how risk averse an agent is, and so it is natural to assume that it is nondecreasing and $l(0)=0$.  The assumption $l(x) \geq x$ means that compared to the risk neutral evaluation, there is more weight on losses and less on gains.  Finally, convexity of $l$ encodes that diversified positions are less risky than concentrated ones.  

The growth rate of $l$ will play an important role in the sequel.  To that end, we let 
\begin{equation*}
    a_l:= \lim_{x \to -\infty} \frac{l(x)}{x} \quad \textnormal{and} \quad b_l:= \lim_{x \to \infty} \frac{l(x)}{x}.
\end{equation*}
Note that $0 \leq a_l \leq 1 \leq b_l \leq \infty$, where $a_l < b_l$ unless $a_l = b_l = 1$, in which case $l$ is the identify function.  We will also repeatedly make use of the fact $(a_l, b_l) \subset \textnormal{dom} \, l^* \subset [a_l, b_l]$, where $l^*$ is the convex conjugate of $l$.  In particular, this means $l^*$ is bounded on any compact subset of $(a_l, b_l)$.

\subsubsection{Expected weighted loss}
The expected weighted loss of $X \in H^{\Phi_l}$ with respect to a loss function $l$ is given by
\begin{equation*}
    \textnormal{EW}^l(X) := \mathbb{E}[l(-X)]
\end{equation*}
where $H^{\Phi_l}$ is the Orlicz heart corresponding to the Young function $\Phi_l:=l|_{\RR_{+}}$.  By the properties of $l$ and the definition of $H^{\Phi_{l}}$, $\textnormal{EW}^l$ is a real-valued convex risk functional (but never cash-invariant unless $l(x)=x$). It is also not difficult to check that it  satisfies the Fatou property on $H^{\Phi_{l}}$.  Therefore, by  Corollary \ref{cor:rho infinity WC}, Theorem \ref{prop:A rho infinity larger than L+} and Proposition \ref{prop:EW SSLL characterisation}, we have the following.

\begin{corollary}
Let $l$ be a loss function for which $a_l = 0$ or $b_l = \infty$.  Assume the market $(S^0,S)$ has returns in $H^{\Phi_{l}}$. The following are equivalent:
\begin{enumerate}
\item The market $(S^0,S)$ does not admit $\textnormal{EW}^l$-arbitrage.
\item The market $(S^0,S)$ does not admit arbitrage of the first kind.
\end{enumerate}
Moreover, if $a_l > 0$ and $b_l < \infty$, then there exists a market with returns in $H^{\Phi_{l}} = L^1$ that admits $\textnormal{EW}^l$-arbitrage, but does not admit arbitrage of the first kind.
\end{corollary}

\begin{remark}
One can check that this result (including Proposition \ref{prop:EW SSLL characterisation}) extends to functions $l:\RR \to \RR$ that are nondecreasing, convex, and satisfy $l(0) = 0$ as well as $\lim_{x \to \infty} l(x) = \infty$ (which is weaker than $l(x) \geq x$ for all $x \in \RR$). %\NK{Actually nondecreasing convex and nonconstant implies $\lim_{x \to \infty} l(x) = \infty$.}
\end{remark}

\subsubsection{Shortfall risk measures}
Shortfall risk measures were first introduced as a case study on $L^\infty$ by F{\"o}llmer and Schied in \cite[Section 3]{follmer2002convex}.  Here, we work on Orlicz hearts.  To that end, let $l$ be a loss function and define the acceptance set
\begin{equation*}
    \mathcal{A}_l := \{ X \in H^{\Phi_{l}} : \textnormal{EW}^l(X) \leq 0 \},
\end{equation*}
where $H^{\Phi_{l}}$ is the Orlicz heart corresponding to the Young function $\Phi_l:=l|_{\RR_+}$.  Then the shortfall risk measure associated with $l$ is given by $\textnormal{SR}^l:H^{\Phi_{l}} \to (-\infty,\infty]$ where
\begin{equation*}
    \textnormal{SR}^l(X) := \inf \{ m \in \RR : X+m \in \mathcal{A}_l \} = \inf \{ m \in \RR : \textnormal{EW}^l(X+m) \leq 0 \}.
\end{equation*}
This is a convex risk measure that satisfies the Fatou property.  It can be interpreted as the cash-invariant analogue of $\textnormal{EW}^l$ in the sense that it is cash-invariant and $\textnormal{SR}^l(X) \leq 0$ if and only if $\textnormal{EW}^l(X) \leq 0$.  
It admits a dual representation, which we now recall. 

\begin{proposition}
\label{prop:dual rep SR}
Let $l$ be a loss function and $l^*$ its convex conjugate.  Then for $X \in H^{\Phi_l}$
\begin{equation}
\label{eq:dual rep SR}
    \textnormal{SR}^l(X) = \sup_{Z \in \cQ^{\alpha^l}} \{\mathbb{E}[-ZX] - \alpha^l(Z)\}, \quad \textnormal{where} \quad \alpha^l(Z) := \inf_{\lambda > 0} \tfrac{1}{\lambda} \mathbb{E}[l^*(\lambda Z)]
\end{equation}
and $\cQ^{\alpha^l} = \{ Z \in \cD : \textnormal{there exists } \lambda > 0 \textnormal{ such that } \mathbb{E}[l^*(\lambda Z)] < \infty \}$.
\end{proposition}

The characterisation of $\textnormal{SR}^l$-arbitrage depends on the constants $a_l$ and $b_l$.  There are two cases.  Firstly, when  $a_l= 0$ or $b_l = \infty$, then since $\textnormal{dom} \; l^* \supset (a_l, b_l)$, Proposition \ref{prop:dual rep SR} implies
\begin{align*}
    \mathcal{Q}^{\alpha^l} \supset \{ Z \in \mathcal{D} \cap L^\infty : \textnormal{there exists $\epsilon > 0$ such that } Z \geq \epsilon \ \mathbb{P}\textnormal{-a.s.} \}.
\end{align*}
As a consequence of this, \eqref{eq:dual char of bar rho} and \cite[Proposition C.6]{herdegen2020dual}, we have $(\textnormal{SR}^l)^\infty \equiv \textnormal{WC}$ in this case.  Whence, Corollary \ref{cor:rho infinity WC} can be applied.

\begin{corollary}
\label{cor:SRl arb equiv ordinary arb}
    Let $l$ be a loss function where $a_l = 0$ or $b_l = \infty$.  Assume the market $(S^0,S)$ has returns in $H^{\Phi_{l}}$.  The following are equivalent:
    \begin{enumerate}
        \item The market $(S^0,S)$ does not admit $\textnormal{SR}^l$-arbitrage.
        \item The market $(S^0,S)$ does not admit arbitrage of the first kind.
    \end{enumerate}
\end{corollary}

When $a_l > 0$ and $b_l < \infty$, as $(a_l, b_l) \subset \textnormal{dom} \, l^* \subset [a_l, b_l]$, we have
\begin{align*}
    \{ Z \in \mathcal{D} : \textnormal{there exists $\lambda > 0$ such that } a_l < \essinf \lambda Z \textnormal{ and } \esssup \lambda Z < b_l \} & \subset \cQ^{\alpha^l} \quad \textnormal{and} \\ \{ Z \in \mathcal{D} : \textnormal{there exists $\lambda > 0$ such that } a_l \leq \lambda Z \leq b_l \ \mathbb{P}\textnormal{-a.s.}\} & \supset \mathcal{Q}^{\alpha^l}.
\end{align*}
Conditions I and UI are satisfied since $\lVert Z \rVert_{\infty} \leq b_l/a_l$ for $Z \in \mathcal{Q}^{\alpha^l}$. And the dual characterisation of (strong) SR$^l$-arbitrage follows from Theorem \ref{thm: no strong reg arb}, Proposition \ref{prop:interior SR} and Proposition \ref{prop:Q co alpha l SRl}. %\footnote{If $Z \in \mathcal{Q}^{\alpha^l}$, then there exists $k > 0$ such that $kZ \in [a_l,b_l] \ \mathbb{P}$-a.s.  Clearly $Z \in L^\infty$.  If $\lVert Z \rVert_{\infty} > b_l/a_l$, then in order for $kZ \leq b_l \ \mathbb{P}$-a.s., it must be that $k < a_l$.  But then since $Z \in \mathcal{D}$, $\mathbb{P}[Z \leq 1] > 0$ so this contradicts the fact $kZ \geq a_l \ \mathbb{P}$-a.s.  Thus, every $Z \in \mathcal{Q}^{\alpha^l}$ satisfies $\lVert Z \rVert_{\infty} \leq b_l/a_l$.}  %And that $\textnormal{SR}^l(X) = \mathbb{E}[-X]$ when $a_l = b_l$.

\begin{corollary}
\label{cor:SRl arbitrage}
    Let $l$ be a loss function, and assume that $a_l > 0$ and $b_l < \infty$.  Assume the market $(S^0,S)$ has returns in $H^{\Phi_{l}}$.
    \begin{enumerate}
        \item The market $(S^0,S)$ does not admit $\textnormal{SR}^l$-arbitrage if and only if there exists $Z \in \mathcal{P}$ such that $a_l + \epsilon < \lambda Z < b_l - \epsilon \ \mathbb{P}$-a.s.~for some $\lambda,\epsilon > 0$.
        \item The market $(S^0,S)$ does not admit strong $\textnormal{SR}^l$-arbitrage if and only if there exists $Z \in \mathcal{M}$ such that $\mathbb{E}[l^*(\lambda Z)] < \infty$ for some $\lambda > 0$. 
    \end{enumerate}
\end{corollary}

\begin{remark}
\label{rem:SR}
(a) All of the above results for shortfall risk measures hold for functions $l:\RR \to \RR$ that are nondecreasing, convex and satisfy $l(0) = 0$ and $l(x) > 0$ for all $x > 0$. %where the last property is required so that $\textnormal{SR}^l$ is normalised. 
 
(b) Using numerical examples, it was shown in \cite[Section 5]{giesecke2008measuring} that shortfall risk measures corresponding to functions of the form $l(x) = cx^\alpha\mathds{1}_{\{x>0\}}$ where $\alpha > 1$ and $c > 0$ ``adequately account for event risk''.  In light of part (a) above, Corollary \ref{cor:SRl arb equiv ordinary arb} reinforces this.
\end{remark}

\subsubsection{OCE risk measures}
\label{subsec:OCE examples}

Optimised Certainty Equivalents (OCEs) were first introduced by Ben-Tal and Teboulle \cite{ben1986expected} and later linked to risk measures on $L^\infty$ by the same authors in \cite{ben2007old}. We follow \cite[Section 5.4]{cheridito2009risk} and define the OCE risk measure associated with a loss function $l$ as the map $\OCE^{l}: H^{\Phi_{l}} \to \RR$,
\begin{equation}
\label{eq:OCE definiton}
    \OCE^{l}(X) := \inf_{\eta \in \RR} \{ \mathbb{E}[l(\eta - X)] - \eta \},
\end{equation}
where $H^{\Phi_{l}}$ is the Orlicz heart corresponding to the Young function $\Phi_l:=l|_{\mathbb{R}_+}$.  By \cite[Section 5.1]{cheridito2009risk} (with $V \equiv \textnormal{EW}^l$), $\OCE^l$ is the largest real-valued convex risk measure on $H^{\Phi_{l}}$ that is dominated by $\textnormal{EW}^l$.\footnote{More generally, cash-invariant hulls of convex functionals have been studied by \cite{filipovic2007monotone,krokhmal2006modeling}.}  By \cite[Theorem 4.3]{cheridito2009risk}, it also satisfies the Fatou property on $H^{\Phi_{l}}$.  Like shortfall risk measures, OCE risk measures admit a dual representation.   

\begin{proposition}
\label{prop:OCE dual rep}
Let $l$ be a loss function and $l^*$ its convex conjugate.  Then for $X \in H^{\Phi_l}$
\begin{equation}
\label{eq:dual char OCE}
    \OCE^l(X) = \sup_{Z \in \cQ^{\alpha^l}} \{\mathbb{E}[-ZX] - \alpha^l(Z)\}, \quad \textnormal{where} \quad  \alpha^l(Z) := \mathbb{E}[l^*(Z)]
\end{equation}
and $\cQ^{\alpha^l} = \{ Z \in \cD : \mathbb{E}[l^*(Z)] < \infty \}$.
\end{proposition}

\begin{remark}
\label{rmk:OCE condition}
Normalisation of $\OCE^l$ is equivalent to $l(x) \geq x$ for all $x \in \RR$.  If $l(x) > x$ for all $x$ with $|x|$ sufficiently large, then $\lim_{|x| \to \infty} (l(x) - x) = \infty$ by convexity of $x$, and the infimum in \eqref{eq:OCE definiton} is attained; cf.~\cite[Lemma 5.2]{cheridito2009risk}.  However, if $l(x) = x$ for either $x \geq 0$ or $x \leq 0$, then the infimum is not necessarily attained, and it is easy to check $\OCE^l(X) = \mathbb{E}[-X]$ for $X \in H^{\Phi_l}$.

%(b) Shortfall risk measures and OCE risk measures are intimately linked. Indeed, combining \eqref{eq:dual rep SR} and \eqref{eq:dual char OCE} gives
%\begin{align*}
%    \textnormal{SR}^l(X) &= \sup_{Z \in \mathcal{D}} \{\mathbb{E}[-ZX] - \inf_{\lambda > 0} \tfrac{1}{\lambda}\mathbb{E}[l^*(\lambda Z)]\} =  \sup_{Z \in \mathcal{D}} \sup_{\lambda > 0}\{\mathbb{E}[-ZX] - \tfrac{1}{\lambda}\mathbb{E}[l^*(\lambda Z)]\} \\
%&=\sup_{\lambda > 0}\sup_{Z \in \mathcal{D}} \{\mathbb{E}[-ZX] - \tfrac{1}{\lambda}\mathbb{E}[l^*(\lambda Z)]\} =
%\sup_{\lambda > 0} \sup_{Z \in \mathcal{D}} \{ \mathbb{E}[-ZX] - \mathbb{E}[(l_\lambda)^*(Z)] \} \\
%&= \sup_{\lambda > 0} \OCE^{l_\lambda}(X),
%\end{align*}
%where $l_\lambda = l/\lambda$ for $\lambda > 0$. This shows that shortfall risk measures can be understood as the supremum of certain OCE risk measures.
%
%(b) If the condition $l(x) > x$ for all $x$ with $|x|$ sufficiently large is not satisfied, then $\OCE^l(X) =\mathbb{E}[-X]$. Indeed, convexity of $l$ together with the fact $l(0)=0$ implies that $l(x)=x$ for all $x \in \RR_+$ or $l(x)=x$ for all $x \in \RR_{-}$.  In either case, it follows that $\OCE^l(X) \leq  \mathbb{E}[-X]$ for $X \in L^\infty$ by choosing $\eta$ sufficiently small or large. By the Fatou property, $\OCE^l(X) \leq  \mathbb{E}[-X]$ for all $X \in H^{\Phi_{l}}$, and the claim follows from expectation boundedness of $\OCE^l(X)$.
\end{remark}

Just like for shortfall risk measures, the characterisation of $\textnormal{OCE}^l$-arbitrage depends on the constants $a_l \geq 0$ and $b_l \leq \infty$.  There are two cases.  Firstly, when  $a_l = 0$ and $b_l = \infty$, then $\textnormal{dom} \, l^* \supset (0,\infty)$ and it follows that
\begin{align*}
    \mathcal{Q}^{\alpha^l} \supset \{ Z \in \mathcal{D} \cap L^\infty : \textnormal{there exists $\epsilon > 0$ such that } Z \geq \epsilon \ \mathbb{P}\textnormal{-a.s.} \}.
\end{align*}
As a consequence of this, \eqref{eq:dual char of bar rho} and \cite[Proposition C.6]{herdegen2020dual}, it follows that $(\textnormal{OCE}^l)^\infty \equiv \textnormal{WC}$ in this case.  Whence, Corollary \ref{cor:rho infinity WC} can be applied.

\begin{corollary}
\label{cor:OCEl arb equiv ordinary arb}
    Let $l$ be a loss function where $a_l = 0$ and $b_l = \infty$.  Assume the market $(S^0,S)$ has returns in $H^{\Phi_{l}}$.  The following are equivalent:
    \begin{enumerate}
        \item The market $(S^0,S)$ does not admit $\textnormal{OCE}^l$-arbitrage.
        \item The market $(S^0,S)$ does not admit arbitrage of the first kind.
    \end{enumerate}
\end{corollary}

When $a_l > 0$ or $b_l < \infty$, we can derive a dual characterisation of (strong) $\OCE^l$-arbitrage.  By Remark \ref{rmk:OCE condition}, it suffices to consider the case $a_l < 1 < b_l$.

\begin{corollary}
\label{cor:OCE:SRA}
Let $l$ be a loss function and assume that either $a_l > 0$ or $b_l < \infty$, and $a_l < 1 < b_l$. Assume the market $(S^0,S)$ has returns in $H^{\Phi_{l}}$. Then,
\begin{enumerate}
    \item The market $(S^0,S)$ does not admit $\OCE^l$-arbitrage if and only if there exists $Z \in \mathcal{P}$ such that $\mathbb{E}[l^*(Z)] < \infty$ and $a_l+\epsilon < Z < b_l - \epsilon \ \mathbb{P}$-a.s.~for some $\epsilon > 0$.
    \item If in addition $b_l < \infty$, the market $(S^0,S)$ does not admit strong $\OCE^l$-arbitrage if and only if there exists $Z \in \mathcal{M}$ such that $\mathbb{E}[l^*(Z)] < \infty$. 
\end{enumerate}
\end{corollary}

\subsection{Adjusted risk functionals}
\label{subsec:adjusted ES}

Suppose we have a family of risk functionals $(\rho_\alpha)_{\alpha \in \mathcal{I}}$ on a Riesz space $L^\infty \subset L \subset L^1$ indexed by a set $\mathcal{I}$.  Let $g:\mathcal{I} \to [0,\infty]$ be a function such that $\inf g = 0$.  Then the functional defined by 
\begin{equation*}
    \rho^g(X):= \sup_{\alpha \in \mathcal{I}} \{ \rho_\alpha(X) - g(\alpha) \}, \quad X \in L,
\end{equation*}
is also a risk functional.\footnote{Here, we tacitly assume that the supremum is only taken over those $\alpha \in \mathcal{I}$ for which $g(\alpha) < \infty$.}  It is cash-invariant/convex/satisfies the Fatou property whenever $\rho_\alpha$ is cash-invariant/convex/satisfies the Fatou property for $\alpha \in \mathcal{I}$.  The way to interpret this \emph{g-adjusted risk functional}, is to look at its acceptance set.  Indeed, $X \in \cA_{\rho^g}$ if and only if $\rho_\alpha(X) \leq g(\alpha)$ for all $\alpha \in \mathcal{I}$.  Thus, whilst the risk of a random variable $X$ is ultimately represented by a single number $\rho^g(X)$, whether or not $X$ is acceptable depends on the entire continuum of values $(\rho_\alpha(X))_{\alpha \in \mathcal{I}}$.  In this way, $g$ can be considered as a \emph{target risk profile}.% and we achieve greater control.

When $\mathcal{I}$ is equipped with a partial order, $\rho_\alpha \geq \rho_{\alpha'}$ for $\alpha \leq \alpha'$ and $\rho_\alpha$ is cash-invariant for all $\alpha \in \mathcal{I}$, then we may assume without loss of generality that $g$ is nonincreasing.  Indeed, otherwise, we may replace $g$ by $\tilde{g}: \mathcal{I} \to [0,\infty]$ where
\begin{equation*}
    \tilde{g}(\alpha) := \inf \{ g(\alpha') : \alpha' \leq \alpha \},
\end{equation*}
and note that the acceptance set of $\rho^g$ coincides with the acceptance set of $\rho^{\tilde{g}}$, which implies that $\rho^g = \rho^{\tilde{g}}$ since they are cash-invariant.  We will consider the cases where $\rho_\alpha$ is given by either Value at Risk or Expected Shortfall.

%In the sequel, we will consider the cases where $\rho_\alpha$ is given by either Value at Risk or Expected Shortfall at confidence level $\alpha \in (0,1]$.\footnote{Recall that VaR$^\alpha$ and ES$^\alpha$ were defined in Section 2.1 for $\alpha \in (0,1)$.  Extending the definitions to $\alpha = 1$, yields that VaR$^1$ is the \emph{best-case} risk measure $X \mapsto \essinf(-X)$, and $\ES^1$ is the expected-loss risk measure $X \mapsto \mathbb{E}[-X]$.} To that end, we let $\cG$ be the set of all nonincreasing functions $g:(0,1] \to [0,\infty]$ with $g(1) = 0$ and $\{1\} \subsetneq \textnormal{dom} \, g$.  Moreover, for later use, we set $\cG_\beta  := \{g \in \cG : \inf \textnormal{dom\,$g$} = \beta \}$ for $\beta \in [0,1)$.

\subsubsection{Adjusted Value at Risk}

For a nonincreasing function $g:(0,1) \to [0,\infty]$ with $\inf g  = 0$,  we define the \emph{$g$-adjusted Value at Risk} as the map $\VaR^g: L^1 \to (-\infty,\infty]$ given by
\begin{equation*}
    \VaR^g(X) := \sup_{\alpha \in (0,1)}\{\VaR^\alpha(X) - g(\alpha)\}.%, \quad \textnormal{where} \quad  \VaR^1(X):=\lim_{\alpha \to 1} \VaR^\alpha(X) = \essinf(-X).
\end{equation*}  

This is a family of risk measures that are neither convex nor positively homogeneous (unless the function $g$ is constant on its effective domain). It was introduced by Bignozzi, Burzoni and Munari \cite{bignozzi2020risk}.\footnote{Our definition is based on \cite[Proposition 4]{bignozzi2020risk}. There, they consider nondecreasing functions that are left-continuous (since these are the properties of the generalised inverse of a \emph{benchmark loss distribution}.  However, in line with the way we defined VaR, the functions $g$ must be nonincreasing for us. They also do not necessarily have to be left-continuous.}  By Corollary \ref{cor:rho infinity WC}, Theorem \ref{prop:A rho infinity larger than L+} and Proposition \ref{prop:VaRg acceptance set}, we have the following result concerning $\VaR^g$-arbitrage.

\begin{proposition}
Assume the market $(S^0,S)$ has returns in $L^1$.  Let $g:(0,1) \to [0,\infty]$ be a nonincreasing function with $\inf g = 0$.  If $g$ is real-valued, then the following are equivalent:
\begin{enumerate}
    \item The market $(S^0,S)$ does not admit $\VaR^g$-arbitrage.
\item The market $(S^0,S)$ does not admit arbitrage of the first kind.
\end{enumerate}
If $g$ is not real-valued and the probability space is atomless, then there exists a market with returns in $L^1$ that admits strong $\VaR^g$-arbitrage, but does not admit arbitrage of the first kind.
\end{proposition}

\subsubsection{Adjusted Expected Shortfall}

%In order to enhance how tail risk is captured, Burzoni, Munari and Wang \cite{burzoni2020adjusted} recently developed a new class of risk measures, which builds on ES.  To introduce this class, let $\cG$ be the set of all nonincreasing functions $g:(0,1] \to [0,\infty]$ with $g(1) = 0$ and $\{1\} \subsetneq \textnormal{dom\,$g$}$.\footnote{Our definition of $g$-adjusted Expected Shortfall is based on \cite[Proposition 2.2]{burzoni2020adjusted}, which considers nondecreasing functions that are not identically $\infty$.  However, in line with the way we defined ES, the functions $g$ must be nonincreasing for us.  We assume $g(1)=0$ to achieve normalisation. But this is without loss of generality since otherwise, we simply replace $g(\cdot)$ by $g(\cdot)-g(1)$, leaving identical preference orders (see Definition \ref{def:rho preferred}). Moreover, the case $\textnormal{dom\,$g$} = \{1\}$ corresponds to the expected loss risk measure $X \mapsto \mathbb{E}[-X]$ and is not interesting.}

Let $\cG$ be the set of nonincreasing functions $g:(0,1] \to [0,\infty]$ with $g(1) = 0$ and $\{1\} \subsetneq \textnormal{dom} \, g$.  For $g \in \cG$, we define the \emph{$g$-adjusted Expected Shortfall} as the map  $\ES^g:L^1 \to (-\infty,\infty]$, given by
\begin{equation*}
    \ES^g(X) := \sup_{\alpha \in (0,1]}\{\ES^\alpha(X) - g(\alpha)\}, \quad \textnormal{where} \quad \ES^1(X):=\lim_{\alpha \to 1} \ES^\alpha(X) = \mathbb{E}[-X].
\end{equation*}  

This is a family of convex risk measures introduced by Burzoni, Munari and Wang \cite{burzoni2020adjusted}.\footnote{Our definition is based on \cite[Proposition 2.2]{burzoni2020adjusted}. There, they consider nondecreasing functions.  However, in line with the way we defined ES, the functions $g$ must be nonincreasing for us.  The case $\textnormal{dom\,$g$} = \{1\}$ corresponds to the expected-loss risk measure and is not interesting.}    We proceed to state the dual representation of $g$-adjusted ES. To this end, for $\beta \in [0,1)$ set $\cG_\beta  := \{g \in \cG : \inf \textnormal{dom\,$g$} = \beta \}$.
%\begin{equation*}
%	\cG_\beta  := \{g \in \cG : \inf \textnormal{dom\,$g$} = \beta \}. %\quad \textnormal{and} \quad \cG_\beta^\infty := \{g \in \cG_\beta : g \textnormal{ is bounded on its effective domain} \}.
%\end{equation*}

\begin{proposition}
\label{prop:dual rep g adjusted ES}
Let $g \in \cG$. Then $\ES^g: L^1 \to (-\infty, \infty]$ satisfies the dual representation
\begin{equation*}
    \ES^g(X) = \sup_{Z \in \cQ^{\alpha^g}}\{\mathbb{E}[-ZX] - g(\lVert Z \rVert_{\infty}^{-1})\}
\end{equation*}
where the penalty function $\alpha^g:\cD \to [0,\infty]$ is given by $\alpha^g(Z) = g(\lVert Z \rVert_{\infty}^{-1})$ if $Z \in \cD \cap L^\infty$ and $\alpha^g(Z) = \infty$ otherwise.\footnote{Note that $\alpha^g$ is in general only quasi-convex since $g$ is nonincreasing.} Moreover, $\cQ^{\alpha^g}=\{Z \in \cD \cap L^\infty : g(\lVert Z \rVert_{\infty}^{-1})  < \infty \}$ is convex and satisfies %\NK{need $g(\beta) < \infty$ too?}
\begin{equation}
	\label{eq:dual set g adjusted ES}
	\cQ^{\alpha^g} =  \begin{cases}
		\cD \cap L^\infty,&\text{if $g \in \cG_0$}, \\ \{Z \in \cD: \lVert Z \rVert_{\infty} \leq \tfrac{1}{\beta}\},&\text{if $g \in \cG_\beta$, $\beta \in (0,1)$ and $g(\beta) < \infty$},  \\ \{Z \in \cD: \lVert Z \rVert_{\infty} < \tfrac{1}{\beta}\},&\text{if $g \in \cG_\beta$, $\beta \in (0,1)$ and $g(\beta) = \infty$}.
	\end{cases}
\end{equation}
\end{proposition}

When $g \in \mathcal{G}_0$, we have that $(\ES^g)^\infty \equiv \WC$ by \eqref{eq:dual char of bar rho}, \eqref{eq:dual set g adjusted ES} and \cite[Proposition C.6]{herdegen2020dual}.  Therefore, we may apply Corollary \ref{cor:rho infinity WC}.

\begin{corollary}
\label{cor:ESg arb equiv ordinary arb}
   Let $g \in \mathcal{G}_0$ and assume the market $(S^0,S)$ has returns in $L^1$.  The following are equivalent:
    \begin{enumerate}
        \item The market $(S^0,S)$ does not admit $\textnormal{ES}^g$-arbitrage.
        \item The market $(S^0,S)$ does not admit arbitrage of the first kind.
    \end{enumerate}
\end{corollary}

We can further provide a dual characterisation of (strong) $\ES^g$-arbitrage when $g \in \mathcal{G}_\beta$ and $\beta \in (0,1)$.  In this case, since $\mathcal{Q}^{\alpha^g}$ is $L^\infty$-bounded, Conditions I and UI are both satisfied if the returns lie in $L^1$.  Moreover,  it is not difficult to check that
\begin{equation*}
    \tilde{\cQ}^{\alpha^g} =\{ 0 < Z \in \cD: \lVert Z \rVert_{\infty} < 1/\beta\}
\end{equation*}
is a subset of $\mathcal{Q}^{\alpha^g}$ that satisfies Conditions POS, MIX and INT and contains $1$; see~\cite[Proposition B.6]{herdegen2020dual} for details. Finally, Proposition \ref{prop:Q g star star} shows that 
\begin{equation*}
    \cQ^{\overline{\textnormal{co}}\,\alpha^g} = \begin{cases}
   \{Z \in \cD: \lVert Z \rVert_{\infty} \leq \tfrac{1}{\beta}\},&\text{if $g \in  \cG_\beta^\infty$}, \\ \{Z \in \cD: \lVert Z \rVert_{\infty} < \tfrac{1}{\beta}\},&\text{if $g \in \cG_\beta \setminus \cG_\beta^\infty$},
   \end{cases}
\end{equation*}
where $\cG_\beta^\infty := \{g \in \cG_\beta : g \textnormal{ is bounded on its effective domain} \}$.  Thus, Theorems \ref{thm: no strong reg arb} and \ref{thm: no reg arb equivalence} yield the following result.

\begin{corollary}
Let $g \in \mathcal{G}_\beta$ where $\beta \in (0,1)$ and assume the market $(S^0,S)$ has returns in $L^1$.
\begin{enumerate}
    \item $(S^0,S)$ does not admit $\ES^g$-arbitrage if and only if there exists $Z \in \mathcal{P}$ such that $\left\Vert Z \right\Vert_\infty < \frac{1}{\beta}$. 
    \item When $g \in \cG_\beta^\infty$ ($g \in \cG_\beta \setminus \cG_\beta^\infty$), $(S^0,S)$ does not admit strong $\ES^g$-arbitrage if and only if there exists $Z \in \mathcal{M}$ with $\left\Vert Z \right\Vert_\infty \leq (<) \frac{1}{\beta}$. 
\end{enumerate}
\end{corollary}

\begin{remark}
	\label{rem:strong rho bar arbitrage}
This result shows that the implication ``(c) $\implies$ (b)'' in Theorem~\ref{prop:strong reg arb:first characterisation} does not hold. Indeed, if  $g \in \cG_\beta \setminus \cG_\beta^\infty$ for $\beta \in (0, 1)$ and there exists no $Z \in \cM$ with $\Vert Z \Vert_\infty < \frac{1}{\beta}$ but a $Z \in \cM$  with $\Vert Z \Vert_\infty = \frac{1}{\beta}$, then the market admits strong $\rho$-arbitrage for $\rho = \ES^g$. However, since the $L^1$-closure of $\cQ^{\alpha^g}$ is $\{Z \in \cD: \lVert Z \rVert_{\infty} \leq \tfrac{1}{\beta}\}$, it follows from \eqref{eq:dual char of bar rho} and 
\cite[Theorem 4.15]{herdegen2020dual} that the market does not admit strong $\rho^{\infty}$-arbitrage.
\end{remark}

\section{Conclusion and outlook}
\label{sec:conclusion}
The goal of this paper has been to answer the four questions posed in the introduction.  We have seen that essentially (Q1) has a positive answer if $\rho$ satisfies sensitivity to large expected losses on the set of excess returns.  However, this axiom is not enough to avoid (strong) $\rho$-arbitrage, which is the main concern of the regulator.

In order to characterise the absence of $\rho$-arbitrage, we discovered the key relationship between mean-$\rho$ portfolio selection and mean-$\rho^\infty$ portfolio selection, where $\rho^\infty$ is the \emph{smallest} positively homogeneous risk functional that dominates $\rho$. This relationship is crucial for the dual characterisation of $\rho$-arbitrage when $\rho$ is a convex risk measure since it allows to lift the results on mean-$\rho$ portfolio selection from the coherent case to the convex case. This link between $\rho$ and $\rho^\infty$ breaks in the case of strong $\rho$-arbitrage.  Nevertheless, we were still able to derive a dual characterisation of strong $\rho$-arbitrage by using tools from convex analysis. 

We answered (Q3) by showing that the well-posedness of the mean-$\rho$ problems (1) and (2) is equivalent to the absence of $\rho$-arbitrage, when $\rho$ is sensitive to large expected losses and satisfies the Fatou property.

Finally, as a byproduct of (Q2), we were able to answer (Q4).  In the case of elliptical markets, this led to the elegant result in the form of Theorem \ref{cor: elliptical market}. More generally, when $\rho$ is a convex risk measure, this boils down to taking discounted expectations with respect to certain (but not necessarily all) absolutely continuous/equivalent martingale measures for the discounted risky assets, cf.~Theorems \ref{cor:strong rho arb pricing bounds} and \ref{cor:rho arb pricing bounds}.  In particular, we stress that unlike the majority of literature on pricing, (strong) $\rho$-consistent intervals can be expressed in a precise way.

\appendix

\section{Counterexamples}
\label{app:examples}
In this appendix we give some counterexamples to complement the results in Section \ref{sec:sensitivity to large losses}.

\begin{example}
\label{exa:binomial}
Suppose that risk is quantified by the Expected Shortfall at level $\alpha \in (0,1)$, $\ES^\alpha$.  Consider the binomial model with one riskless asset and one risky asset with returns $R^{0}$ and $R^{1}$, respectively, satisfying
	\begin{equation*}
	R^{0} = r\;\as{\P}, \quad   \P[R^{1}=u] = p  \quad \text{and} \quad\P[R^1 =d] = 1-p,
	\end{equation*}
where we assume that $p \in (0, 1)$ and $-1 < d < r < u$.  We also assume that $\mathbb{E}[R^{1}] > r$, i.e., $up+d(1-p) > r$.  For given desired excess return $\nu > 0$, we have to invest $\pi^1_\nu = \frac{\nu}{up+d(1-p) -r}$ into the risky asset to obtain an expected excess return of $\nu$. Denote the corresponding ES at level $\alpha$ by $\ES^\alpha_{\nu}$. Then $\ES^\alpha_\nu = \nu \ES^\alpha_1$ where
\begin{equation*}
	\ES^\alpha_{1} = 
	\begin{cases}
	\pi^{1}_{1}(r-d), &\text{if } \alpha \leq 1-p,\\
	\pi^{1}_{1}\tfrac{u-r}{\alpha}\left(\frac{1-p}{1-q}- \alpha\right),&\text{if } \alpha > 1-p.
	\end{cases}
\end{equation*}
Setting $\alpha^* := \frac{1-p}{1-q}$, we obtain
\begin{equation*}
	\ES^\alpha_{1}
	\begin{cases}
	> 0, &\text{if } \alpha < \alpha^*, \\
	= 0, &\text{if } \alpha = \alpha^*, \\
<0, &\text{if } \alpha > \alpha^*.
	\end{cases}
\end{equation*}
Thus, for $\alpha \geq \alpha^*$, $\ES^\alpha_\nu \leq 0$ for all $\nu \geq 0$ (the market admits $\ES^\alpha$-arbitrage).  And for $\alpha > \alpha^*$, we have that $\ES^\alpha_\nu \downarrow -\infty$ as $\nu \uparrow \infty$ (the market admits strong $\ES^\alpha$-arbitrage).
\end{example}

\begin{example}
\label{exa:no WSTLL}
%In this example we show that when $\rho$ is not weakly sensitive to large losses on $\cX$ or does not satisfy the Fatou property on $\cX$, then Proposition \ref{thm:reg arb:first characterisation} can fail.

Consider a three-dimensional market where $r=0$, $R^1 \sim N(0,1)$ and $R^2 \sim N(1,1)$ and assume that $R^1$ and $R^2$ are not perfectly correlated.  Then for $\nu \in \RR$, $\Pi_{\nu} = \{(\pi^1,\nu) : \pi^1 \in \RR\}$.  Let $\cX:=\{\pi^1 R^1 + \nu R^2: (\pi^1,\nu) \in \RR^2 \}$.  

\medskip
\noindent (a) Define $\eta:\cX \to (-\infty,\infty]$ by
 \begin{equation*}
     \eta(\pi^1 R^1 + \nu R^2) = \begin{cases}
    \max\{\nu^2 \e^{-\pi^1/\nu} - 1 ,0\},&\text{if $\nu > 0$ and $\pi^1 \in \RR$,} \\ 0,&\text{if $\nu = 0$ and $\pi^1 \in \RR$}, \\ \infty,&\text{otherwise}.
     \end{cases} %\ \textnormal{ where } \  f_{\nu}(\pi^1) = \begin{cases}
     %\nu^2 \e^{-\pi^1/\nu},&\text{if $\pi^1 \geq 0$ and $\nu > 0$}, \\ 0,&\text{if $\pi^1 \geq 0$ and $\nu = 0$}, \\ \infty,&\text{otherwise}.
     %\end{cases}
 \end{equation*}
Then $\eta$ is normalised and star-shaped but not sensitive to large expected losses on $\cX$. By Proposition \ref{prop:extending risk measure}, we can extend $\eta$ to a risk measure $\rho:L^1 \to (-\infty,\infty]$ that is normalised, star-shaped, monotone and satisfies $\rho|_{\cX} \equiv \eta$. It also satisfies the Fatou property on $\cX$ but not sensitivity to large expected losses on $\cX$.  For $\nu \geq 0$, $\rho_\nu$ is attained and equal to $0$.  Thus, the market admits $\rho$-arbitrage. Now it is not difficult to check that 
\begin{equation*}
    \rho^\infty(\pi^1 R^1 + \nu R^2) = \begin{cases}
     0,&\text{if $\nu = 0$ and $\pi^1 \in \RR$,} \\ \infty,&\text{if $\nu \neq 0$ and $\pi^1 \in \RR$}.
     \end{cases}
\end{equation*}
Whence, $\rho^\infty_1 = \infty$.  Therefore, in the absence of sensitivity to large expected losses (even if the Fatou property is satisfied), $\rho$-arbitrage does not imply $\rho^\infty$-arbitrage.

\medskip
\noindent (b) Alter the above risk functional by defining $\eta(\pi^1 R^1 + \nu R^2) = \infty$ for $\nu = 0$ and $\pi^1 \neq 0$.  Then $\rho$ satisfies sensitivity to large expected losses on $\mathcal{X}$, but no longer satisfies the Fatou property on $\mathcal{X}$.  By arguing as above, it is not difficult to check that the market admits $\rho$-arbitrage but not $\rho^\infty$-arbitrage.  Whence, if the Fatou property is not satisfied (even if sensitivity to large expected losses is satisfied), $\rho$-arbitrage does not imply $\rho^\infty$-arbitrage.
\end{example}

\section{Key definitions and results on convex analysis}
\label{app:convex analysis}
In this appendix, we recall some key definitions and results regarding convex functions and convex conjugates.

\medskip
Let $X$ be a topological vector space and $f:X \to [-\infty,\infty]$ a function.
\begin{itemize}
	\item The \emph{epigraph} of $f$ is given by 
	\begin{align*}
		\textnormal{epi$\,f$} & := \{(x,t) \in X \times \RR : f(x) \leq t \}.
	\end{align*}
	Note that $f$ can be recovered from its epigraph, $f(x)=\inf\{t \in \RR : (x,t) \in \text{epi$\,f$}\}$.  Also, a function $g:X \to [-\infty,\infty]$ is dominated by $f$ if and only if epi\,$f \subset$ epi\,$g$.
	\item The \emph{effective domain} of $f$ is given by
		\begin{align*}
		\textnormal{dom$\,f$} & :=\{x \in X : f(x) < \infty\}.
	\end{align*}
We say $f$ is \emph{proper} if dom$\,f \neq \emptyset$ and $f(x) > -\infty$ for all $x \in X$. 
\item  We say $f$ is \emph{convex} if epi$\,f$ is a convex subset of $X \times \RR$. Note that if $f$ is convex, $\textnormal{dom$\,f$}$ is a convex subset of $X$. 

\item We say $f$ is \emph{quasi-convex} if $\{x \in X: f(x) \leq t\}$ is a convex subset of $X$ for all $t \in \RR$. Every convex function is quasi-convex, but the converse is not true. However, if $f$ is quasi-convex, $\textnormal{dom$\,f$}$ is a convex subset of  $X$. 

	\item  We say $f$ is \emph{lower semi-continuous} if epi$\,f$ is a closed subset of $X \times \RR$.
		\item  The \emph{convex hull} of $f$, co$\,f:X \to [-\infty,\infty]$, is the largest convex function majorised by $f$,  
	\begin{equation*}
		\textnormal{co$\,f$}(x) := \sup\{g(x) | \, g:X \to [-\infty,\infty] \textnormal{ is convex and } g \leq f\}.
	\end{equation*}
	By \cite[Equation (3.5)]{rockafellar1974conjugate}, epi\,co\,$f = \{(x,t) \in X \times \RR : (x,s) \in \textnormal{co\,epi\,$f$ for all } s > t \}$, where co\,epi\,$f:=\bigcap\{C \subset X \times \RR : \textnormal{epi\,$f$} \subset C \textnormal{ and $C$ is convex}  \}$. Moreover, it is not difficult to check that $\mathrm{dom}\,\textnormal{co}\,f= \textnormal{co}\,\mathrm{dom}\,f$, where $\textnormal{co}\,\mathrm{dom}\,f = \bigcap\{C \subset X : \textnormal{dom\,$f$} \subset C \textnormal{ and } C \textnormal{ is convex}  \}$.
	
	\item The \emph{quasi-convex hull} of $f$, qco$\,f:X \to [-\infty,\infty]$, is the largest quasi-convex function majorised by $f$,  
	\begin{equation*}
		\textnormal{qco$\,f$}(x) := \sup\{g(x) | \, g:X \to [-\infty,\infty] \textnormal{ is quasi-convex and } g \leq f\}.
	\end{equation*}
	Since every convex function is quasi-convex, it follows that $\textnormal{co$\,f$} \leq \textnormal{qco$\,f$} \leq f$. Moreover, it is not difficult to check that $\mathrm{dom}\,\textnormal{co}\,f=\mathrm{dom}\,\textnormal{qco}\,f = \textnormal{co}\,\mathrm{dom}\,f$.

	\item The \emph{lower semi-continuous hull} of $f$, lsc\,$f:X \to [-\infty,\infty]$ is the largest lower semi-continuous function majorised by $f$,
	\begin{equation*}
		%\label{eq:lsc hull}
		\textnormal{lsc$\,f$}(x) := \sup\{h(x) | \, h:X \to [-\infty,\infty] \textnormal{ is lower semi-continuous and } h \leq f\}.
	\end{equation*}
	By \cite[Equation (3.6)]{rockafellar1974conjugate}, epi\,lsc\,$f = $ cl\,epi\,$f$, or equivalently we have \cite[Equation (3.7)]{rockafellar1974conjugate},
	\begin{equation}
		\label{eq:lsc hull}
		\textnormal{lsc\,}f(x) = \inf \{ \liminf_{i \in I} f(x_i) : \lim_{i \in I}x_i = x \}.
	\end{equation}
In particular, this implies that $\mathrm{dom}\,\textnormal{lsc}\,f \subset \mathrm{cl}\,\mathrm{dom}\, f$.
	\item The \emph{lower semi-continuous convex hull} of $f$, $\overline{\textnormal{co}}\,f:X \to [-\infty,\infty]$ is given by $\overline{\textnormal{co}}\,f := \textnormal{lsc\,co\,}f$ (which may not be the same as co\,lsc\,$f$). Since the closure of a convex set is again convex and epi\,$\overline{\textnormal{co}}\,f = \mathrm{cl\,co\,epi}\,f$, it follows that $\overline{\textnormal{co}}\,f$ is the largest lower semi-continuous convex function majorised by $f$. Moreover, 
		\begin{equation}
			\label{eq:lsc co hull}
		\mathrm{dom}\,\overline{\textnormal{co}}\,f\subset \mathrm{cl}\,\mathrm{co}\,\mathrm{dom}\,f.
		\end{equation}
	\item If $Y$ is a nonempty subset of $X$ and $f: Y \to [-\infty, \infty]$ a function, we can extend $f$ to $X$ by considering the function $\bar f: X \to [-\infty, \infty]$ defined by 
	\begin{equation*}
		\bar f(x) = 
		\begin{cases}
			f(x),&\text{if } x \in Y, \\
			\infty,&\text{if } x \in X \setminus Y.
		\end{cases}
	\end{equation*}
This extension is \emph{natural} in that $\textnormal{epi$\,\bar f$} \subset Y \times \RR$, $\mathrm{dom}\,\bar f \subset Y$, $\mathrm{dom}\,\textnormal{co}\,\bar f, \mathrm{dom}\,\textnormal{qco}\,\bar f \subset Y$ if $Y$ is convex, $\mathrm{dom}\,\textnormal{lsc}\,\bar f \subset Y$ if $Y$ is closed and $\mathrm{dom}\,\overline{\textnormal{co}}\,\bar f \subset Y$ if $Y$ is convex and closed. For this reason, if $Y$ is convex, we may define the functions $\textnormal{co}\,f, \textnormal{qco}\, f:Y \to [-\infty, \infty]$ by  $\textnormal{co}\,f(x) :=\textnormal{co}\,\bar f(x)$, $\textnormal{qco}\,f(x) :=\textnormal{qco}\,\bar f(x)$ and call this the convex hull and quasi-convex hull of $f$, respectively. Similarly, if $Y$ is closed (and convex), we may define the functions $\textnormal{lsc}\,f:Y \to [-\infty, \infty]$ (and $\overline{\textnormal{co}}\,f:Y \to [-\infty, \infty]$) by $\textnormal{lsc}\,f(x) :=\textnormal{lsc}\,\bar f(x)$ (and $\overline{\textnormal{co}}\,f(x) :=\overline{\textnormal{co}}\,\bar f(x)$) and call this the lower-semi-continuous (convex) hull of $f$.
\end{itemize}
In order to discuss convex conjugates, we assume that  $\langle X,X' \rangle$ is a dual pair under the duality $\langle \cdot , \cdot \rangle: X \times X' \to \RR$, i.e., $X$ and $X'$ are vector spaces together with a bilinear functional $(x,x') \mapsto \langle x,x' \rangle$ such that
\begin{itemize}
    \item If $\langle x,x' \rangle = 0$ for each $x' \in X'$, then $x = 0$;
    \item If $\langle x,x' \rangle = 0$ for each $x \in X$, then $x' = 0$.
\end{itemize}
%In this way, $X$ ($X'$) can be interpreted as a set of linear functionals on $X'$ ($X$) since each $x \in X$ ($x' \in X'$) induces the linear functional $x' \mapsto \langle x,x' \rangle$ on $X'$ ($x \mapsto \langle x,x' \rangle$ on X).  
We endow $X$ with the weak topology, $\sigma(X,X')$, 
\begin{equation*}
    x_\alpha \xrightarrow{w} x \textnormal{ in $X$ if and only if } \langle x_\alpha,x' \rangle \xrightarrow{} \langle x,x' \rangle \textnormal{ in } \RR \textnormal{ for each } x' \in X',
\end{equation*}
and $X'$ with the weak* topology, $\sigma(X',X)$,
\begin{equation*}
    x'_\alpha \xrightarrow{w^*} x' \textnormal{ in $X'$ if and only if } \langle x,x'_\alpha \rangle \xrightarrow{} \langle x,x' \rangle \textnormal{ in } \RR \textnormal{ for each } x \in X.
\end{equation*}
These topologies are locally convex and Hausdorff; the topological dual of $(X,\sigma(X,X'))$ is $X'$; and the topological dual of $(X',\sigma(X',X))$ is $X$; see \cite[Section 5.14]{guide2006infinite} for details.
\begin{comment}
These last two results explain where the name dual pair comes from.  Important examples include:
\begin{enumerate}[label=(\roman*)]
    \item $\langle \RR^n, \RR^n \rangle$ under the duality $\langle x,y \rangle= x \cdot y$.
    \item $\langle L^1, L^\infty \rangle$ under the duality $\langle Z, X \rangle=\mathbb{E}[ZX]$.
    \item $\langle X, X' \rangle$ where $X'$ is a \emph{total subspace} of the algebraic dual of $X$ (i.e., for any $x,y \in X$, if $x'(x)=x'(y)$ for all $x' \in X'$ then $x=y$), under the evaluation duality $(x,x') \mapsto x'(x)$.
\end{enumerate}
\end{comment}

\begin{itemize}
\item  The \emph{convex conjugate} of $f$, $f^*:X' \to [-\infty,\infty]$, and the \emph{biconjugate} of $f$, $f^{**}:X \to [-\infty,\infty]$, are defined as
\begin{equation*}
    f^*(x') := \sup \{\langle x,x' \rangle - f(x):x \in X\} \quad \textnormal{and} \quad f^{**}(x) := \sup \{\langle x,x' \rangle - f^*(x'):x' \in X'\}.
\end{equation*}
\item It follows from \cite[Theorem 5]{rockafellar1974conjugate} that epi\,$f^{**}$ is the intersection of all the ``non-vertical'' closed half spaces in $X \times \RR$ that contain epi\,$f$, i.e.,
\begin{equation}
	\label{eq:biconjugate}
    f^{**}(x) = \sup\{a(x) \,| \, a:X \to \RR \textnormal{ is affine and continuous and } a \leq f \},
\end{equation}
where a function $a:X \to \RR$ is \emph{affine and continuous} if it is of the form $a(x) = \langle x,x' \rangle + c$ for some $x' \in X'$ and $c \in \RR$.
\item If $\overline{\textnormal{co}}\,f(x) > -\infty$ for all $x \in X$, then $f^{**} = \overline{\textnormal{co}}\,f$ by \cite[Theorems 4 and 5]{rockafellar1974conjugate}. In particular if $f$ is convex, lower semi-continuous and proper, then $f = f^{**}$, which is the famous Fenchel-Moreau theorem.
\end{itemize}

\section{Additional results and proofs}
\label{app:additional results}

\begin{proposition}
    \label{prop:general result relating WSTD SSTD}
Let $\rho$ be a risk functional that satisfies the Fatou property on $\mathcal{X}$ and let $c \geq 0$.  Assume there exists an unbounded sequence of portfolios $(\pi_n)_{n \geq 1} \subset \RR^d$ with $\rho(X_{\pi_{n}}) \leq c$ for all $n \in \NN$.  Then there exists a portfolio $\pi \in \RR^d \setminus \{\mathbf{0}\}$ with $\rho(\lambda X_\pi) \leq c$ for all $\lambda > 0$.  Moreover, if $\mathbb{E}[X_{\pi_n}] = 0$ for all $n$, we may further assume $\mathbb{E}[X_\pi] = 0$.
\end{proposition}

\begin{proof}
By passing to a subsequence and relabelling the assets, we may assume without loss of generality that $|\pi^1_n| \geq |\pi^i_n|$ for all $n \in \NN$ and $i \in \{1,\dots,d\}$.  As $\lVert \pi_n \rVert \to \infty$ we must have that $|\pi^1_n| \to \infty$, and by shifting the sequence we may assume $|\pi^1_n| > 0$ for all $n \in \NN$.  Then for all $i \in \{1,\dots,d\}$ we have $\pi^i_n/\pi^1_n \in [-1,1]$ and by compactness we can pass to a further subsequence and assume that $\pi^i_n/|\pi^1_n| \to \pi^i \in [-1,1]$, where $\pi^1 \in \{-1, 1\}$.  It follows that
\begin{equation}
\label{eq:portfolio X c}
    X_{\pi_n}/|\pi^1_n| \to X_\pi \ \as{\mathbb{P}},
\end{equation}
where $\pi \neq \mathbf{0}$ since $\pi^1 \in \{-1, 1\}$.  Since $|\pi^1_n| \to \infty$, for any $\lambda > 0$, there exists $N$ such that $\lambda/|\pi^1_n| \in (0,1)$ for all $n \geq N$. Now star-shapedness of $\rho$ gives
\begin{equation*}
    \rho(\lambda  X_{\pi_n}/|\pi^1_n|) \leq \lambda \rho(X_{\pi_n})/|\pi^1_n| \leq c, \quad n \geq N.
\end{equation*}
By the Fatou property ($L \supset \cX$ being a Riesz space), $\rho(\lambda X_\pi) \leq \liminf_{n \to \infty}\rho(\lambda X_{\pi_n}/|\pi^1_n|) \leq c$.  Hence $\rho(\lambda X_\pi) \leq c$ for all $\lambda > 0$.

If in addition $\mathbb{E}[X_{\pi_{n}}] = 0$ for all $n$, then linearity of the expectation and the dominated convergence theorem gives $\mathbb{E}[X_\pi] = 0$.  Indeed, since $\pi^i_n/\pi^1_n \in [-1,1]$ we have 
\begin{equation*}
    |X_{\pi_n}/|\pi^1_n|| = |X^1+\tfrac{\pi^2_n}{|\pi^1_n|}X^2 + \dots \tfrac{\pi^d_n}{|\pi^1_n|}X^d| \leq |X^1| + |X^2| + \dots |X^d|
\end{equation*}
where $X^i:=R^i-r \in L^1$ for $i \in \{1,\dots,d\}$.  This, together with \eqref{eq:portfolio X c} and the dominated convergence theorem gives $\mathbb{E}[X_\pi] = \lim_{n \to \infty} \mathbb{E}[X_{\pi_n}/|\pi^1_n|] = 0$. 
\end{proof}

 \begin{proof}[Proof of Theorem \ref{thm:WSTD equiv to boundedness and existence of rho optimal sets}]
Define the function $f_\rho:\RR^d \to [0,\infty]$ by $f_\rho(\pi) = \max\{\rho(X_\pi),0\}+|\mathbb{E}[X_\pi]|$.  Then $f_\rho$ is lower semi-continuous by the Fatou property of $\rho$ on $\cX$ (and the fact that $L \supset \cX$ is a Riesz space) and linearity of the expectation. Moreover, it is star-shaped, i.e., $f_\rho(\lambda \pi) \geq \lambda f_\rho(\pi)$ for all $\lambda \geq 1$ and $\pi \in \RR^d$,  by the star-shapedness of $\rho$ and linearity of  the expectation. 

For $\delta \geq 0$, set $A_\delta :=\{\pi \in \RR^d : f_\rho(\pi) \leq \delta \}$. Then each $A_\delta$ is closed by lower semi-continuity of $f_\rho$. We proceed to show that each $A_\delta$ is also bounded and hence compact. 

For $\delta = 0$, using $f_\rho(\pi) \geq |\mathbb{E}[X_\pi]| > 0$ for any $\pi \in \RR^d \setminus \Pi_0$, it follows that $A_0 \subset \Pi_0$. Also note that for each $\pi \in A_0$, $X_\pi \in \cA_\rho$. If $A_0$ were unbounded, then Proposition \ref{prop:general result relating WSTD SSTD} would imply the existence of a portfolio $\pi \in \Pi_0 \setminus \{\mathbf{0}\}$ with $\rho(\lambda X_\pi) \leq 0$ for all $\lambda >0$.  But this would contradict $\rho$ being sensitive to large expected losses on $\mathcal{X}$.  Therefore, $A_0$ must be bounded.

For $\delta > 0$, we argue as follows: Since $A_0$ is bounded, there exists $d > 0$ such that $f_\rho(\pi) > 0$ for any portfolio $\pi$ belonging to the set $D:=\{ x \in \RR^d : \lVert x \rVert_2 = d \}$. Compactness of $D$ and lower semi-continuity of $f_\rho$ give $m:=\min\{f_\rho(x) : x \in D\} \in (0,\infty]$.  Star-shapedness of $f_\rho$ in turn implies that $	f_\rho(\pi) \geq m\lVert \pi \rVert_{2}/d$ for all  $\pi \in \RR^d$ with $\lVert \pi \rVert_{2} \geq d$, which in turn implies that each $A_\delta$ is bounded.

We finish by a standard argument.  Fix $\nu \geq 0$ and assume $\rho_\nu < \infty$.  By definition, there exists a sequence of portfolios $(\pi_n)_{n \geq 1} \subset \Pi_\nu$ such that $\rho(X_{\pi_n}) \searrow \rho_\nu$ and $\rho_\nu + 1 \geq \rho(X_{\pi_n})$ for all $n$.  Setting $\delta^*:=\max\{\rho_\nu + 1, 0\} + \nu$, it follows that $(\pi_n)_{n \geq 1} \subset A_{\delta^*}$.  Compactness of $A_{\delta^*}$, closedness of $\Pi_\nu$ and the Fatou property of $\rho$ imply the existence of a portfolio $\pi \in \Pi_\nu$ with $\rho(X_\pi) \leq \rho_\nu$, i.e., $\Pi^\rho_\nu$ is nonempty.  Furthermore, $\Pi^\rho_\nu$ is bounded since it is a subset of $A_{\delta^*}$, and closed since $\rho$ satisfies the Fatou property. 
 \end{proof}

\begin{proposition}
\label{prop:interchange min and lim}
Suppose $X$ is a topological space and $K \subset X$ is compact.  Then for any nondecreasing sequence of lower semi-continuous functions $f_t : K \to [-\infty,\infty]$ with $f(x):=\lim_{t \to \infty} f_t(x)$ for all $x \in K$, we have
\begin{equation*}
    \min_{x \in K} f(x) = \lim_{t \to \infty} \min_{x \in K} f_t(x).
\end{equation*}
Furthermore, if $(x_t)_{t \geq 1}$ is a sequence where $\min_{x \in K} f_t(x) = f_t(x_t)$, then any limit point is a minimiser for $f$.\footnote{Convergence of minima and convergence of minimisers are often delicate, but important notions in optimisation problems.  A similar result to Proposition \ref{prop:interchange min and lim} is \cite[Lemma 2.7(c)]{hernandez1992discrete}.  The application there was to relate \emph{finite} horizon discrete time Markov decision processes with \emph{infinite} horizon ones.}
\end{proposition}

\begin{proof}
First note that $f$ is lower semi-continuous because it is the supremum of lower semi-continuous functions.  By the compactness of $K$ and lower semi-continuity, $f$ and $f_t$ attain their minimum values.  Now since $f_t$ is a nondecreasing sequence, it is easy to see that
\begin{equation*}
    \min_{x \in K} f(x) \geq \lim_{t \to \infty} \min_{x \in K} f_t(x) =: m.
\end{equation*}
For the reverse inequality, consider the sets $A_t:=\{x \in K : f_t(x) \leq m \}$.  These are nonempty (because $\emptyset \neq \argmin f_t \subset A_t$), closed (by the lower semi-continuity of $f_t$) and compact (since $K$ is compact and $A_t$ is closed).  Moreover, they are nested in the sense that $A_t \supset A_{t+1}$.  It follows by Cantor's intersection theorem that 
\begin{equation*}
    A := \bigcap_{t=1}^{\infty} A_t \neq \emptyset,
\end{equation*}
i.e., there exists $x^* \in K$ such that $f_t(x^*) \leq m$ for all $t$.  Taking the limit as $t \to \infty$ yields
\begin{equation*}
    \min_{x \in K} f(x) \leq f(x^*) \leq m = \lim_{t \to \infty} \min_{x \in K} f_t(x).
\end{equation*}

To prove the final claim, note that $\argmin f = A$ because $f(x) \leq m$ if and only if $f_t(x) \leq m$ for all $t$.  Whence, any limit point of a sequence of minimisers $(x_t)_{t \in \NN}$ -- that is where $x_t \in \argmin f_t$ for all $t \geq 1$ -- is contained in $A$, and hence, is a minimiser for $f$.
%Why are limit points contained in A?  Because for any n, the sequence is contained in A_n and as A_n is closed, all its limit points are in A_n.  Hence also in A.  We do not need sequential compactness, but if we did have sequential compactness, then we would know there is at least one limit point always.
\end{proof}

\begin{proof}[Proof of Proposition \ref{prop:properties of rho optimal boundary}]
     First note that by Theorem \ref{thm:WSTD equiv to boundedness and existence of rho optimal sets}, the map $\nu \mapsto \rho_\nu$ is $(-\infty,\infty]$-valued.  
    
    Next we establish lower semi-continuity. Fix $y \in \RR$ and let $B_y:=\{\nu \in \RR_+ : \rho_\nu \leq y \}$.  We must show that this set is closed.  So let $(\nu_n)_{n \geq 1} \subset B_y$ and assume $\nu_n \to \nu$.  By Theorem \ref{thm:WSTD equiv to boundedness and existence of rho optimal sets}, for each $n$ there exists a portfolio $\pi_n$ such that $\rho(X_{\pi_n}) = \rho_{\nu_n} \leq y$ and $\mathbb{E}[X_{\pi_n}] = \nu_n$. We proceed to show that the sequence $(\pi_n)_{n \geq 1}$ belongs to a compact set. To this end, let $c \in \RR$ be such that $|\nu_n| \leq c$ for all $n$.  Setting $\delta := \max\{y,0\}+c$ it follows that each $\pi_n$ lies in  $A_\delta := \{ \pi \in \RR^d : \max\{\rho(X_\pi),0\} + |\mathbb{E}[X_\pi]| \leq \delta \}$, which is compact by the proof of Theorem \ref{thm:WSTD equiv to boundedness and existence of rho optimal sets}. Passing to a subsequence, we may assume that $(\pi_n)_{n \geq 1}$ converges to some $\pi \in \RR^d$, and by dominated convergence and the Fatou property, it follows that $\mathbb{E}[X_{\pi}] = \nu$ and $\rho(X_{\pi}) \leq y$. Whence, $\rho_\nu \leq y$ and so $\nu \in B_y$.
    
    We now show that $\nu \mapsto \rho_\nu^\infty$ is the smallest positively homoegeneous majorant of $\nu \mapsto \rho_\nu$.  Since $\rho^\infty$ is sensitive to large expected losses, $\rho^\infty_0 = 0 \geq \rho_0$.  Thus, it suffices to show that 
    \begin{equation}
 		\label{eq:bar rho 1}
 		\rho^\infty_1 = \lim_{t \to \infty} \rho_{t}/t.
 	\end{equation}
The key idea is to consider the risk functionals $\rho^{t}:L \to (-\infty,\infty]$ defined by $\rho^{t}(X)=\rho(tX)/t$ for $t \geq 1$.  They satisfy the Fatou property on $\mathcal{X}$, sensitivity to large expected losses on $\cX$ and $\rho_t/t = \rho^t_1$.  By star-shapedness of $
 	\rho$ and definition of $\rho^{\infty}$ in \eqref{eq:rho bar}, we have $\rho^{t+1}(X) \geq \rho^{t}(X)$  and $\lim_{t \to \infty}\rho^{t}(X) = \rho^\infty(X)$ for all $X \in L$. This implies $(\rho_1^{t})_{t \geq 1}$ is a nondecreasing sequence and
 	\begin{equation*}
 		\rho^\infty_1 \geq m := \lim_{t \to \infty} \rho^{t}_1.
 	\end{equation*}
 	If $m = \infty$, the reverse inequality is clear, so assume $m < \infty$. Then as  $\rho^{t} \geq \rho$ and $m \geq \rho_1^{t}$ for each $t \geq 1$, it follows that
 	\begin{equation*}
 		\Pi^{\rho^{t}}_1 \subset \{\pi \in \Pi_1 : \rho(X_\pi) \leq m \} \subset \{ \pi \in \Pi_1 : \max\{\rho(X_\pi),0\} + |\mathbb{E}[X_\pi]| \leq \max\{m,0\} + 1 \}  := K.
 	\end{equation*}
 	Since $K$ is compact by the proof of Theorem \ref{thm:WSTD equiv to boundedness and existence of rho optimal sets}, \eqref{eq:bar rho 1} follows by applying Proposition \ref{prop:interchange min and lim} to the sequence of functions $f_t:K \to (-\infty,\infty]$ given by $f_t(\pi) := \rho^{t}(X_\pi)$. 
 	
The statements in (b) and (c) as well as the equivalence between $\rho^\infty_1 > 0$ and $\nu^+ < \infty$ follow directly from the fact $\nu \mapsto \rho^\infty_\nu$ is the smallest positively homogeneous majorant of $\nu \mapsto \rho_\nu$.
 
Finally, we establish (a).  If $\rho^\infty_1 > 0$, then $\nu^{+} < \infty$ by the above and hence $\rho_\nu > 0$ for all $\nu > \nu^+$.  By lower semi-continuity of $\nu  \mapsto \rho_\nu$ and compactness of $[0,\nu^{+}]$, there exists a global minimum $m \leq \rho_0 \leq 0$ that is attained at $\nu^*:=\sup\{\nu \in [0,\nu^{+}]: \rho_\nu = m\}$. By construction, $\rho_\nu > m$ for all $\nu > \nu^*$.  Whence, by definition $\nu_{\min} = \nu^* < \infty$ and $\rho_{\min} = \rho_{\nu_{\min}} \in (-\infty,0]$.
\end{proof}

\begin{proof}[Proof of Proposition \ref{prop:convex and WSTD optimal boundary properties}]
First, we establish convexity of $\nu \mapsto \rho_\nu$.
Let $\nu, \nu' \in \RR_+$, $\lambda \in [0,1]$ and $A:=\Pi_\nu \times \Pi_{\nu'}$.  Using convexity of $\rho$ and the fact $(\pi,\pi') \in A$ implies $\lambda \pi + (1-\lambda) \pi' \in  \Pi_{\lambda\nu + (1- \lambda)\nu'}$, we obtain
\begin{align*}
    \rho_{\lambda \nu + (1-\lambda)\nu'} \leq  \inf_{(\pi,\pi') \in A } \{ \rho(X_{\lambda \pi + (1-\lambda)\pi'}) \} \leq \inf_{(\pi,\pi') \in A } \{ \lambda \rho(X_\pi) + (1-\lambda)\rho(X_\pi') \} \leq \lambda \rho_\nu + (1-\lambda) \rho_{\nu'}.
\end{align*}
Thus, $\nu \mapsto \rho_\nu$ is convex on $\RR_+$. 

Next, since $\nu \mapsto \rho_\nu$ is convex, it is continuous in the interior of its effective domain $\{\nu \in \RR_+ : \rho_{\nu} < \infty\}$, which is an interval. This together with lower semi-continuity shown in Proposition \ref{prop:properties of rho optimal boundary}, implies that $\nu \mapsto \rho_\nu$ is finite and continuous on the closure of $\{\nu \in \RR_+ : \rho_{\nu} < \infty\}$, which a fortiori implies that $\{\nu \in \RR_+ : \rho_{\nu} < \infty\}$ is closed. The other claims follow directly from Propositions \ref{prop:basic properties of rho optimal boundary} and  \ref{prop:properties of rho optimal boundary} together with standard properties of convex functions.
\end{proof}

\begin{proof}[Proof of Theorem \ref{prop:strong reg arb:first characterisation}]
	``(a) $\iff$ (b)''. This is \cite[Theorem 3.18]{herdegen2020dual}. 
``(b) $\implies$ (c)''.  This follows from the definition of strong $\rho$-arbitrage and the fact $\rho^\infty$ dominates $\rho$.  
\end{proof}

\begin{proof}[Proof of Theorem \ref{thm:reg arb:first characterisation}]
``(c) $\implies$ (b).''  This follows from the definition of $\rho$-arbitrage and the fact $\rho^\infty$ dominates $\rho$.

``(b) $\implies$ (a).'' We prove by contraposition, so assume $\rho^\infty_1 \leq 0$.  Since $\rho$ satisfies the Fatou property on $\cX$ and sensitivity to large expected losses on $\cX$, so too does $\rho^\infty$. Whence, by Theorem \ref{thm:WSTD equiv to boundedness and existence of rho optimal sets}, $\rho^\infty_1 = \rho^\infty(X_\pi)$ for some $\pi \in \Pi_1$.  Letting $\pi_n = n\pi$ yields a sequence of portfolios for which $\mathbb{E}[X_{\pi_n}] \uparrow \infty$ and $\rho^\infty(X_{\pi_n}) = n\rho^\infty_1 \leq 0$ for all $n$.  Thus, the market admits $\rho^\infty$-arbitrage.

``(a) $\implies$ (c).'' Assume $\rho^\infty_1 > 0$.  This implies $\nu^+ < \infty$ by Proposition \ref{prop:properties of rho optimal boundary}. By definition of $\nu^+$, it follows that for any sequence of portfolios with $\mathbb{E}[X_{\pi_n}] \uparrow \infty$, $\rho(X_{\pi_n}) > 0$ eventually.  Therefore, the market does not admit $\rho$-arbitrage.
\end{proof}

\begin{proposition}
    \label{prop:extending risk measure}
Let $\cX$ be a subspace of $L^1$ such that $\{(X,Y) \in \cX^2 : Y \geq X \ \mathbb{P}\textnormal{-a.s.~and } \mathbb{P}[Y > X] > 0 \} = \emptyset$.  Suppose $\eta:\cX \to (-\infty,\infty]$ is normalised, star-shaped and $\eta(X) \leq \WC(X)$ for all $X \in \cX$.  Define $\mathcal{Y} := \{ Y \in L^1 \setminus \cX : Y \geq X \ \mathbb{P}\textnormal{-a.s.~for some } X \in \cX \}$ and  $\rho:L^1 \to (-\infty,\infty]$ by
\begin{equation*}
    \rho(X) = \begin{cases}
			\eta(X),&\text{if $X \in \cX$}, \\ 
			\mathbb{E}[-X],&\text{if $X \in \mathcal{Y}$}, \\ 
			\WC(X),&\text{otherwise}.
		\end{cases}
\end{equation*}
Then $\rho$ is a risk functional such that $\rho|_{\cX} \equiv \eta$.
\end{proposition}

\begin{proof}
Normalisation of $\rho$ is clear, as is the fact $\rho|_{\cX} \equiv \eta$.  

To show monotonicity, let $X,Y \in L^1$ and assume $Y \geq X \ \mathbb{P}$-a.s.~and $\mathbb{P}[Y > X] > 0$.  If $X \in \cX$ or $X \in \mathcal{Y}$, then $Y \in \mathcal{Y}$ and $\rho(X) \geq \mathbb{E}[-X] \geq \mathbb{E}[-Y] = \rho(Y)$.  If $X \in L^1 \setminus (\cX \cup \mathcal{Y})$, then $\rho(X) = \WC(X) \geq \WC(Y) \geq \rho(Y)$.  Whence, $\rho$ is monotone.

Finally, to show $\rho$ is star-shaped, fix $\lambda \geq 1$ and $X \in L^1$.  If $X \in \cS$ where $\cS \in \{\cX, \mathcal{Y}, L^1 \setminus (\cX \cup \mathcal{Y})\}$, then also $\lambda X \in \cS$ and so it follows that $\rho(\lambda X ) \geq \lambda \rho(X)$.
\end{proof}

\begin{proof}[Proof of Theorem \ref{prop:A rho infinity larger than L+}]
Assume $\cA_{\rho^\infty} \supsetneq L_+$ and let $X \in \cA_{\rho^\infty} \setminus L_+$.  Since $\rho$ (and hence $\rho^\infty$) are monotone, we may assume without loss of generality that $\mathbb{P}[X > 0] > 0$.  Fix $n \in \NN$ so that $\mathbb{E}[X^+] - \tfrac{1}{n} \mathbb{E}[X^-] > 0$ and let $R:= X^+ - \frac{1}{n} X^-$.  Consider the market $(S^0,S)$ defined by $S^0 \equiv 1$ and $S := S^1$ where $S^1_0 = 1$ and $S^1_1 = 1+ R$.  This market does not admit arbitrage of the first kind since $\mathbb{P}[ R < 0] >0$ and $\mathbb{P}[ R > 0] > 0$.  Moreover, $\mathbb{E}[R] > 0$ and $R \in \cA_{\rho^\infty}$ since $R \geq X$.  Thus, $\rho^\infty_1 \leq 0$ and the market admits $\rho^\infty$-arbitrage. Whence, the market admits $\rho$-arbitrage by Remark \ref{rmk:rho arb scalable acceptable deal relation}(a).

When $\rho$ is cash-invariant,  $\rho^\infty$ is also cash-invariant.  By adding $\epsilon > 0$ sufficiently small to $R$ above, we can find $\tilde{R}$ such that $\mathbb{P}[ \tilde{R} < 0] > 0$, $\mathbb{E}[\tilde{R}] > 0$ and $\rho^\infty(\tilde{R}) < 0$.  Replacing $R$ with $\tilde{R}$ in the market above produces a market where $\rho^\infty_1 < 0$.  This market does not admit arbitrage of the first kind but admits strong $\rho$-arbitrage by Theorem \ref{prop:strong reg arb:first characterisation}.
\end{proof}

\begin{proof}[Proof of Theorem \ref{thm:well posedness}]
Assume the market admits $\rho$-arbitrage.  Then by definition, it follows that the mean-$\rho$ problem (2) does not have any solutions for any $\rho^* \geq 0$.  

\medskip
Assume the market does not admit $\rho$-arbitrage.  Then by Theorem \ref{thm:reg arb:first characterisation}, $\rho^\infty_1 > 0$.  We first show that the set 
\begin{equation*}
    K_c:=\{ \pi \in \RR^d : \rho(X_\pi) \leq c \}
\end{equation*}
is compact for all $c < \infty$.  Closedness follows from the Fatou property.  We show boundedness via contradiction.  If $(\pi_n)_{n \geq 1} \subset K_c$ is unbounded, then by Proposition \ref{prop:general result relating WSTD SSTD}, there exists a nonzero portfolio $\pi \in \RR^d$ with $\rho(\lambda X_\pi) \leq c$ for all $\lambda > 0$.  Thus $\rho^\infty(X_\pi) \leq 0$, and by sensitivity to large expected losses, $\mathbb{E}[X_\pi] > 0$.  But this would contradict the fact $\rho^\infty_1 > 0$. 
 
To show that the mean-$\rho$ problem (1) admits a solution, let $\nu^* \geq 0$.  If $\rho_{\nu^*} = \infty$, then $\rho_\nu = \infty$ for all $\nu \geq \nu^*$ and so any portfolio in $\Pi_\nu$ for $\nu \geq \nu^*$ is a solution to (1).  If $\rho_{\nu^*} < \infty$, then as $K_{\rho_{\nu^*}}$ is compact, and $\rho_{\nu^*}$ is attained by Theorem \ref{thm:WSTD equiv to boundedness and existence of rho optimal sets}, this means that the set
 \begin{equation*}
     K_{\rho_{\nu^*}}^{\nu^*}:=\{ \pi \in \RR^d : \rho(X_\pi) \leq \rho_{\nu^*} \textnormal{ and } \mathbb{E}[X_\pi] \geq \nu^* \}
 \end{equation*}
 is nonempty and compact.  It follows that there exists a convergent sequence $(\pi_n)_{n \geq 1} \subset K_{\rho_{\nu^*}}^{\nu^*}$ with limit $\pi^* \in K_{\rho_{\nu^*}}^{\nu^*}$ such that
 \begin{equation*}
     \lim_{n \to \infty} \rho(X_{\pi_n}) = \inf \{\rho(X_\pi) : \pi \in  K_{\rho_{\nu^*}}^{\nu^*}\}.
 \end{equation*}
 By the Fatou property, $\pi^*$ must be a solution to (1).

 Finally, to show the mean-$\rho$ problem (2) admits a solution, let $\rho^* \in [0,\infty)$.  Then as $K_{\rho^*}$ is nonempty (since $\mathbf{0} \in K_{\rho^*}$) and compact, there exists a convergent sequence $(\pi_n)_{n \geq 1} \subset K_{\rho^*}$ with limit $\pi^* \in K_{\rho^*}$ such that 
 \begin{equation*}
     \lim_{n \to \infty} \mathbb{E}(X_{\pi_n}) = \sup \{\mathbb{E}(X_\pi) : \pi \in  K_{\rho^*}\}.
 \end{equation*}
 As $\mathbb{E}[X_{\pi_n}] \to \mathbb{E}[X_{\pi^*}]$, it follows that $\pi^*$ is a solution to (2).
\end{proof}

\begin{proof}[Proof of Proposition \ref{prop:basic pricing properties}]
The first part of the statement in (a) is because strong $\rho^\infty$-arbitrage implies strong $\rho$-arbitrage, which in turn implies $\rho$-arbitrage.  The second part of (a) is a consequence of Theorem \ref{thm:reg arb:first characterisation} and Corollary \ref{cor:rho infinity WC}. Parts (b) and (c) are clear.
\end{proof}

\begin{proof}[Proof of Theorem \ref{cor: elliptical market}]
Consider the market $(S^0,S,S^{d+1})$ with $S^{d+1}_0=x$ and $S^{d+1}_1 = X$.  We will work with the economically equivalent market $(S^0,\tilde{S},\tilde{S}^{d+1})$ where $\tilde S^i := S^{i} + (1- S^i_0) S^0$.  Then the mean vector of the excess returns is given by $\tilde{\mu}(x)$ and its covariance matrix is given by $\Sigma$.  The maximal Sharpe ratio in the augmented market is then given by 
    \begin{equation*}
        \SR_{\max}(x) := \max_{\pi \in \RR^d \setminus \{\0\}} \frac{\mathbb{E}[X_\pi]}{\sqrt{\Var(X_{\pi})}} = \sqrt{\tilde{\mu}(x)^\top \Sigma^{-1}\tilde{\mu}(x)}.
    \end{equation*}
If $\rho$ is a law-invariant risk measure for which (strong) $\rho$-arbitrage is equivalent to (strong) $\rho^\infty$-arbitrage, then by \cite[Corollary 3.28]{herdegen2020dual}, the market $(S^0,\tilde{S},\tilde{S}^{d+1})$ admits (strong) $\rho$-arbitrage if and only if $\SR_{\max}(x) < (\leq) \ \rho^\infty(Z)$.  Whence,
\begin{align*}
        I_\rho(X) & = I_{\rho^\infty}(X) = \{ x \in \RR : \SR_{\max}(x) < \rho^{\infty}(Z) \} \quad \textnormal{and} \\ 
        I^{s}_{\rho}(X) & = I^{s}_{\rho^\infty}(X) = \{ x \in \RR : \SR_{\max}(x) \leq \rho^{\infty}(Z) \}. \qedhere
    \end{align*}
\end{proof}

\begin{proof}[Proof of Proposition \ref{prop:Cond UI preliminary result}]
Let $X \in \cX$, i.e., there is $\pi \in \RR^d$ such that $X = X_\pi = \pi \cdot (R- r\mathbf{1})$. By Condition UI, this implies that $\cQ^\alpha$ and $X  \cQ^\alpha$ are UI. 
	
	First, $\textnormal{co\,}\alpha$ (whose effective domain is $\cQ^\alpha$) represents $\rho$ by Remark \ref{rem:dual char}(d) since $\alpha^\rho \leq \textnormal{co\,}\alpha \leq \alpha$ by Remark \ref{rem:dual char}(c)
	and the definition of the convex hull. This together with $ \overline{\textnormal{co}}\, \alpha \leq \textnormal{co}\,\alpha$ and  Remark \ref{rmk:Q subset Q start start subset Q bar}(a)  implies that $\cQ^{\overline{\textnormal{co\,}}\alpha} \subset \ol{\cQ}^\alpha$ and 
	\begin{equation}
		\label{eq: sup over Q leq sup over Q bar}
		\rho(X) = \sup_{Z \in \cQ^\alpha} \{\mathbb{E}[-ZX] - \textnormal{co}\,\alpha(Z)\} \leq \sup_{Z \in \ol\cQ^\alpha} \{ \mathbb{E}[-ZX]-\overline{\textnormal{co}}\, \alpha(Z) \}.
	\end{equation}
	If we can show that the supremum on the right side of \eqref{eq: sup over Q leq sup over Q bar} is attained and the inequality is an equality, then \eqref{eq:preliminary representation} follows.
	
	To see that the supremum on the right side of \eqref{eq: sup over Q leq sup over Q bar} is attained, let $(Z_{n})_{n \in \NN}$ be a maximising sequence in $\ol \cQ^\alpha$.  As $\cQ^\alpha$ is uniformly integrable and convex, $\ol \cQ^\alpha$ is convex and $\sigma(L^1, L^\infty)$-sequentially compact by the Dunford-Pettis and the Eberlein-\v{S}mulian theorems. After passing to a subsequence, we may assume that $Z_n$ converges weakly to some $Z^* \in \ol\cQ^\alpha$. Then because the map $\tilde Z \mapsto \mathbb{E}[-\tilde{Z}X]$ is weakly continuous on $\ol \cQ^\alpha$ (by \cite[Proposition C.2]{herdegen2020dual}) and $\overline{\textnormal{co}}\, \alpha$ is also $\sigma(L^1,L^\infty)$-lower semi-continuous by \cite[Theorem 2.2.1]{zalinescu2002convex}, $Z^*$ is a maximiser.
	
	Finally, we show that the inequality in \eqref{eq: sup over Q leq sup over Q bar} is an equality. We may assume without loss of generality that the right hand side of  \eqref{eq: sup over Q leq sup over Q bar} is larger than $-\infty$. Hence, $\overline{\textnormal{co}}\, \alpha(Z^*)$ is finite.  Let $\epsilon > 0$.  Since $\overline{\textnormal{co}}\, \alpha$ is the $L^1$-lower semi-continuous hull of $\textnormal{co}\,\alpha$ and $\textnormal{co}\,\alpha(Z) = \infty$ for $Z \notin \cQ^\alpha$ and $\overline{\textnormal{co}}\, \alpha(Z^*) < \infty$,  by \eqref{eq:lsc hull}, there exists a sequence $(Z_n)_{n \in \NN} \subset \cQ^\alpha$ that converges in $L^1$ to the maximiser $Z^*$ and for which $\lim_{n \to \infty}\textnormal{co}\,\alpha(Z_n) \leq \overline{\textnormal{co}}\,\alpha(Z^*) + \epsilon$.  Using again that the map $\tilde Z \mapsto \mathbb{E}[-\tilde{Z}X]$ is weakly and hence strongly continuous
	yields
	\begin{equation*}
		\rho(X) \geq \lim_{n \to \infty} \{\mathbb{E}[-Z_n X] - \textnormal{co}\,\alpha(Z_n)\} \geq \mathbb{E}[-Z^*X] - \overline{\textnormal{co}}\,\alpha(Z^*) - \epsilon.
	\end{equation*}
	Now the claim follows by letting $\epsilon \to 0$.
\end{proof}

\begin{proof}[Proof of Proposition \ref{prop:cond UI}]
	It is clear by the definition of $C_{\cQ^{\overline{\textnormal{co}}\, \alpha}}$ that $\mathrm{dom}\, f_{\overline{\textnormal{co}}\, \alpha} = C_{\cQ^{\overline{\textnormal{co}}\, \alpha}}$. By Remark \ref{rmk:Q subset Q start start subset Q bar}(a), $\cQ^{\overline{\textnormal{co}}\, \alpha} \subset \ol{\cQ}^\alpha$ and as $\overline{\textnormal{co}}\, \alpha(Z) = \infty$ for $Z \in  \ol{\cQ}^\alpha \setminus \cQ^{\overline{\textnormal{co}}\, \alpha}$ it follows that 
	\begin{equation*}
		f_{\overline{\textnormal{co}}\, \alpha}(c) =
		\inf\{\overline{\textnormal{co}}\, \alpha(Z):Z \in \ol{\cQ}^\alpha \textnormal{ and } \mathbb{E}[-Z(R-r\mathbf{1})] = c\}.
	\end{equation*}
	Since $\ol{\cQ}^\alpha$ is $\sigma(L^1, L^\infty)$-sequentially compact by Dunford-Pettis and the Eberlein-\v{S}mulian theorems and $\overline{\textnormal{co}}\, \alpha$ is $\sigma(L^1, L^\infty)$-lower semi-continuous by \cite[Theorem 2.2.1]{zalinescu2002convex}, it follows that the infimum is attained and (finite) if  $c \in C_{\cQ^{\overline{\textnormal{co}}\, \alpha}}$. Moreover, since $\cQ^{\overline{\textnormal{co}}\, \alpha} \subset \ol{\cQ}^\alpha$ is convex, it follows that 
	$C_{\cQ^{\overline{\textnormal{co}}\, \alpha}} \subset C_{\ol{\cQ}^\alpha} = \{\mathbb{E}[-Z(R-r\mathbf{1})] : Z \in \ol{\mathcal{Q}}^\alpha \}$ is convex and bounded since  $C_{\ol{\mathcal{Q}}^\alpha} = \textnormal{cl}(C_{\cQ^\alpha})$ is a (convex) compact subset of $\RR^d$ by \cite[Proposition 4.5]{herdegen2020dual}.
	
	Next, we show that $f_{\overline{\textnormal{co}}\, \alpha}$ is convex and lower semi-continuous. Convexity follows easily from convexity of $\overline{\textnormal{co}}\, \alpha$. To argue lower semi-continuity let $(c_n)_{n \in \NN}$ be a sequence in $\RR^d$ that converges to $c \in \RR^d$. Without loss of generality, we may assume that each $c_n$ and $c$ lies in $C_{\cQ^{\overline{\textnormal{co}}\, \alpha}}$. Let $(Z_n)_ {n \in \NN}$ in $\cQ^{\overline{\textnormal{co}}\, \alpha}$ be a corresponding sequence of minimisers. Since $\ol{\cQ}^\alpha$ is $\sigma(L^1, L^\infty)$-sequentially compact by the Dunford--Pettis and the Eberlein-\v{S}mulian theorems, after passing to a subsequence, we may assume that $(Z_n)_ {n \in \NN}$ converges weakly to some $Z \in \cQ^{\overline{\textnormal{co}}\, \alpha}$.
	As the map $\tilde Z \mapsto \mathbb{E}[-\tilde{Z}X]$ is $\sigma(L^1, L^\infty)$-continuous on $\ol \cQ^\alpha$ by \cite[Proposition C.2]{herdegen2020dual}, it follows that $\E[-Z(R - r\1)] = c$. By $\sigma(L^1, L^\infty)$-lower semi-continuity of $\overline{\textnormal{co}}\, \alpha$ this implies that $	f_{\overline{\textnormal{co}}\, \alpha}(c) \leq \overline{\textnormal{co}}\, \alpha(Z) \leq \liminf_{n \to \infty} \overline{\textnormal{co}}\, \alpha(Z_n) = \liminf_{n \to \infty} f_{\overline{\textnormal{co}}\, \alpha}(c_n)$.
	
	We proceed to show that $f_{\overline{\textnormal{co}}\, \alpha}$ is the the lower semi-continuous convex hull of $f_\alpha$. To this end, for a function $g:\mathbb{R}^d \to [0,\infty]$, define the map 
	$\alpha^{g}:\cD \to [0,\infty]$ by 
	\begin{equation*}
		\alpha^{g}(Z) = \begin{cases}
			g(\mathbb{E}[-Z(R-r\mathbf{1})]),&\text{if $Z \in \cQ^{\overline{\textnormal{co}}\, \alpha}$} \\ \infty,&\text{otherwise},
		\end{cases}
	\end{equation*}
	If $g$ is convex and lower-semicontinuous, then $\alpha^{g}$ is convex and $\sigma(L^1,L^\infty)$-lower semi-continuous because the map $\tilde Z \mapsto \E[-\tilde Z(R - r\1)]$ is linear and $\sigma(L^1, L^\infty)$-continuous on $\ol \cQ^\alpha \supset \cQ^{\overline{\textnormal{co}}\, \alpha}$ by \cite[Proposition C.2]{herdegen2020dual}.
	
	Seeking a contradiction, suppose now that there exists a convex lower semi-continuous function $g:\mathbb{R}^d \to [0,\infty]$ such that $g \leq f_\alpha$ and  $f_{\overline{\textnormal{co}}\, \alpha}(c^*) < g(c^*)$ for some $c^* \in C_{\cQ^{\overline{\textnormal{co}}\, \alpha}}$. Then
	\begin{equation*}
		\alpha^{g}(Z) \leq 	\alpha^{f_\alpha}(Z) \leq \alpha(Z) ,  \quad Z \in \cQ^{\overline{\textnormal{co}}\, \alpha},
	\end{equation*}
	and hence $\alpha^{g} \leq \overline{\textnormal{co}}\, \alpha$. Let $Z^* \in \cQ^{\overline{\textnormal{co}}\, \alpha}$ be such that
	$\mathbb{E}[-Z^*(R-r\mathbf{1})] = c^*$ and  $\overline{\textnormal{co}}\, \alpha(Z^*) = f_{\overline{\textnormal{co}}\, \alpha}(c^*)$. Then 
	\begin{equation*}
		\overline{\textnormal{co}}\, \alpha(Z^*) = f_{\overline{\textnormal{co}}\, \alpha}(c^*) < g(c^*) = \alpha^g(Z^*) 
	\end{equation*}
	and we arrive at a contradiction.

	Finally, \eqref{eq:relate rho to C hat} follows from Proposition \ref{prop:Cond UI preliminary result}.
\end{proof}

\begin{proof}[Proof of Theorem \ref{thm: no strong reg arb}]
First we show that the market admits strong $\rho$-arbitrage if and only if $\inf_{\pi \in \RR^d}\rho(X_\pi) = - \infty$.  For the nontrivial direction, let $(\pi_n)_{n \geq 1} \subset \RR^d$  be a sequence of portfolios such that $\rho(X_{\pi_n}) \searrow -\infty$. By the dual representation of $\rho$,  this implies that $\mathbb{E}[-X_{\pi_n}] - \alpha(1) \searrow -\infty$, and since $\alpha(1) < \infty$, this gives $\mathbb{E}[X_{\pi_n}] \nearrow \infty$. 

Now let $f_{\overline{\textnormal{co}}\,\alpha}$ be as in Proposition \ref{prop:cond UI}. Since dom\,$f_{\overline{\textnormal{co}}\,\alpha} = C_{\cQ^{\overline{\textnormal{co}}\,\alpha}}$, the convex conjugate of $f_{\overline{\textnormal{co}}\,\alpha}$ is given by
\begin{equation*}
   f_{\overline{\textnormal{co}}\,\alpha}^*(\pi)=\sup_{c \in \RR^d}(\pi \cdot c - f_{\overline{\textnormal{co}}\,\alpha}(c)) = \sup_{c \in C_{\cQ^{\overline{\textnormal{co}}\,\alpha}}}(\pi \cdot c - f_{\overline{\textnormal{co}}\,\alpha}(c)), \quad \pi \in \RR^d.
\end{equation*}
By \eqref{eq:relate rho to C hat}, this implies
\begin{equation}
	\label{eq:pf:thm: no strong reg arb}
	f_{\overline{\textnormal{co}}\,\alpha}^*(\pi) = \rho(X_\pi), \quad \pi \in \RR^d.
\end{equation}
 Since $f_{\overline{\textnormal{co}}\,\alpha}$ is a nonnegative lower semi-continuous convex function, the Fenchel-Moreau theorem (cf.~Appendix \ref{app:convex analysis}) and \eqref{eq:pf:thm: no strong reg arb} give
\begin{equation*}
-f_{\overline{\textnormal{co}}\,\alpha} (\0) = -f_{\overline{\textnormal{co}}\,\alpha}^{**}(\0) = -\sup_{\pi \in \RR^d}(- f_{\overline{\textnormal{co}}\,\alpha^*}(\pi)) = \inf_{\pi \in \RR^d}\rho(X_\pi).
\end{equation*}
The result follows since $\cQ^{\overline{\textnormal{co}}\,\alpha} \cap \mathcal{M} = \emptyset$ if and only if $f_{\overline{\textnormal{co}}\,\alpha}(\mathbf{0}) = \infty$, and
the market admits strong $\rho$-arbitrage if and only if $\inf_{\pi \in \RR^d}\rho(X_\pi) = - \infty$.  The final claim is a consequence of Remark \ref{rmk:Q subset Q start start subset Q bar}(a).
\end{proof}

\begin{proof}[Proof of Theorem \ref{thm: no reg arb equivalence}]
The result follows from Theorem \ref{thm:reg arb:first characterisation} and \cite[Theorem 4.20]{herdegen2020dual}, noting that by Remark \ref{rmk:rho WSLL iff rho bar WSLL}(b), $\rho$ satisfying sensitivity to large expected losses on $L$ implies that $\rho^\infty$ is strictly expectation bounded.
\end{proof}

\begin{proof}[Proof of Theorem \ref{cor:strong rho arb pricing bounds}] 
If the original market $(S^0,S)$ admits strong $\rho$-arbitrage, then $I^s_{\rho}(X) = \emptyset$ by Proposition \ref{prop:basic pricing properties}(b) and $\cQ^{\overline{\textnormal{co}}\,\alpha} \cap \mathcal{M} = \emptyset$ by Theorem \ref{thm: no strong reg arb}.  Whence, \eqref{eq:no strong rho arb price bounds} holds.  

So assume the original market does not admit strong $\rho$-arbitrage.  Consider the market $(S^0,S,S^{d+1})$ where $S^{d+1}_0=1$ and $S^{d+1}_1=X+(1-x)S^0_1$.  Then, by  Theorem \ref{thm: no strong reg arb}, $x$ is a strong $\rho$ consistent price for $X$ if and only if there is an ACMM $Z \in  \mathcal{Q}^{\overline{\textnormal{co}}\,\alpha}$ for the extended market, i.e.,
\begin{equation*}
    S^i_0=\mathbb{E}[ZS^{i}_1/(1+r)], \quad \textnormal{for } i=1,\dots,d+1.
\end{equation*}
In particular, $Z$ is necessarily contained in $\mathcal{Q}^{\overline{\textnormal{co}}\,\alpha} \cap \mathcal{M}$, and we obtain the inclusion $\subset$ in \eqref{eq:no strong rho arb price bounds}.  Conversely, if $S^{d+1}_0 = \mathbb{E}[\hat{Z}S^{d+1}_1/(1+r)]$ for some $\hat{Z} \in \mathcal{Q}^{\overline{\textnormal{co}}\,\alpha} \cap \mathcal{M}$, then this $\hat{Z}$ is also an ACMM for the extended market model, and so the two sets in \eqref{eq:no strong rho arb price bounds} are equal.
\end{proof}

\begin{proof}[Proof of Theorem \ref{cor:rho arb pricing bounds}]
If the original market $(S^0,S)$ admits $\rho$-arbitrage, then $I_{\rho}(X) = \emptyset$ by Proposition \ref{prop:basic pricing properties}(b) and $\tilde{\mathcal{Q}}^\alpha \cap \mathcal{P} = \emptyset$ by Theorem \ref{thm: no reg arb equivalence}.  Whence, \eqref{eq:no rho arb price bounds} holds.  

So assume the original market does not admit $\rho$-arbitrage.  Consider the market $(S^0,S,S^{d+1})$ where $S^{d+1}_0=1$ and $S^{d+1}_1=X+(1-x)S^0_1$.  Then, by Theorem \ref{thm: no reg arb equivalence}, $x \in \RR$ is a $\rho$ consistent price for $X$ if and only if there exists an EMM $Z \in  \tilde{\mathcal{Q}}^\alpha$ for the extended market, i.e.,
\begin{equation*}
    S^i_0=\mathbb{E}[ZS^{i}_1/(1+r)], \quad \textnormal{for } i=1,\dots,d+1.
\end{equation*}
In particular, $Z$ is necessarily contained in $\tilde{\mathcal{Q}}^\alpha \cap \mathcal{P}$, and we obtain the inclusion $\subset$ in \eqref{eq:no rho arb price bounds}.  Conversely, if $S^{d+1}_0 = \mathbb{E}[\hat{Z}S^{d+1}_1/(1+r)]$ for some $\hat{Z} \in \tilde{\mathcal{Q}^\alpha} \cap \mathcal{P}$, then this $\hat{Z}$ is also an EMM for the extended market model, and so the two sets in \eqref{eq:no rho arb price bounds} are equal.
\end{proof}

\begin{proposition}
\label{prop:EW SSLL characterisation}
 $(\cA_{\textnormal{EW}^l})^\infty = H^{\Phi_{l}}_+$ if and only if $a_l = 0$ or $b_l = \infty$. 
\end{proposition}

\begin{proof}
Assume first that $b_l = \infty$ and suppose $X \in H^{\Phi_{l}}$ and $\mathbb{P}[X<0]>0$.  Then since there exists $a \geq 0$ and $b \leq 0$ such that $l(x) \geq ax+b$ for all $x \leq 0$, for any $\lambda >0$ we have
\begin{equation*}
    \textnormal{EW}^l(\lambda X) = \mathbb{E}[l(-\lambda X)] \geq  \mathbb{E}[(-a \lambda X + b) \mathds{1}_{\{X \geq 0\}}] + \mathbb{E}[l(-\lambda X)\mathds{1}_{\{X < 0\}}] \geq \lambda k_1 + c + p l(\lambda k_2)
\end{equation*}
where $k_1:=\mathbb{E}[-aX\mathds{1}_{\{X \geq 0\}}] \leq 0$, $c:=b\mathbb{P}[X \geq 0] \leq 0$, $k_2:=\min\{1,-\essinf(X)/2\} > 0$ and $p:=\mathbb{P}[X \leq -k_2] > 0$.  Now as $\lambda \to \infty$, $(\lambda k_1 + c + p l(\lambda k_2))/\lambda \to \infty$ since $b_l = \infty$.  Therefore, there exists $\tilde{\lambda} \geq 1 $ such that $\textnormal{EW}^l(\tilde{\lambda} X) > 0$ and so $(\cA_{\textnormal{EW}^l})^\infty = H^{\Phi_{l}}_+$.

Now assume that $a_l = 0$ and suppose $X \in H^{\Phi_{l}}$ and $\mathbb{P}[X<0]>0$.  If $X$ is constant, then of course there exists $\tilde{\lambda} \geq 1$ such that $\textnormal{EW}^l(\tilde{\lambda} X) > 0$, so assume $X$ is not constant.  Then there exists $Y \in H^{\Phi_{l}}$ such that $Y \geq X \ \mathbb{P}$-a.s., $\esssup(Y) > 0$ and $\mathbb{P}[Y<0]>0$.  By monotonicity, $\textnormal{EW}^l(\lambda X) \geq \textnormal{EW}^l(\lambda Y)$, and since there exists $a > 0$ and $b \leq 0$ such that $l(x) \geq ax+b$ for all $x \geq 0$, for any $\lambda >0$ we have
\begin{equation*}
     \textnormal{EW}^l(\lambda Y)=\mathbb{E}[l(-\lambda Y)] \geq \mathbb{E}[(-a\lambda Y+b) \mathds{1}_{\{Y<0\}}] + \mathbb{E}[l(-\lambda Y) \mathds{1}_{\{Y \geq 0\}}] \geq \lambda j_1 + d + q l(\lambda j_2)
\end{equation*}
where $j_1:=\mathbb{E}[-aY \mathds{1}_{\{Y < 0\}}] > 0$, $d:=b\mathbb{P}[Y < 0] \leq 0$, $j_2:=\max\{-1,-\esssup(Y)/2\} < 0$ and $q:=\mathbb{P}[Y \geq -c_2] > 0$.  As $\lambda \to \infty$, $(\lambda j_1 + d + q l(\lambda j_2))/\lambda \to j_1 > 0$ since $a_l = 0$.  Therefore, there exists $\tilde{\lambda} \geq 1 $ such that $\textnormal{EW}^l(\tilde{\lambda} X) > 0$ and so $(\cA_{\textnormal{EW}^l})^\infty = H^{\Phi_{l}}_+$.

\medskip
On the other hand, if $b_l \neq \infty$ and $a_l \neq 0$, then define the loss function $\tilde{l}:\RR \to \RR$ by 
\begin{equation*}
    \tilde{l}(x) = \begin{cases}
        b_l x,&\textnormal{if $x \geq 0$,} \\ a_l x,&\textnormal{if $x < 0$}.
    \end{cases}
\end{equation*}
Then $\tilde{l} \geq l$, so $\textnormal{EW}^{\tilde{l}} \geq \textnormal{EW}^{l}$, and to complete the proof, it suffices to find $X \in H^{\Phi_l} = H^{\Phi_{\tilde{l}}} = L^1$ such that $\mathbb{P}[X < 0] > 0$ and $\textnormal{EW}^{\tilde{l}}(\lambda X) \leq 0$ for all $\lambda > 0$.  To that end, let $A \in \mathcal{F}$ be a nontrivial event, $p:=\mathbb{P}[A] \in (0,1)$ and consider the random variable $X = \alpha \mathds{1}_A - \beta \mathds{1}_{A^c}$ where $\alpha ,\beta > 0$ satisfy $-p a_l \alpha + (1-p) b_l \beta < 0$.  Then $X \in L^1$, $\mathbb{P}[X < 0] > 0$ and for any $\lambda > 0$,
\begin{equation*}
\textnormal{EW}^{\tilde{l}}(\lambda X) = \mathbb{E}[\tilde{l}(-\lambda X)] = \lambda [-p a_l \alpha + (1-p) b_l \beta] \leq 0.\qedhere
\end{equation*}
\end{proof}

\begin{proof}[Proof of Proposition \ref{prop:dual rep SR}]
\label{pf:prop:dual rep SR}
If $l|_{\RR_-} = 0$, then $\textnormal{SR}^l \equiv \WC$.  In this case, \eqref{eq:dual rep SR} holds since $l^*|_{[0,1]} = 0$ and $\alpha^l(Z) = 0$ for $Z \in D \cap L^\infty$.  Otherwise,  if  $l|_{\RR_-} \neq 0$, then $\textnormal{SR}^l$ is a real-valued convex risk measure on $H^{\Phi_l}$ and \eqref{eq:dual rep SR} follows from \cite[Theorem 4.3]{cheridito2009risk}, and the proof of  \cite[Theorem 10]{follmer2002convex}. 
\end{proof}

\begin{comment}

\begin{proposition}
	\label{prop:Luxembourg:norm}
	Suppose $\Phi$ is a Young function and $X \geq 0$. Then, 
	\begin{equation*}
		\lVert X \rVert_{\Phi} \leq 1 + \mathbb{E}[\Phi(X)].
	\end{equation*}
\end{proposition}

\begin{proof}
	Set $k:=1+\mathbb{E}[\Phi(X)] \geq 1$. We may assume without loss of generality that $k < \infty$. Since $\Phi(X/k) \leq \Phi(X)/k$ by nonnegativity of $X$, $\Phi(0)=0$ and convexity of $\Phi$, we obtain
	\begin{equation*}
		\mathbb{E}[\Phi(X/k)] \leq \mathbb{E}[\Phi(X)]/k \leq 1.
	\end{equation*}
	Thus, $\lVert X \rVert_{\Phi} \leq k$ by the definition of the Luxemburg norm.
\end{proof}

\end{comment}

\begin{proposition}
\label{prop:interior SR}
    Let $l$ be a loss function and assume that $0 < a_l < b_l < \infty$.  Consider the penalty function $\alpha^{l}(Z) := \inf_{\lambda > 0} \tfrac{1}{\lambda} \mathbb{E}[l^*(\lambda Z)]$ for $Z \in \cD$.  Then
    \begin{equation*}
        \tilde{\mathcal{Q}}^{\alpha^l} = \{ Z \in \mathcal{D} : \textnormal{there exists $k > 0$ and $\epsilon > 0$ such that } a_l+\epsilon < kZ < b_l - \epsilon \ \mathbb{P}\textnormal{-a.s.}  \}
    \end{equation*}
    is a nonempty subset of $\mathcal{Q}^{\alpha^l}$ satisfying Conditions POS, MIX and INT.
\end{proposition}

\begin{proof}
First, since $(a_l, b_l) \subset \textnormal{dom} \, l^* \subset [a_l, b_l]$ and $l^*$ is bounded on any compact subset of $(a_l, b_l)$, $\tilde{\mathcal{Q}}^{\alpha^l} \subset \mathcal{Q}^{\alpha^l}$. Moreover, it is clear that $1 \in \tilde{\mathcal{Q}}^{\alpha^l}$, and by definition $\tilde{\mathcal{Q}}^{\alpha^l}$ satisfies Condition POS.
    
We next show Condition MIX.  To that end, let $Z \in \mathcal{Q}^{\alpha^l}$, $\tilde{Z} \in \tilde{\mathcal{Q}}^{\alpha^l}$ and $\lambda \in (0,1)$.  Then there exists $k,\tilde{k},\tilde{\epsilon} > 0$ such that $\mathbb{E}[l^*(kZ)] < \infty$ and $a_l + \tilde{\epsilon} \leq \tilde{k}\tilde{Z} \leq b_l - \tilde{\epsilon} \ \mathbb{P}$-a.s.  In particular, since $a_l \leq kZ \leq b_l$ $\mathbb{P}$-a.s., it follows that
\begin{equation*}
    a_l + \epsilon^* < k^*(\lambda Z + (1-\lambda)\tilde{Z}) < b_l - \epsilon^* \quad \mathbb{P}\textnormal{-a.s.,}
\end{equation*}
where $k^*:=k\tilde{k}/(\lambda \tilde{k} + (1-\lambda)k) > 0$ and $\epsilon^*:=(1-\lambda)k^*\tilde{\epsilon} > 0$.  Therefore, $\lambda Z + (1-\lambda)\tilde{Z} \in \tilde{\mathcal{Q}}^{\alpha^l}$ and $\tilde{\mathcal{Q}}^{\alpha^l}$ satisfies Condition MIX.

We now show Condition INT is satisfied.  Let $\tilde{Z} \in \tilde{\mathcal{Q}}^{\alpha^l}$, set $\mathcal{E}:=\cD \cap L^\infty$ and let $Z \in \mathcal{E}$.  Then there exists $k, \epsilon > 0$ such that $\essinf k\tilde{Z} > a_l + \epsilon$ and $\esssup k\tilde{Z} < b_l - \epsilon$.  Let
\begin{equation*}
    \lambda_1:=\begin{cases}
        \frac{b_l - \epsilon - k\lVert \tilde{Z} \rVert_\infty}{k(\lVert Z \rVert_\infty - \lVert \tilde{Z} \rVert_\infty)},&\textnormal{if $\lVert Z \rVert_\infty > \lVert \tilde{Z} \rVert_\infty$}, \\ \tfrac{1}{2},&\textnormal{otherwise},
    \end{cases}
\end{equation*}
and $\lambda_2:=1-(a_l + \epsilon)/\essinf k\tilde{Z}$.  Then setting $\lambda:=\min\{\lambda_1,\lambda_2\}$ yields $\essinf k (\lambda Z+(1-\lambda)\tilde{Z}) \geq a_l + \epsilon$ and $\esssup k (\lambda Z+(1-\lambda)\tilde{Z}) \leq b_l - \epsilon$.  Therefore, $\lambda Z + (1-\lambda)\tilde{Z} \in \mathcal{Q}^{\alpha^l}$ and $\tilde{\mathcal{Q}}^{\alpha^l}$ satisfies Condition INT.
\end{proof}

\begin{proposition}
\label{prop:Q co alpha l SRl}
     Let $l$ be a loss function where $a_l > 0$ and $b_l < \infty$. Then the penalty function $\alpha^{l}(Z) := \inf_{\lambda > 0} \tfrac{1}{\lambda} \mathbb{E}[l^*(\lambda Z)]$ is convex and $L^1$-lower semi-continuous.  Thus, $\mathcal{Q}^{\overline{\textnormal{co}} \, \alpha^l} = \mathcal{Q}^{\alpha^l}$. 
\end{proposition}

\begin{proof}
First note that $\alpha^l$ is the minimal penalty function for SR$^l$ by the proof of \cite[Theorem 10]{follmer2002convex}, and $\textnormal{SR}^l$ is a real-valued convex risk measure on $H^{\Phi_l}$.  Therefore $\alpha^l$ is convex by \cite[Theorem 4.3]{cheridito2009risk}.  To prove $L^1$-lower semi-continuity, consider an arbitrary sequence $(Z_n)_{n \geq 1} \subset \cQ^{\alpha^l}$ that converges in $L^1$ to $Z$ and assume $\liminf_{n \to \infty} \alpha^l(Z_n) \leq y$.  By restricting to a subsequence and relabelling, we may assume without loss of generality $Z_n \to Z$ $\mathbb{P}$-a.s.  Thus, $Z \in \cD$ and $\alpha^l(Z)$ is well-defined.  To complete the proof, we must show $\alpha^l(Z) \leq y$.  Now for any $n \in \NN$: since $Z_n \in \cQ^{\alpha^l}$, $l^*$ is nonnegative and $\textnormal{dom} \, l^* \subset [a_l, b_l]$, that means $\alpha^l(Z_n) = \tfrac{1}{\lambda_n^*} \mathbb{E}[l^*(\lambda_n^* Z_n)]$ for some $\lambda_n^* \in [a_l, b_l]$ (i.e., the infimum is attained -- this can be shown by taking a minimising sequence $(\lambda_k)_{k \geq 1} \subset [a_l, b_l]$, restricting to a convergent subsequence and using Fatou's lemma).  By restricting to a further subsequence and relabelling, we may assume that $\lambda_n^* \to \lambda^*$.  Then $\tfrac{1}{\lambda_n^*} l^*(\lambda_n^* Z_n) \to \tfrac{1}{\lambda^*} l^*(\lambda^* Z) \ \mathbb{P}$-a.s.~(since $l^*$ is continuous on its effective domain and $\lambda_n Z_n \in \textnormal{dom} \, l^*$ $\mathbb{P}$-a.s.), and so by Fatou's lemma, which we can apply since $l^*$ is nonnegative:
\begin{equation*}
    y \geq \liminf_{n \to \infty} \alpha^l(Z_n) = \liminf_{n \to \infty} \tfrac{1}{\lambda_n^*} \mathbb{E}[l^*(\lambda_n^* Z_n)] \geq  \tfrac{1}{\lambda^*} \mathbb{E}[l^*(\lambda^* Z)] \geq \alpha^l(Z).\qedhere
\end{equation*}
\end{proof}

\begin{comment}
\begin{proof}
We first show $\alpha^l$ is bounded on its effective domain.  So let $Z \in \mathcal{Q}^{\alpha^l}$ and note that there exists $k > 0$ such that $kZ \in [a_l,b_l] \ \mathbb{P}$-a.s.  Now for any $\lambda \in (0,a_l)$, $\mathbb{E}[l^*(\lambda Z)] = \infty$ since $\mathbb{P}[Z \leq 1] > 0$ and $\textnormal{dom} \, l^* = [a_l,b_l]$.  Similarly, since $\mathbb{P}[Z \geq 1] > 0$, for any $\lambda \in (b_l,\infty)$, $\mathbb{E}[l^*(\lambda Z)] = \infty$.  Thus, $k \in [a_l,b_l]$ and since $l^*$ is nondecreasing on $[a_l,b_l]$
\begin{equation*}
    \alpha^{l}(Z)  \leq \frac{1}{k}\mathbb{E}[l^*(k Z)] \leq \frac{1}{a_l}l^*(b_l) < \infty.
\end{equation*}
It follows by Remark \ref{rmk:Q subset Q start start subset Q bar}(a) that $\mathcal{Q}^{\overline{\textnormal{co}} \, \alpha^l} = \bar{\cQ}^{\alpha^l}$, and so to complete the proof we must show that $\mathcal{Q}^{\alpha^l}$ is $L^1$-closed.  To that end, let $(Z_{n})_{n \geq 1} \subset \mathcal{Q}^{\alpha^l}$ and assume the sequence converges in $L^1$ to $Z^*$.  We must show $Z^* \in \mathcal{Q}^{\alpha^l}$.  Clearly $Z^* \in \mathcal{D}$, and by what we have shown above, there exists $k_n \in [a_l,b_l]$ such that $k_n Z_n \in [a_l,b_l] \ \mathbb{P}$-a.s.  By restricting to a subsequence we may assume without loss of generality that $k_n$ converges to $k^* \in [a_l,b_l]$ and $k_n Z_n$ converges to $k^* Z^*$ $\mathbb{P}$-a.s.  Thus, $k^* Z^* \in [a_l,b_l] \ \mathbb{P}$-a.s., and $Z^* \in \mathcal{Q}^{\alpha^l}$ as desired. 
\end{proof}
\end{comment}

\begin{proof}[Proof of Proposition \ref{prop:OCE dual rep} ]
    This follows from: \cite[Equation (5.23)]{cheridito2009risk} and Remark \ref{rmk:OCE condition} in the case that $l(x) > x$ for all $x$ with $|x|$ sufficiently large; and from $\alpha^l(1) = 0$ and $\alpha^l(Z) = \infty$ for all $Z \in \cD \setminus \{1\}$ in the case that $l$ is equal to the identity either on $\RR_+$ or $\RR_-$.
\begin{comment}
and Remark \ref{rmk:OCE condition} together with the fact that 
    \begin{equation*}
    \textnormal{dom}\,(l^*|_{\RR_+}) \begin{cases}
    \supset (0,\infty),&\text{if $a_l = 0$ and $b_l = \infty$,} \\ \subset [a_l,b_l],&\text{if $a_l > 0$ or $b_l < \infty$}. %\\ = \{ Z \in \mathcal{D} : 0 < Z \leq b_l \ \mathbb{P}\textnormal{-a.s.}\},&\text{if $a_l = 0$ and $b_l < \infty$ and $l$ is not bounded below}, \\ = \{Z \in \mathcal{D}: a_l \leq Z \leq b_l \ \mathbb{P}\textnormal{-a.s.} \},&\text{otherwise}.
    \end{cases}
    \qedhere
\end{equation*}
\end{comment}
\end{proof}

\begin{proposition}
\label{prop:interior OCE}
  Let $l$ be a loss function and assume that either $ a_l  > 0$, or $b_l  < \infty$ and $a_l < 1 < b_l$.  Define $\alpha^{l}(Z) := \mathbb{E}[l^*(Z)]$.  Then 
\begin{equation*}
    \tilde{\mathcal{Q}}^{\alpha^l} =  \{ Z \in \mathcal{Q}^{\alpha^l} : a_l + \epsilon < Z < b_l - \epsilon \ \mathbb{P}\textnormal{-a.s.~for some $\epsilon > 0$}\},
\end{equation*}
is a nonempty subset of $\mathcal{Q}^{\alpha^l}$ satisfying Conditions POS, MIX and INT.
\end{proposition}

\begin{proof}
It is clear that $1 \in \tilde{\mathcal{Q}}^{\alpha^l} \subset \mathcal{Q}^{\alpha^l}$, and by definition $\tilde{\mathcal{Q}}^{\alpha^l}$ satisfies Condition POS.  For the rest of the proof, it is useful to recall that $(a_l, b_l) \subset \textnormal{dom} \, l^* \subset [a_l, b_l]$ and $l^*$ is bounded on any compact subset of $(a_l, b_l)$.
    
    To show Condition MIX, let $Z \in \mathcal{Q}^{\alpha^l}$, $\tilde{Z} \in \tilde{\mathcal{Q}}^{\alpha^l}$ and $\lambda \in (0,1)$.  Since $l^*$ is convex, $\mathbb{E}[l^*(\lambda Z + (1-\lambda) \tilde{Z})] < \infty$ so $\lambda Z + (1-\lambda) \tilde{Z} \in \mathcal{Q}^{\alpha^l}$.  Furthermore, since $a_l \leq Z \leq b_l \ \mathbb{P}$-a.s.~and $a_l + \epsilon < \tilde{Z} < b_l - \epsilon \ \mathbb{P}$-a.s.~for some $\epsilon > 0$, it follows that 
    \begin{equation*}
        a_l + (1-\lambda) \epsilon < \lambda Z + (1-\lambda)\tilde{Z} < b_l - (1-\lambda)\epsilon \quad \mathbb{P}\textnormal{-a.s.}
    \end{equation*}
Therefore, $\lambda Z + (1-\lambda)\tilde{Z} \in \tilde{\mathcal{Q}}^{\alpha^l}$ and $\tilde{\mathcal{Q}}^{\alpha^l}$ satisfies Condition MIX.

Finally we show Condition INT is satisfied.  Assume first that $b_l=\infty$ and let $\tilde{Z} \in \tilde{\mathcal{Q}}^{\alpha^l}$.  Set $\mathcal{E}:=\mathcal{D} \cap L^\infty$ and $Z \in \mathcal{E}$.  Then there exists $\epsilon > 0$ such that $\tilde{Z} > a_l + \epsilon \ \mathbb{P}$-a.s.  By choosing $\lambda \in (0,\tfrac{1}{2}\epsilon/(a_l+\epsilon)]$, it follows that
\begin{equation*}
    \lambda Z + (1-\lambda) \tilde{Z} \geq (1-\lambda) (a_l+\epsilon) \geq a_l + \tfrac{\epsilon}{2} \quad \mathbb{P}\textnormal{-a.s.}
\end{equation*}
Now since $l^*$ is convex, real-valued on $(a_l,\infty)$ and its minimum is $l^*(1)=0$, it follows that $l^*$ is nonincreasing on $(a_l,1)$ and nondecreasing on $(1,\infty)$.  Whence, setting $A:=\{\lambda Z + (1-\lambda)\tilde{Z} \leq 1\}$ and $B:=\{Z \leq \tilde{Z}\}$, we have 
\begin{align*}
    \mathbb{E}[l^*(\lambda Z + (1-\lambda)\tilde{Z})] & = \mathbb{E}[l^*(\lambda Z + (1-\lambda)\tilde{Z})\mathds{1}_A] + \mathbb{E}[l^*(\lambda Z + (1-\lambda)\tilde{Z})\mathds{1}_{A^c}] \\ & \leq l^*(a_l+\tfrac{\epsilon}{2}) + \mathbb{E}[l^*(\lambda Z + (1-\lambda)\tilde{Z})\mathds{1}_{A^c \cap B}] + \mathbb{E}[l^*(\lambda Z + (1-\lambda)\tilde{Z})\mathds{1}_{A^c \cap B^c}] \\ & \leq l^*(a_l+\tfrac{\epsilon}{2}) + \mathbb{E}[l^*(\tilde{Z})] + l^*(\lVert Z \rVert_\infty) < \infty,
\end{align*}
Therefore, $\lambda Z + (1-\lambda)\tilde{Z} \in \mathcal{Q}^{\alpha^l}$ and $\tilde{\mathcal{Q}}^{\alpha^l}$ satisfies Condition INT when $b_l=\infty$.  Now assume $1 < b_l < \infty$ and let $\tilde{Z} \in \tilde{\mathcal{Q}}^{\alpha^l}$.  Set $\mathcal{E}:=\mathcal{D} \cap L^\infty$ and $Z \in \mathcal{E}$.  Then there exists $\epsilon > 0$ such that $a_l + \epsilon < \tilde{Z} < b_l - \epsilon$.  By choosing $\lambda \in (0,\tfrac{1}{2}\epsilon/(a_l+\epsilon)]$ if $\lVert Z \rVert_\infty < b_l$ and choosing $\lambda \in (0,\min\{\tfrac{1}{2}\epsilon/(a_l+\epsilon),\epsilon/(\tfrac{1}{2}\epsilon+\lVert Z \rVert_\infty - b_l)\}]$ otherwise; it follows that there exists $\epsilon' > 0$ such that $a_l + \epsilon' \leq \lambda Z + (1-\lambda) \tilde{Z} \leq b_l - \epsilon' \ \mathbb{P}\textnormal{-a.s.}$
Therefore, $\lambda Z + (1-\lambda)\tilde{Z} \in \mathcal{Q}^{\alpha^l}$ and $\tilde{\mathcal{Q}}^{\alpha^l}$ satisfies Condition INT when $1 < b_l < \infty$.
\end{proof}

\begin{proof}[Proof of Corollary \ref{cor:OCE:SRA}]
(a) This follows from Proposition \ref{prop:interior OCE} and Theorem \ref{thm: no reg arb equivalence}, noting that Condition I follows from the generalised H{\"o}lder inequality; see e.g.~\cite[Equation (B.1)]{herdegen2020dual}.

(b) First, note that since $(a_l, b_l) \subset \textnormal{dom} \, l^* \subset [a_l, b_l]$ we have
\begin{equation*}
    \{ Z \in \cD : a_l < \essinf Z \textnormal{ and } \esssup Z < b_l \} \subset \cQ^{\alpha^l} \subset \{ Z \in \mathcal{D} : a_l \leq Z \leq b_l \ \mathbb{P}\textnormal{-a.s.}\}.
\end{equation*}
Thus, Condition UI holds if $b_l < \infty$. Using this, the result follows from Theorem \ref{thm: no strong reg arb} if we can show the penalty function $\cD \ni Z \mapsto \alpha^l(Z) := \mathbb{E}[l^*(Z)]$ is convex and $L^1$-lower semi-continuous.  Since $l^*$ is convex, $\alpha^l$ is convex.  To show $L^1$-lower semi-continuity, it suffices to consider a sequence $(Z_n)_{n \geq 1} \subset \cQ^{\alpha^l}$ that converges in $L^1$ to $Z$ and show $\alpha^l(Z) \leq \liminf_{n \to \infty} \alpha^l(Z_n)$.  By restricting to a subsequence and relabelling, we may assume $Z_n$ converges to $Z$ $\mathbb{P}$-a.s., and hence $Z \in \cD$.  Moreover, $l^*(Z_n)$ is a nonnegative sequence that converges to $l^*(Z)$ $\mathbb{P}$-a.s., since $l^*$ is nonnegative, continuous on its effective domain and $Z_n \in \textnormal{dom} \, l^*$ $\mathbb{P}$-a.s.  Applying Fatou's lemma yields:
\begin{equation*}
    \alpha^l(Z) = \mathbb{E}[l^*(Z)] \leq \liminf_{n \to \infty} \mathbb{E}[l^*(Z_n)] = \liminf_{n \to \infty} \alpha^l(Z_n).\qedhere
\end{equation*}
\end{proof}

\begin{proposition}
\label{prop:VaRg acceptance set}
Assume $g: (0,1) \to [0,\infty]$ is a nonincreasing function with $\inf g = 0$.  If $g$ is real-valued, then $(\cA_{\VaR^g})^\infty = L^1_+$.  If $g$ is not real-valued and the probability space is atomless, then $(\cA_{\VaR^g})^\infty \supsetneq L^1_+$.
\end{proposition}

\begin{proof}
Let $X \in L^1$.  By definition, $\VaR^g(\lambda X) \leq 0$ for all $\lambda \in (0,\infty)$ if and only if $\VaR^\alpha(\lambda X) \leq g(\alpha)$ for all $\lambda \in (0,\infty)$ and $\alpha \in (0,1)$.  By the positive homogeneity of $\VaR$, this is equivalent to $\VaR^\alpha(X) \leq 0$ for all $\alpha \in \textnormal{dom} \, g$.

    If $g$ is real-valued, this implies that $0 \geq \lim_{\alpha \to 0} \VaR^\alpha(X) = \WC(X)$ for all $X \in (\cA_{\VaR^g})^\infty$.  We may conclude that $(\cA_{\VaR^g})^\infty = L^1_+$.

    If $g$ is not real-valued, then there exists $\beta \in (0,1)$ such that $\inf \textnormal{dom} \, g = \beta$. It follows that $(\cA_{\VaR^g})^\infty \supset \cA_{\VaR^\beta}$ because $0 \geq \VaR^\beta(X) \geq \VaR^\alpha(X)$ for any $X \in \cA_{\VaR^\beta}$ and $\alpha \in \textnormal{dom} \, g$. Since the probability space is atomless, it is straightforward to check that  $\cA_{\VaR^\beta} \supsetneq L^1_+$.  %Indeed, by nonatomicity, we can find an event $A\in\cF$ such that $\mathbb{P}[A] \in (0,\alpha]$. Let $X:=-\mathds{1}_A$ and observe that $\mathbb{P}[X<0]=\mathbb{P}[A]>0$ and $\VaR_\alpha(X)=0$.  Therefore, $(\cA_{\VaR^g})^\infty \supset \cA_{\VaR^\beta} \supsetneq L^1_+$.
\end{proof}

\begin{proof}[Proof of Proposition \ref{prop:dual rep g adjusted ES}]
The dual representation has been shown in \cite[Proposition 3.7]{burzoni2020adjusted}.  \eqref{eq:dual set g adjusted ES} follows directly from the definition of $\alpha^g$, which  also gives convexity of  $\cQ^{\alpha^g}$.
\end{proof}

\begin{proposition}
	\label{prop:g bounded dom g iff g hat bounded dom g hat}
	Let $\beta \in (0,1)$ and $g \in \cG_\beta \setminus \cG_\beta^\infty$.  Define $\hat{g}:(0,1] \to [0,\infty]$ by $\hat{g}(x) = \overline{\textnormal{co}}\,\tilde{g}(1/x)$ where $\tilde{g}:[1,\infty) \to [0,\infty]$ is given by $\tilde{g}(x) = g(1/x)$.  Then $\hat{g} \in \cG_\beta \setminus \cG_\beta^\infty$ and $\hat g(\beta) = \infty$.
\end{proposition}

\begin{proof}
	Since $g \in \cG_\beta$, it follows that $\tilde{g}$ is nondecreasing, real-valued on $[1,1/\beta)$, $\infty$ on $(1/\beta,\infty)$ and $\tilde{g}(1) = 0$. By the definition of the lower semi-continuous convex hull, it is not difficult to check that $\overline{\textnormal{co}}\,\tilde{g}$ has the same properties and so $\hat g \in \cG_\beta$. It remains to show that $\overline{\textnormal{co}}\,\tilde{g}(1/\beta) = \infty$.
	
Seeking a contradiction, suppose that $\overline{\textnormal{co}}\,\tilde{g}(1/\beta) =: k < \infty$.  As $\overline{\textnormal{co}}\,\tilde{g}$ is a proper lower semi-continuous convex function, the Fenchel-Moreau theorem gives  $\overline{\textnormal{co}}\,\tilde{g} = \tilde{g}^{**}$ where $\tilde{g}^{**}$ is the biconjugate of $\tilde{g}$.	Since $\tilde{g}$ is nondecreasing and $\lim_{x \uparrow 1/\beta} \tilde{g}(x) = \infty$, there exists $c \in [1,1/\beta)$ such that $\tilde{g}(x) > k+1$ for all $x \in [c,1/\beta)$.  Thus, the affine (and continuous) function $a:[1, \infty) \to \RR$ with $a(c) = 0$ and $a(1/\beta) = k+1$ satisfies $a \leq \tilde{g}$ and $a(1/\beta) > k = \tilde{g}^{**}(1/\beta)$.  This is in contradiction to the fact that by \eqref{eq:biconjugate}, $\tilde{g}^{**}$ dominates any affine (and continuous) function dominated by $\tilde g$.
\end{proof}

\begin{proposition}
\label{prop:Q g star star}
    Let $\beta \in (0,1)$ and $g \in \cG_\beta$.  Let $\overline{\textnormal{co}}\,\alpha^g$ be the  $L^1$-lower semi-continuous convex hull of $\alpha^g$.  Then its effective domain is given by
    \begin{equation*}
    \cQ^{\overline{\textnormal{co}}\,\alpha^g} = \begin{cases}
   \{Z \in \cD: \lVert Z \rVert_{\infty} \leq \tfrac{1}{\beta}\},&\text{if $g \in \cG_\beta^\infty$}, \\ \{Z \in \cD: \lVert Z \rVert_{\infty} < \tfrac{1}{\beta}\},&\text{if $g \in \cG_\beta \setminus \cG_\beta^\infty$}.
   \end{cases}
\end{equation*}
\end{proposition}

\begin{proof}
If $g \in \cG_\beta^\infty$, then the result follows from Remark \ref{rmk:Q subset Q start start subset Q bar}(a).  So assume $g \in \cG_\beta \setminus \cG_\beta^\infty$.  Define the function $\hat{g}:(0,1] \to [0,\infty]$ by $\hat{g}(x) = \overline{\textnormal{co}}\,\tilde{g}(1/x)$ where $\tilde{g}:[1,\infty) \to [0,\infty]$ is given by $\tilde{g}(x) = g(1/x)$. By Proposition \ref{prop:g bounded dom g iff g hat bounded dom g hat}, $\hat{g} \in \cG_\beta \setminus \cG_\beta^\infty$ and $\hat{g}(\beta) = \infty$. Moreover, $\hat g(x) \leq \tilde g(1/x) = g(x)$ for $x \in (0, 1]$. Moreover, by the fact that $\overline{\textnormal{co}}\,\tilde{g}$ is convex and lower semi-continuous, nondecreasing, real-valued on $[1,1/\beta)$ and $\infty$ on $[1/\beta, \infty]$, it follows that 
$\alpha^{\hat{g}}:\cD \to [0,\infty]$, given by 
\begin{equation*}
    \alpha^{\hat{g}}(Z) = \begin{cases}
    \hat{g}(\lVert Z \rVert_\infty^{-1}) = \overline{\textnormal{co}}\,\tilde{g}(\lVert Z \rVert_\infty),&\text{if $Z \in \cQ^{\hat{g}} = \{Z \in \cD: \lVert Z \rVert_{\infty} < 1/\beta\},$} \\ \infty,&\text{otherwise},
    \end{cases}
\end{equation*}
is convex and $L^1$-lower semi-continuous. Thus, $\alpha^g \geq \overline{\textnormal{co}}\,\alpha^g \geq \overline{\textnormal{co}}\,\alpha^{\hat{g}} = \alpha^{\hat{g}}$, which implies $\cQ^{\alpha^g} \subset \cQ^{\overline{\textnormal{co}}\,\alpha^g} \subset \cQ^{ \alpha^{\hat{g}}}$. Since $\cQ^{\alpha^g} = \cQ^{ \alpha^{\hat{g}}} = \{Z \in \cD: \lVert Z \rVert_{\infty} < 1/\beta\}$, the result follows.
\begin{comment}
\footnote{One may be tempted to think that $\alpha^g = \alpha^{\hat g}$ in the proof of Proposition \ref{prop:Q g star star} and so $\ES^g(X) = \ES^{\hat{g}}(X)$ for all $X \in L^1$. However, this is not true. For example, consider the function $g \in \cG_{0.5} \setminus \cG^\infty_{0.5}$ defined by 
 \begin{equation*}
 	g(x) := \begin{cases}
 		\infty & \text{ if } x \in (0, 0.5], \\
 		\tfrac{1}{2x-1} - \tfrac{109}{25}& \text{ if } x \in (0.5, 0.6], \\
 		 	1-x^2& \text{ if } x \in (0.6, 1],
 		\end{cases}
 	\end{equation*}
 and consider the random variable $Y = -X$ where $X$ has an exponential distribution with parameter $0.7$.  Then, 
  \begin{equation*}
\hat g 	(x) := \begin{cases}
 		\infty & \text{ if } x \in (0, 0.5], \\
 		\tfrac{1}{2x-1} - \tfrac{109}{25}& \text{ if } x \in (0.5, 0.6], \\
 		\tfrac{24}{25x} - \tfrac{24}{25} & \text{ if } x \in (0.6, 1],
 	\end{cases}
 \end{equation*}
 and $\ES^g(Y) \approx 1.52 \neq 1.53 \approx \ES^{\hat{g}}(Y)$.}
\end{comment}
\end{proof}

\small
\bibliography{convexregarb} 

\providecommand{\bysame}{\leavevmode\hbox to3em{\hrulefill}\thinspace}
\providecommand{\MR}{\relax\ifhmode\unskip\space\fi MR }
\providecommand{\MRhref}[2]{%
  \href{http://www.ams.org/mathscinet-getitem?mr=#1}{#2}
}
\providecommand{\href}[2]{#2}
\begin{thebibliography}{10}

\bibitem{alexander2002economic}
G.~J. Alexander and A.~M. Baptista, \emph{Economic implications of using a
  mean-var model for portfolio selection: A comparison with mean-variance
  analysis}, J. Econ. Dyn. Control \textbf{26} (2002), no.~7-8, 1159--1193.

\bibitem{guide2006infinite}
C.~Aliprantis and K.~Border, \emph{Infinite dimensional analysis: A
  hitchhiker’s guide}, Springer, 2013.

\bibitem{arai2011good}
T.~Arai, \emph{Good deal bounds induced by shortfall risk}, SIAM J. Financial
  Math. \textbf{2} (2011), no.~1, 1--21.

\bibitem{arduca2020fundamental}
M.~Arduca and C.~Munari, \emph{Fundamental theorem of asset pricing with
  acceptable risk in markets with frictions}, Finance and Stochastics
  \textbf{27} (2023), no.~3, 831--862.

\bibitem{armstrong2019statistical}
J.~Armstrong and D.~Brigo, \emph{Coherent risk measures alone are ineffective
  in constraining portfolio losses}, J. Bank. Finance \textbf{140} (2022),
  106315.

\bibitem{armstrong2019risk}
J.~Armstrong and D.~\vspace{0mm} Brigo, \emph{Risk managing tail-risk seekers:
  Var and expected shortfall vs s-shaped utility}, J. Bank. Finance
  \textbf{101} (2019), 122--135.

\bibitem{ben1986expected}
A.~Ben-Tal and M.~Teboulle, \emph{Expected utility, penalty functions, and
  duality in stochastic nonlinear programming}, Manage. Sci. \textbf{32}
  (1986), no.~11, 1445--1466.

\bibitem{ben2007old}
A.~Ben-Tal and M.~\vspace{0mm}Teboulle, \emph{An old-new concept of convex risk
  measures: The optimized certainty equivalent}, Math. Finance \textbf{17}
  (2007), no.~3, 449--476.

\bibitem{bernardo2000gain}
A.~E. Bernardo and O.~Ledoit, \emph{Gain, loss, and asset pricing}, J. Polit.
  Econ. \textbf{108} (2000), no.~1, 144--172.

\bibitem{bignozzi2020risk}
V.~Bignozzi, M.~Burzoni, and C.~Munari, \emph{Risk measures based on benchmark
  loss distributions}, Journal of Risk and Insurance \textbf{87} (2020), no.~2,
  437--475.

\bibitem{bosch2006reflections}
A.~Bosch-Dom{\`e}nech and J.~Silvestre, \emph{Reflections on gains and losses:
  A $2 \times 2 \times 7$ experiment}, J. Risk Uncertainty \textbf{33} (2006),
  no.~3, 217--235.

\bibitem{bosch2010averting}
A.~Bosch-Dom{\`e}nech and J.\vspace{0.0mm} Silvestre, \emph{Averting risk in
  the face of large losses: {B}ernoulli vs. {T}versky and {K}ahneman}, Econ.
  Lett. \textbf{107} (2010), no.~2, 180--182.

\bibitem{burzoni2020adjusted}
M.~Burzoni, C.~Munari, and R.~Wang, \emph{Adjusted expected shortfall}, J.
  Bank. Finance \textbf{134} (2022), 106297.

\bibitem{castagnoli2021star}
E.~Castagnoli, G.~Cattelan, F.~Maccheroni, C.~Tebaldi, and R.~Wang,
  \emph{Star-shaped risk measures}, Oper. Res. \textbf{70} (2022), no.~5,
  2637--2654.

\bibitem{vcerny2002theory}
A.~{\v{C}}ern\`{y} and S.~Hodges, \emph{The theory of good-deal pricing in
  financial markets}, Mathematical Finance—Bachelier Congress 2000, Springer,
  2002, pp.~175--202.

\bibitem{cheridito2009risk}
P.~Cheridito and T.~Li, \emph{Risk measures on {O}rlicz hearts}, Math. Finance
  \textbf{19} (2009), no.~2, 189--214.

\bibitem{cherny2008pricing}
A.~S.\vspace{0.0mm} Cherny, \emph{Pricing with coherent risk}, Theory Probab.
  its Appl. \textbf{52} (2008), no.~3, 389--415.

\bibitem{cochrane2000beyond}
J.~H. Cochrane and J.~Saa-Requejo, \emph{Beyond arbitrage: Good-deal asset
  price bounds in incomplete markets}, J. Polit. Econ. \textbf{108} (2000),
  no.~1, 79--119.

\bibitem{de2005reward}
E.~De~Giorgi, \emph{Reward--risk portfolio selection and stochastic dominance},
  J. Bank. Finance \textbf{29} (2005), no.~4, 895--926.

\bibitem{delbaen2006mathematics}
F.~Delbaen and W.~Schachermayer, \emph{The mathematics of arbitrage}, Springer
  Science \& Business Media, 2006.

\bibitem{embrechts2021robustness}
P.~Embrechts, A.~Schied, and R.~Wang, \emph{Robustness in the optimization of
  risk measures}, Oper. Res. \textbf{70} (2021), no.~1, 95--110.

\bibitem{filipovic2007monotone}
D.~Filipovi{\'c} and M.~Kupper, \emph{Monotone and cash-invariant convex
  functions and hulls}, Insur. Math. Econ. \textbf{41} (2007), no.~1, 1--16.

\bibitem{follmer2002convex}
H.~F{\"o}llmer and A.~Schied, \emph{Convex measures of risk and trading
  constraints}, Finance Stoch. \textbf{6} (2002), no.~4, 429--447.

\bibitem{follmerschied:2016}
H.~F{\"o}llmer and A.~\vspace{0mm}Schied, \emph{Stochastic finance}, fourth
  ed., de Gruyter Studies in Mathematics, vol.~27, Walter de Gruyter \& co.,
  Berlin, 2016.

\bibitem{frittelli2002putting}
M.~Frittelli and E.~R. Gianin, \emph{Putting order in risk measures}, J. Bank.
  Finance \textbf{26} (2002), no.~7, 1473--1486.

\bibitem{gaivoronski2005value}
A.~Gaivoronski and G.~Pflug, \emph{Value-at-risk in portfolio optimization:
  properties and computational approach}, J. Risk \textbf{7} (2005), no.~2,
  1--31.

\bibitem{giesecke2008measuring}
K.~Giesecke, T.~Schmidt, and S.~Weber, \emph{Measuring the risk of large
  losses}, J. invest. manag. \textbf{6} (2008), no.~4, 1--15.

\bibitem{han:al:21}
X.~Han, B.~Wang, R.~Wang, and Q.~Wu, \emph{Risk concentration and the
  mean-expected shortfall criterion}, arXiv preprint arXiv:2108.05066, 2021.

\bibitem{herdegen2020dual}
M.~Herdegen and N.~Khan, \emph{Mean-$\rho$ portfolio selection and
  $\rho$-arbitrage for coherent risk measures}, Math. Finance \textbf{32}
  (2022), no.~1, 226--272.

\bibitem{HKM2024}
M.~Herdegen, N.~Khan, and C.~Munari, \emph{Risk, utility and sensitivity to
  large losses}, Available at SSRN 4739077 (2024).

\bibitem{hernandez1992discrete}
O.~Hern{\'a}ndez-Lerma and M.~Mu{\~n}oz~de Ozak, \emph{Discrete-time markov
  control processes with discounted unbounded costs: optimality criteria},
  Kybernetika \textbf{28} (1992), no.~3, 191--212.

\bibitem{jaschke2001coherent}
S.~Jaschke and U.~K{\"u}chler, \emph{Coherent risk measures and good-deal
  bounds}, Finance Stoch. \textbf{5} (2001), no.~2, 181--200.

\bibitem{susanne2007dynamic}
S.~Kl{\"o}ppel and M.~Schweizer, \emph{Dynamic utility-based good deal bounds},
  Stat. Dec. \textbf{25} (2007), no.~4, 285--309.

\bibitem{krokhmal2006modeling}
P.~A. Krokhmal and R.~Murphey, \emph{Modeling and implementation of risk-averse
  preferences in stochastic programs using risk measures}, Robust
  optimization-directed design, Springer, 2006, pp.~95--116.

\bibitem{mao2020risk}
T.~Mao and R.~Wang, \emph{Risk aversion in regulatory capital principles}, SIAM
  Journal on Financial Mathematics \textbf{11} (2020), no.~1, 169--200.

\bibitem{Markowitz1952}
H.~Markowitz, \emph{Portfolio selection}, J. Finance \textbf{7} (1952), no.~1,
  77--91.

\bibitem{mastrogiacomo2015portfolio}
E.~Mastrogiacomo and E.~R. Gianin, \emph{Portfolio optimization with
  quasiconvex risk measures}, Math. Oper. Res. \textbf{40} (2015), no.~4,
  1042--1059.

\bibitem{QRMbook}
A.~J. McNeil, R.~Frey, and P.~Embrechts, \emph{{Quantitative} {Risk}
  {Management}}, Princeton University Press, Princeton, New Jersey, 2005.

\bibitem{rockafellar1974conjugate}
R.~T. Rockafellar, \emph{Conjugate duality and optimization}, no.~16, Society
  for Industrial and Applied Mathematics, 1974.

\bibitem{rockafellar2002conditional}
R.~T. Rockafellar and S.~Uryasev, \emph{Conditional value-at-risk for general
  loss distributions}, J. Bank. Finance \textbf{26} (2002), no.~7, 1443--1471.

\bibitem{MasterFundsRockafellar}
R.~T. Rockafellar, S.~Uryasev, and M.~Zabarankin, \emph{Master funds in
  portfolio analysis with general deviation measures}, J. Bank. Finance
  \textbf{30} (2006), no.~2, 743--778.

\bibitem{ruszczynski2006optimization}
A.~Ruszczy{\'n}ski and A.~Shapiro, \emph{Optimization of convex risk
  functions}, Math. Oper. Res. \textbf{31} (2006), no.~3, 433--452.

\bibitem{staum2004fundamental}
J.~Staum, \emph{Fundamental theorems of asset pricing for good deal bounds},
  Math. Finance \textbf{14} (2004), no.~2, 141--161.

\bibitem{wang:16}
R.~Wang, \emph{Regulatory arbitrage of risk measures}, Quant. Finance
  \textbf{16} (2016), no.~3, 337--347.

\bibitem{zalinescu2002convex}
C.~Z\u{a}linescu, \emph{Convex analysis in general vector spaces}, World
  scientific, 2002.

\end{thebibliography}
\bibliographystyle{amsplain}
 
\end{document}